\documentclass[a4paper,10pt]{article}
\usepackage[utf8]{inputenc}
\usepackage[letterpaper, margin=1.15in]{geometry}

\usepackage{epigraph}

\usepackage{csquotes}
\usepackage{hyperref}       
\usepackage{url}            
\usepackage{booktabs}       
\usepackage{amsfonts}       
\usepackage{nicefrac}       
\usepackage{microtype}      
\usepackage{color}
\usepackage{scalerel}
\usepackage[toc,page]{appendix}
\usepackage{tabularx,ragged2e,booktabs,caption}
\usepackage[font=small,labelfont=bf]{caption}
\usepackage{blindtext}
\usepackage{amssymb}
\usepackage{mathtools}
\usepackage{dsfont}
\usepackage{amsmath}
\usepackage{amsfonts}
\usepackage{amsthm}
\usepackage{float}
\usepackage{graphicx}
\usepackage{subcaption}
\graphicspath{ {images/} }
\usepackage{letltxmacro}%
\usepackage{thmtools,thm-restate}
\usepackage{relsize}

\newcommand\givenbase[1][]{\:#1\lvert\:}
\let\given\givenbase

\DeclarePairedDelimiterX\Basics[1](){\let\given\sgiven #1}

\DeclarePairedDelimiter\abs{\lvert}{\rvert}
\DeclarePairedDelimiter\norm{\lVert}{\rVert}
\DeclarePairedDelimiterX{\infdivx}[2]{(}{)}{#1\;\delimsize\|\;#2}

\DeclareMathOperator*{\argmax}{arg\,max}
\DeclareMathOperator*{\argmin}{arg\,min}

\makeatletter
\newcommand{\distas}[1]{\mathbin{\overset{#1}{\kern\z@\sim}}}%
\newsavebox{\mybox}\newsavebox{\mysim}
\newcommand{\distras}[1]{%
	\savebox{\mybox}{\hbox{\kern3pt$\scriptstyle#1$\kern3pt}}%
	\savebox{\mysim}{\hbox{$\sim$}}%
	\mathbin{\overset{#1}{\kern\z@\resizebox{\wd\mybox}{\ht\mysim}{$\sim$}}}%
}
\makeatother

\usepackage{enumitem}
\newlist{inparaenum}{enumerate}{2}
\setlist[inparaenum]{nosep}
\setlist[inparaenum,1]{label=\bfseries\arabic*.}
\setlist[inparaenum,2]{label=\arabic{inparaenumi}\emph{\alph*})}

\usepackage{bm}
\usepackage{breqn}
\usepackage{physics}
\usepackage{enumitem}
\usepackage[colorinlistoftodos]{todonotes}

\theoremstyle{plain}
\newtheorem{thm}{Theorem}[section]
\newtheorem*{thm*}{Theorem}
\newtheorem{corollary}{Corollary}[section]
\newtheorem*{corollary*}{Corollary}
\newtheorem{lemma}{Lemma}[section]
\newtheorem*{lemma*}{Lemma}

\newtheorem*{prop*}{Proposition}
\theoremstyle{definition}

\theoremstyle{remark}
\newtheorem{remark}{Remark}[thm]

\newcommand{\bX}{\boldsymbol{X}}
\newcommand{\bY}{\boldsymbol{Y}}
\newcommand{\bM}{\boldsymbol{M}}
\newcommand{\bE}{\boldsymbol{E}}
\newcommand{\bQ}{\boldsymbol{Q}}
\newcommand{\bA}{\boldsymbol{A}}
\newcommand{\bB}{\boldsymbol{B}}
\newcommand{\bC}{\boldsymbol{C}}
\newcommand{\bP}{\boldsymbol{P}}
\newcommand{\bI}{\boldsymbol{I}}
\newcommand{\bU}{\boldsymbol{U}}
\newcommand{\bV}{\boldsymbol{V}}
\newcommand{\bD}{\boldsymbol{D}}

\newcommand{\bSigma}{\boldsymbol{\Sigma}}
\newcommand{\bLambda}{\boldsymbol{\Lambda}}
\newcommand{\bhB}{\hat{\bB}}
\newcommand{\bhM}{\hat{\bM}}
\newcommand{\bhQ}{\hat{\bQ}}

\newcommand{\MSE}{\emph{MSE}}
\newcommand{\bbX}{\bar{\bX}}
\newcommand{\bbY}{\bar{\bY}}
\newcommand{\bbM}{\bar{\bM}}

\newcommand{\E}{\mathbb{E}}
\newcommand{\Pb}{\mathbb{P}}

\usepackage{algorithm}
\usepackage[noend]{algpseudocode}

\makeatletter
\def\BState{\State\hskip-\ALG@thistlm}
\makeatother

\usepackage{parskip}
\title{Robust Synthetic Control\footnote{We would like to thank Alberto Abadie for careful reading and comments that have helped in improving the manuscript.}}
\author{Muhammad Jehangir Amjad, Devavrat Shah, and Dennis Shen \\ \\
Laboratory for Information and Decision Systems, \\
Statistics and Data Science Center, \\
Massachusetts Institute of Technology \\
\{mamjad, devavrat, deshen\}@mit.edu}

\begin{document} 
\date{}

\maketitle

\section*{\hfil \small{Abstract} \hfil}

\begin{changemargin}{1cm}{1cm} 

We present a robust generalization of the synthetic control method for comparative case studies. Like the classical method cf. \cite{abadie1, abadie3, abadie2}, we present an algorithm to estimate the unobservable counterfactual of a treatment unit. A distinguishing feature of our algorithm is that of de-noising the data matrix via singular value thresholding, which renders our approach robust in multiple facets: it automatically identifies a good subset of donors for the synthetic control, overcomes the challenges of missing data, and continues to work well in settings where covariate information may not be provided. To begin with, we establish the condition under which the fundamental assumption in synthetic control-like approaches holds, i.e. when the linear relationship between the treatment unit and the donor pool prevails in both the pre- and post-intervention periods. We provide the first finite sample analysis (coupled with asymptotic results) for a broader class of models, the Latent Variable Model (LVM), in contrast to Factor Models previously considered in the literature, while also relating the interpolation and extrapolation abilities of our estimator to the amount of data available. In particular, we show that our de-noising procedure accurately imputes missing entries and filters corrupted observations in producing a consistent estimator of the underlying signal matrix, provided $p = \Omega( T^{-1 + \zeta})$ for some $\zeta > 0$; here, $p$ is the fraction of observed data and $T$ is the time interval of interest. Under the same proportion of observations, we demonstrate that the mean-squared-error 
in our prediction estimation scales as $\mathcal{O}(\sigma^2/p +  1/\sqrt{T})$, where $\sigma^2$ is the variance of the inherent noise. Using a ``data aggregation'' method, we show that the mean-square-error can be made as small as {$\mathcal{O}(T^{-1/2+\gamma})$} for any $\gamma \in (0, 1/2)$, and thus leading to a consistent estimator. In order to move beyond point estimates, we introduce a Bayesian framework that not only provides the ability to readily develop different estimators under various loss functions, but also quantifies the uncertainty of the model/estimates through posterior probabilities. Our experiments, using both synthetic and real-world datasets, demonstrate that our robust generalization yields an improvement over 
the classical synthetic control method, underscoring the value of our key de-noising procedure. 

\end{changemargin}

\section{Introduction} \label{sec:intro}
On November 8, 2016 in the aftermath of several high profile mass-shootings, voters in California passed Proposition 63 in to law \cite{prop63}. Prop. 63 ``outlaw[ed] the possession of ammunition magazines that [held] more than 10 rounds, requir[ed] background checks for people buying bullets,'' and was proclaimed as an initiative for ``historic progress to reduce gun violence'' \cite{prop63latimes}. Imagine that we wanted to study the impact of Prop. 63 on the rates of violent crime in California. Randomized control trials, such as A/B testings, have been successful in establishing effects of interventions by randomly exposing segments of the population to various types of interventions. Unfortunately, a randomized control trial is not applicable in this scenario since only one California exists. Instead, a statistical comparative study could be conducted where the rates of violent crime in California are compared to a ``control'' state after November 2016, which we refer to as the post-intervention period. To reach a statistically valid conclusion, however, the control state must be demonstrably similar to California sans the passage of a Prop. 63 style legislation. In general, there may not exist a natural control state for California, and subject-matter experts tend to disagree on the most appropriate state for comparison.
 
As a suggested remedy to overcome the limitations of a classical comparative study outlined above, Abadie et al. proposed a powerful, data-driven approach to construct a ``synthetic'' control unit absent of intervention \cite{abadie1, abadie3, abadie2}. In the example above, the synthetic control method would construct a ``synthetic'' state of California such that the rates of violent crime of that hypothetical state would best match the rates in California before the passage of Prop. 63. This synthetic California can then serve as a data-driven counterfactual for the period after the passage of Prop. 63. Abadie et al. propose to construct such a synthetic California by choosing a convex combination of other states (donors) in the United States. For instance, synthetic California might be 80\% like New York and 20\% like Massachusetts. This approach is nearly entirely data-driven and appeals to intuition. For optimal results, however, the method still relies on subjective covariate information, such as employment rates, and the presence of domain ``experts'' to help identify a useful subset of donors. The approach may also perform poorly in the presence of non-negligible levels of noise and missing data. 


\subsection{Overview of main contributions.}
As the main result, we propose a simple, two-step robust synthetic control algorithm, wherein the first step de-noises the data and the second step learns a linear relationship between the treated unit and the donor pool under the de-noised setting. The algorithm is robust in two senses: first, it is fully data-driven in that it is able to find a good donor subset even in the absence of helpful domain knowledge or supplementary covariate information; and second, it provides the means to overcome the challenges presented by missing and/or noisy observations. As another important contribution, we establish analytic guarantees (finite sample analysis and asymptotic consistency) -- that are missing from the literature -- for a broader class of models.

{\em \textbf{Robust algorithm.}}
A distinguishing feature of our work is that of de-noising the observation data via singular value thresholding. Although this spectral procedure is commonplace in the matrix completion arena, it is novel in the realm of synthetic control. Despite its simplicity, however, thresholding brings a myriad of benefits and resolves points of concern that have not been previously addressed. For instance, while classical methods have not even tackled the obstacle of missing data, our approach is well equipped to impute missing values as a consequence of the thresholding procedure. Additionally, thresholding can help prevent the model from overfitting to the idiosyncrasies of the data, providing a knob for practitioners to tune the ``bias-variance'' trade-off of their model and, thus, reduce their mean square error (MSE). From empirical studies, we hypothesize that thresholding may possibly render auxiliary covariate information (vital to several existing methods) a luxury as opposed to a necessity. However, as one would expect, the algorithm can only benefit from useful covariate and/or ``expert'' information and we do not advocate ignoring such helpful information, if available.

In the spirit of combatting overfitting, we extend our algorithm to include regularization techniques such as ridge regression and LASSO. We also move beyond point estimates in establishing a Bayesian framework, which allows one to quantitatively compute the uncertainty of the results through posterior probabilities.


{\em \textbf{Theoretical performance.}}
To the best of our knowledge, ours is the first to provide finite sample analysis of the MSE for the synthetic control method, in addition to guarantees in the presence of missing data. 
Previously, the main theoretical result from the synthetic control literature (cf. \cite{abadie1, abadie3, abadie2}) pertained to bounding the bias of the synthetic control 
estimator; however, the proof of the result assumed that the latent parameters, which live in the simplex, have a perfect pre-treatment match in the noisy predictor variables 
-- our analysis, on the other hand, removes this assumption. We begin by demonstrating that our de-noising procedure produces a consistent estimator of the latent signal matrix (Theorems \ref{thm:imputation_lowrank}, \ref{thm:imputation_lipschitz}), proving that our thresholding method accurately imputes and filters missing and noisy observations, respectively. We then provide finite sample analysis that not only highlights the value of thresholding in balancing 
the inherent ``bias-variance'' trade-off of forecasting, but also proves that the prediction efficacy of our algorithm degrades gracefully with an increasing number of randomly missing data 
(Theorems \ref{thm:finite-sample}, \ref{thm:post_rmse}, and Corollary \ref{corollary:universal_threshold}). Further, we show that a computationally beneficial pre-processing data aggregation step allows us to establish the asymptotic consistency of our estimator in generality (Theorem \ref{thm:consistency}). 



Additionally, we prove a simple linear algebraic fact that justifies the basic premise of synthetic control, which has not been formally established in literature, i.e. the linear relationship 
between the treatment and donor units that exists in the pre-intervention continues to hold in post-intervention period (Theorem \ref{thm:post}). We introduce a latent variable model, 
which subsumes many of the models previously used in literature (e.g. econometric factor models). Despite this generality, a unifying theme that connects these 
models is that they all induce (approximately) low rank matrices, which is well suited for our method.


{\em \textbf{Experimental results.}} We conduct two sets of experiments: (a) on existing case studies from real world datasets referenced in \cite{abadie1, abadie2, abadie3}, and (b) on synthetically generated data. Remarkably, while \cite{abadie1, abadie2, abadie3} use numerous covariates and employ expert knowledge in selecting their donor pool, our algorithm achieves similar results without any such assistance; additionally, our algorithm detects subtle effects of the intervention that were overlooked by the original synthetic control approach. { Since it is impossible to simultaneously observe the evolution of a treated unit and its counterfactual}, we employ synthetic data to validate the efficacy of our method. Using the MSE as our evaluation metric, we demonstrate that our algorithm is robust to varying levels of noise and missing data, { reinforcing the importance of de-noising}.  



\subsection{Related work.} Synthetic control has received widespread attention since its conception by Abadie and Gardeazabal in their pioneering work \cite{abadie3, abadie1}. It has been employed in numerous case studies, ranging from criminology \cite{crime} to health policy \cite{health} to online advertisement to retail; other notable studies include \cite{abadie4, economic_liberalization, tax, turkey}. In their paper on the state of applied econometrics for causality and policy evaluation, Athey and Imbens assert that synthetic control is ``one of the most important development[s] in program evaluation in the past decade'' and ``arguably the most important innovation in the evaluation literature in the last fifteen years'' \cite{athey}. In a somewhat different direction, Hsiao et al. introduce the panel data method \cite{hsiao1, hsiao2}, which seems to have a close bearing with some of the approaches of this work. In particular, \cite{hsiao1, hsiao2} only uses data for the outcome variable and solves an ordinary least squares problem in learning synthetic control. However, \cite{hsiao1, hsiao2} restrict the subset of possible controls to units that are within the geographical or economic proximity of the treated unit. Therefore, there is still some degree of subjectivity in the choice of the donor pool. In addition, \cite{hsiao1, hsiao2} do not include a ``de-noising'' step, which is a key feature of our approach. For an empirical comparison between the synthetic control and panel data methods, see \cite{comparison}. It should be noted that \cite{comparison} also adapts the panel data method to automate the donor selection process. \cite{doudchenko} allows for an additive difference between the treated unit and donor pool, similar to the difference-in-differences (DID) method. Moreover, similar to our exposition, \cite{doudchenko} relaxes the convexity aspect of synthetic control and proposes an algorithm that allows for unrestricted linearity as well as regularization. In an effort to infer the causal impact of market interventions, \cite{bayes_sc} introduce yet another evaluation methodology based on a diffusion-regression state-space model that is fully Bayesian; similar to \cite{abadie1, abadie3, hsiao1, hsiao2}, their model also generalizes the DID procedure. Due to the subjectivity in the choice of covariates and predictor variables, \cite{ferman1} provides recommendations for specification-searching opportunities in synthetic control applications. The recent work of \cite{xu} extends the synthetic control method to allow for multiple treated units and variable treatment periods as well as the treatment being correlated with unobserved units. Similar to our work, \cite{xu} computes uncertainty estimates; however, while \cite{xu} obtains these measurements via a parametric bootstrap procedure, we obtain uncertainty estimates under a Bayesian framework.

Matrix completion and factorization approaches are well-studied problems with broad applications (e.g. recommendation systems, graphon estimation, etc.). As shown profusely in the literature, spectral methods, such as singular value decomposition and thresholding, provide a procedure to estimate the entries of a matrix from partial and/or noisy observations \cite{recht1}. With our eyes set on achieving ``robustness'', spectral methods become particularly appealing since they de-noise random effects and impute missing information within the data matrix \cite{svdjha1}. For a detailed discussion on the topic, see \cite{usvt}; for algorithmic implementations, see \cite{mazumder2010} and references there in. We note that our goal differs from traditional matrix completion applications in that we are using spectral methods to estimate a low-rank matrix, allowing us to determine a linear relationship between the rows of the mean matrix. This relationship is then projected into the future to determine the counterfactual evolution of a row in the matrix (treated unit), which is traditionally not the goal in matrix completion applications. Another line of work within this arena is to impute the missing entries via a nearest neighbor based estimation algorithm under a latent variable model framework \cite{lee_song, lee2}. 

There has been some recent work in using matrix norm methods in relation to causal inference, including for synthetic control. In \cite{athey1}, the authors use matrix norm regularization techniques to estimate counterfactuals for panel data under settings that rely on the availability of a large number of units relative to the number of factors or characteristics, and under settings that involve limited number of units but plenty of history (synthetic control). This is different from our approach, which increases robustness by ``de-noising'' using matrix completion methods, and then using linear regression on the de-noised matrix, instead of relying on matrix norm regularizations.

Despite its popularity, there has been less theoretical work in establishing the consistency of the synthetic control method or its variants. \cite{abadie1} demonstrates that the bias of the synthetic control estimator can be bounded by a function that is close to zero when the pre-intervention period is large in relation to the scale of the transitory shocks, but under the additional condition that a perfect convex match between the pre-treatment noisy outcome and covariate variables for the treated unit and donor pool exists. \cite{ferman2} relaxes the assumption in \cite{abadie1}, and derives conditions under which the synthetic control estimator is asymptotically unbiased under non-stationarity conditions. To our knowledge, however, no prior work has provided finite-sample analysis, analyzed the performance of these estimators with respect to the mean-squared error (MSE), established asymptotic consistency, or addressed the possibility of missing data, {a common handicap in practice}. 

\section{Background} \label{sec:model}

\subsection{Notation.} We will denote $\mathbb{R}$ as the field of real numbers. For any positive integer $N$, let $[N] = \{1, \dots, N\}$. For any vector $v \in \mathbb{R}^n$, we denote its Euclidean ($\ell_2$) norm by $\norm{v}_2$, and define $\norm{v}_2^2 = \sum_{i=1}^n v_i^2$. We define its infinity norm as $\norm{v}_{\infty} = \max_{i} \abs{v_i}$. In general, the $\ell_p$ norm for a vector $v$ is defined as $\norm{v}_p = \Big( \sum_{i=1}^n \abs{v_i}^p \Big)^{1/p}$. Similarly, for an $m \times n$ real-valued matrix $\bA = [A_{ij}]$, its spectral/operator norm, denoted by $\norm{\bA}_2$, is defined as $\norm{\bA}_2 = \max_{1 \le i \le k} \abs{\sigma_i}$, where $k = \min\{m,n\}$ and $\sigma_i$ are the singular values of $\bA$. The Moore-Penrose pseudoinverse $\bA^{\dagger}$ of $\bA$ is defined as 
\begin{align}
\bA^{\dagger} & = \sum_{i=1}^k (1/ \sigma_i) y_i x_i^T, \quad \text{where} \quad \bA = \sum_{i=1}^k \sigma_i x_i y_i^T,
\end{align} 
with $x_i$ and $y_i$ being the left and right singular vectors of $\bA$, respectively. We will adopt the shorthand notation of $\norm{\cdot} \equiv \norm{\cdot}_2$. 
To avoid any confusions between scalars/vectors and matrices, we will represent all matrices in bold, e.g. $\bA$.

Let $f$ and $g$ be two functions defined on the same space. We say that $f(x) = \mathcal{O}(g(x))$ and $f(x) = \Omega(g(x))$ if and only if there exists a positive real number $M$ and a real number $x_0$ such that for all $x \ge x_0$,
\begin{align}
	\abs{f(x)} &\le M \abs{g(x)} \quad \text{and} \quad \abs{f(x)} \ge M \abs{g(x)},
\end{align}
respectively. 




\subsection{Model.} 
The data at hand is a collection of time series with respect to an aggregated metric of interest (e.g. violent crime rates) comprised of both the treated unit and the donor pool outcomes. Suppose we observe $N \geq 2$ units across $T \geq 2$ time periods. We denote $T_0$ as the number of pre-intervention periods with $1 \leq T_0 < T$, rendering $T - T_0$ as the length of the post-intervention stage. Without loss of generality, let the first unit represent the treatment unit -- exposed to the intervention of interest at time $t = T_0 + 1$. The remaining donor units, $2\leq i\leq N$, are unaffected by the intervention for the entire time period $[T]=\{1,\dots, T\}$.

Let $X_{it}$ denote the measured value of metric for unit $i$ at time $t$. We posit
\begin{align} \label{eq:1}
	X_{it} &= M_{it} + \epsilon_{it},
\end{align}
where $M_{it}$ is the deterministic mean while the random variables $\epsilon_{it}$ represent zero-mean noise that are independent across $i,t$. Following the philosophy of latent variable models \cite{usvt, lee_song, aldous, hoover1, hoover2}, we further posit that for all $2 \le i \le N$, $t \in [T]$
\begin{align} \label{eq:2}
	M_{it} &= f(\theta_i, \rho_t),
\end{align}
where $\theta_i \in {\mathbb R}^{d_1}$ and $\rho_t \in {\mathbb R}^{d_2}$ are latent feature vectors capturing unit and time specific information, respectively, for some $d_1, d_2 \geq 1$; the latent function $f: {\mathbb R}^{d_1} \times {\mathbb R}^{d_2} \to {\mathbb R}$ captures the model relationship. We note that this formulation subsumes popular econometric factor models, such as the one presented in \cite{abadie1}, as a special case with (small) constants $d_1 = d_2$ and $f$ as a bilinear function. 

The treatment unit obeys the same model relationship during the pre-intervention period. That is, for $t \leq T_0$ 
\begin{align} \label{eq:1a}
	X_{1t} & = M_{1t} + \epsilon_{1t}, 
\end{align} 
where $M_{1t} = f(\theta_1, \rho_t)$ for some latent parameter $\theta_1 \in {\mathbb R}^{d_1}$. If unit one was never exposed to the intervention, then the same relationship as \eqref{eq:1a} would continue to hold during the post-intervention period as well. In essence, we are assuming that the outcome random variables for {\em all} unaffected units follow the model relationship defined by \eqref{eq:1a} and \eqref{eq:1}. Therefore, the ``synthetic control'' would ideally help estimate the underlying counterfactual means $M_{1t} = f(\theta_1, \rho_t)$ for $T_0 < t \le T$ by using an appropriate combination of the post-intervention observations from the donor pool since the donor units are immune to the treatment.

To render this feasible, we make the key operating assumption (as done in literature cf. \cite{abadie1, abadie2, abadie3}) that the mean vector of the treatment unit over the pre-intervention period, i.e. the vector $M_{1}^- = [M_{1t}]_{t\leq T_0}$, lies within the span of the mean vectors within the donor pool over the pre-intervention period, i.e. the span of the donor mean vectors $M_{i}^- = [M_{it}]_{2\le i \le N, t \le T_0}$ \footnote{We note that this is a minor departure from the literature on synthetic control starting in \cite{abadie3} -- in literature, the pre-intervention {\em noisy} observation (rather than the mean) vector $X_{1}$, is assumed to be a {\em convex} (rather than linear) combination of the noisy donor observations. We believe our setup is more reasonable since we do not want to fit noise.}. More precisely, we assume there exists a set of weights $\beta^* \in {\mathbb R}^{N-1}$ such that for all $t\leq T_0$,
\begin{equation}\label{eq:3}
	M_{1t} = \sum_{i=2}^N \beta_i^* M_{it}.
\end{equation} 
This is a reasonable and intuitive assumption, utilized in literature, hypothesizing that the treatment unit can be modeled as some combination of the donor pool. In fact, the set of weights $\beta^*$ are the very definition of a synthetic control. 

In order to distinguish the pre- and post-intervention periods, we use the following notation for all (donor) matrices: $\bA = [\bA^{-}, \bA^{+}]$, where $\bA^{-} = [A_{ij}]_{2 \le i \le N, j \in [T_0]}$ and $\bA^{+} = [A_{ij}]_{2 \le i \le N, T_0 < j \le T}$ denote the pre- and post-intervention submatrices, respectively; vectors will be defined in the same manner, i.e. $A_i = [A_i^{-}, A_i^{+}]$, where $A_i^{-} = [A_{it}]_{t \in [T_0]}$ and $A_i^{+} = [A_{it}]_{T_0 < t \le T}$ denote the pre- and post-intervention subvectors, respectively, for the $i$th donor. Moreover, we will denote all vectors related to the treatment unit with the subscript ``1'', e.g. $A_1 = [A_1^{-}, A_1^{+}]$.

In contrast with the classical synthetic control work, we allow our model to be robust to incomplete observations. To model randomly missing data, the algorithm observes each data point $X_{it}$ in the donor pool with probability $p \in (0, 1]$, independently of all other entries. While the assumption that $p$ is constant across all rows and columns of our observation matrix is standard in literature, our results remain valid even in situations where the probability of observation is dependent on the row and column latent parameters, i.e. $p_{ij} = g(\theta_i, \rho_j) \in (0, 1]$. In such situations, $p_{ij}$ can be estimated as $\hat{p}_{ij}$ using consistent graphon estimation techniques described in a growing body of literature, e.g. see \cite{lee2, usvt, wolfe1,yang1}. These estimates can then be used in our analysis presented in Section \ref{sec:results}.

\section{Algorithm} \label{sec:algorithm}

\subsection{Intuition.}
We begin by exploring the intuition behind our proposed two-step algorithm: (1) \textit{de-noising the data:} since the singular values of our observation matrix, $\bX = [X_{it}]_{2 \le i \le N, t \in [T]}$, encode both signal and noise, we aim to discover a low rank approximation of $\bX$ that only incorporates the singular values associated with useful information; simultaneously, this procedure will naturally impute any missing observations. We note that this procedure is similar to the algorithm proposed in \cite{usvt}. (2) \textit{learning $\beta^*$:} using the pre-intervention portion of the de-noised matrix, we learn the linear relationship between the treatment unit and the donor pool prior to estimating the post-intervention counterfactual outcomes. Since our objective is to produce accurate predictions, it is not obvious why the synthetic treatment unit should be a convex combination of its donor pool as assumed in \cite{abadie1, abadie3, abadie4}. In fact, one can reasonably expect that the treatment unit and some of the donor units may exhibit negative correlations with one another. In light of this intuition, we learn the optimal set of weights via linear regression, allowing for both positive and negative elements. 


\subsection{Robust algorithm (algorithm \ref{euclid}).} \label{sec:robust_algo}
We present the details of our robust method in Algorithm \ref{euclid} below. The algorithm utilizes two hyperparameters: (1) a thresholding hyperparameter $\mu \geq 0$, which serves as a knob to effectively trade-off between the bias and variance of the estimator, and (2) a regularization hyperameter $\eta \ge 0$ that controls the model complexity. We discuss the procedure for determining the hyperparameters in Section \ref{sec:fineprint}. To simplify the exposition, we assume the entries of $\bX$ are bounded by one in absolute value, i.e. $ \abs{X_{it} } \leq 1$.


\begin{algorithm} [H]
\caption{Robust synthetic control}\label{euclid}
\noindent \\

{\bf Step 1. De-noising the data: singular value thresholding (inspired by \cite{usvt}).} \\
\begin{enumerate}
	\item Define $\bY = [Y_{it}]_{2 \le i \le N, t \in [T]}$ with
	\begin{align}
		Y_{it} = \begin{cases}
		X_{it} & \text{if}\ X_{it} \text{ is observed}, \\
		0 & \text{otherwise.}
		\end{cases}
	\end{align}

	\item Compute the singular value decomposition of $\bY$:
	\begin{align}
		\bY  = \sum_{i=1}^{N-1} s_i u_i v_i^T.
	\end{align}

	\item Let $S = \{ i : s_i \geq \mu\}$ be the set of singular values above the threshold $\mu$.

	\item Define the estimator of $\bM$ as 
	\begin{align} \label{eq:p_hat}
		\bhM & = \dfrac{1}{\hat{p}} \sum_{i \in S} s_i u_i v_i^T,
	\end{align}
	where $\hat{p}$ is the maximum of the fraction of observed entries in $\bX$ and $\frac{1}{(N-1)T}$.	
\end{enumerate} 

\noindent \\
{\bf Step 2. Learning and projecting} \\

\begin{enumerate}
	\item For any $\eta \ge 0$, let
	\begin{align} 
		\hat{\beta}(\eta) &= \argmin_{v \in \mathbb{R}^{N-1}} \norm{Y_1^{-} - (\bhM^{-})^T v}^2 + \eta \norm{v}^2 \label{eq:ls}.
	\end{align}
	
	\item Define the counterfactual means for the treatment unit as 
	\begin{align}
		\hat{M}_1 & = \bhM^T \hat{\beta}(\eta).
	\end{align}
\end{enumerate} 

%
%
%
%
%
\end{algorithm}

\subsection{Bayesian algorithm: measuring uncertainty (algorithm \ref{euclid:bayesian}).}
In order to quantitatively assess the uncertainty of our model, we will transition from a frequentist perspective to a Bayesian viewpoint. As commonly assumed in literature, we consider a zero-mean, isotropic Gaussian noise model (i.e. $\epsilon \distas{} \mathcal{N}(0, \sigma^2 \bI)$) and use the square loss for our cost function. We present the Bayesian method as Algorithm \ref{euclid:bayesian}. Note that we perform step one of our robust algorithm exactly as in Algorithm \ref{euclid}; as a result, we only detail the alterations of step two in the Bayesian version (Algorithm \ref{euclid:bayesian}).

\begin{algorithm}[H]
\caption{Bayesian robust synthetic control}\label{euclid:bayesian}
	\noindent \\
	
	{\bf Step 2. Learning and projecting} \\
	\begin{enumerate}
		\item Estimate the noise variance via (bias-corrected) maximum likelihood, i.e.
		\begin{align} \label{eq:sample_variance}
			\hat{\sigma}^2 &= \dfrac{1}{T_0 - 1} \sum_{t=1}^{T_0} (Y_{1t} - \bar{Y})^2, 
		\end{align}
		where $\bar{Y}$ denotes the pre-intervention sample mean.
		
		\item Compute posterior distribution parameters for an appropriate choice of the prior $\alpha$:
		\begin{align}
			\bSigma_D &= \Big(\dfrac{1}{\hat{\sigma}^2} \bhM^{-} (\bhM^{-})^T +  \alpha \bI \Big)^{-1} \label{eq:post_cov}
			\\ \beta_D &= \dfrac{1}{\hat{\sigma}^2} \bSigma_D\bhM^{-} Y_1^{-} \label{eq:post_mean}.
		\end{align} 
		
		\item Define the counterfactual means for the treatment unit as 
		\begin{align}
			\hat{M}_1 & = \bhM^T \beta_D.
		\end{align}
		
		\item For each time instance $t \in [T]$, compute the model uncertainty (variance) as
		\begin{align}
			\sigma^2_D(\hat{M}_{\cdot, t}) &= \hat{\sigma}^2 + \hat{M}_{\cdot, t}^T \bSigma_D \hat{M}_{\cdot, t},
		\end{align}
		where $\hat{M}_{\cdot, t} = [\hat{M}_{it}]_{2 \le i \le N}$ is the de-noised vector of donor outcomes at time $t$. 
	\end{enumerate}
	\noindent \\
%
	
\end{algorithm}

\subsection{Algorithmic features: the fine print.} \label{sec:fineprint}

\subsubsection{Bounded entries transformation.}
Several of our results, as well as the algorithm we propose, assume that the observation matrix is bounded such that $\abs{X_{it}} \leq 1$. For any data matrix, we can achieve this by using the following pre-processing transformation: suppose the entries of $\bX$ belong to an interval $[a,b]$. Then, one can first pre-process the matrix $\bX$ by subtracting $(a+b)/2$ from each entry, and dividing by $(b-a)/2$ to enforce that the entries lie in the range $[-1, 1]$. The reverse transformation, which can be applied at the end of the algorithm description above, returns a matrix with values contained in the original range. Specifically, the reverse transformation equates to multiplying the end result by $(b-a)/2$ and adding by $(a + b)/2$. 

\subsubsection{Solution interpretability.}
For the practitioner who seeks a more interpretable solution, e.g. a convex combination of donors as per the original synthetic control estimator of Abadie et. al, we recommend using an $\ell_1$-regularization penalty in the learning procedure of step 2. Due to the geometry of LASSO, the resulting estimator will be often be a sparse vector. Specifically, for any $\eta > 0$, we define the LASSO estimator to be
\[ \hat{\beta}(\eta) = \argmin_{v \in \mathbb{R}^{N-1}} \norm{Y_1^{-} - (\bhM^{-})^T v}^2 + \eta \norm{v}_1. \]

\subsubsection{Choosing the hyperparameters.} Here, we discuss several approaches to choosing the hyperparameter $\mu$ for the singular values. If it is known a priori that the underlying model is low rank with rank at most $k$, then it may make sense to choose $\mu$ such that $\abs{S} = k$. A data driven approach, however, could be implemented based on cross-validation. Precisely, reserve a portion of the pre-intervention period for validation, and use the rest of the pre-intervention data to produce an estimate $\hat{\beta}(\eta)$ for each of the finitely many choices of $\mu$ ($s_1, \dots, s_{N-1}$). Using each estimate $\hat{\beta}(\eta)$, produce its corresponding treatment unit mean vector over the validation period. Then, select the $\mu$ that achieves the minimum MSE with respect to the observed data. Finally, \cite{usvt} provides a universal approach to picking a threshold; similarly, we also propose another such universal threshold, \eqref{eq:goldilocks}, in Section \ref{sec:imputation}. We utilize the data driven approach in our experiments in this work.

The regularization parameter, $\eta$, also plays a crucial role in learning the synthetic control and influences both the training and generalization errors. As is often the case in model selection, a popular strategy in estimating the ideal $\eta$ is to employ cross-validation as described above. However, since time-series data often have a natural temporal ordering with causal effects, we also recommend employing the forward chaining strategy. Although the forward chaining strategy is similar to leave-one-out (LOO) cross-validation, an important distinction is that forward chaining does not break the temporal ordering in the training data. More specifically, for a particular candidate of $\eta$ at every iteration $t \in [T_0]$, the learning process uses $[Y_{11}, \dots, Y_{1,t-1}]$ as the training portion while reserving $Y_{1t}$ as the validation point. As before, the average error is then computed and used to evaluate the model (characterized by the choice of $\eta$). The forward chaining strategy can also be used to learn the optimal $\mu$.

\subsubsection{Scalability.}
In terms of scalability, the most computationally demanding procedure is that of evaluating the singular value decomposition (SVD) of the observation matrix. Given the ubiquity of SVD methods in the realm of machine learning, there are well-known techniques that enable computational and storage scaling for SVD algorithms. For instance, both Spark (through alternative least squares) and Tensor-Flow come with built-in SVD implementations. As a result, by utilizing the appropriate computational infrastructure, our de-noising procedure, and algorithm in generality, can scale quite well. Also note that for a low rank structure, we typically only need to compute the top few singular values and vectors. Various truncated-SVD algorithms provide resource-efficient implementations to compute the top $k$ singular values and vectors instead of the complete-SVD.

\subsubsection{Low rank hypothesis.}
The factor models that are commonly used in the Econometrics literature, cf. \cite{abadie1, abadie2, abadie3}, often lead to a low rank structure for the underlying mean matrix $\bM$. When $f$ is nonlinear, $\bM$ can still be well approximated by a low rank matrix for a large class of functions. For instance, if the latent parameters assumed values from a bounded, compact set, and if $f$ was Lipschitz continuous, then it can be argued that $\bM$ is well approximated by a low rank matrix, cf. see \cite{usvt} for a very simple proof. As the reader will notice, while we establish results for low rank matrix, the results of this work are robust to low rank approximations whereby the approximation error can be viewed as ``noise''. Lastly, as shown in \cite{log_rank}, many latent variable models can be well approximated (up to arbitrary accuracy $\epsilon$) by low rank matrices. Specifically, \cite{log_rank} shows that the corresponding low rank approximation matrices associated with ``nice'' functions (e.g. linear functions, polynomials, kernels, etc.) are of log-rank. 

\subsubsection{Covariate information.}
Although the algorithm does not appear to rely on any helpful covariate information and the experimental results, presented in Section \ref{sec:experiments}, suggest that it performs on par with that of the original synthetic control algorithm, we want to emphasize that we are not suggesting that practitioners should abandon the use of any additional covariate information or the application of domain knowledge. Rather, we believe that our key algorithmic feature -- the de-noising step -- may render covariates and domain expertise as luxuries as opposed to necessities for many practical applications. If the practitioner has access to supplementary predictor variables, we propose that step one of our algorithm be used as a pre-processing routine for de-noising the data before incorporating additional information. Moreover, other than the obvious benefit of narrowing the donor pool, domain expertise can also come in handy in various settings, such as determining the appropriate method for imputing the missing entries in the data. For instance, if it is known a priori that there is a trend or periodicity in the time series evolution for the units, it may behoove the practitioner to impute the missing entries using ``nearest-neighbors'' or linear interpolation.

\section{Theoretical Results} \label{sec:results}

In this section, we derive the finite sample and asymptotic properties of the estimators $\bhM$ and $\hat{M}_1$.  We begin by defining necessary notations and recalling a few operating assumptions prior to presenting the results, with the corresponding proofs relegated to the Appendix. To that end, we re-write \eqref{eq:1} in matrix form as $\bX = \bM + \bE$, where $\bE = [\epsilon_{it}]_{2\leq i\leq N, t \in [T]}$ denotes the noise matrix. We shall assume that the noise parameters $\epsilon_{it}$ are independent zero-mean random variables with bounded second moments. Specifically, for all $2\leq i\leq N, t \in [T]$, 
\begin{align}\label{eq:4}
	\mathbb{E}[ \epsilon_{it} ] &= 0, \quad \text{and} \quad \text{Var}(\epsilon_{it}) \le \sigma^2.
\end{align}
We shall also assume that the treatment unit noise in \eqref{eq:1a} obeys \eqref{eq:4}. Further, we assume the relationship in \eqref{eq:3} holds. To simplify the following exposition, we assume that $\abs{M_{ij}} \le 1$ and $\abs{X_{ij}} \le 1$.

As previously discussed, we evaluate the accuracy of our estimated means for the treatment unit with respect to the deviation between $\hat{M}_1$ and $M_1$ measured in $\ell_2$-norm, and similarly between $\bhM$ and $\bM$. Additionally, we aim to establish the validity of our pre-intervention linear model assumption (cf. \eqref{eq:3}) and investigate how the linear relationship translates over to the post-intervention regime, i.e. if $M_1^{-} = (\bM^{-})^T \beta^*$ for some $\beta^*$, does $M^{+}_1$ (approximately) equal to $(\bM^{+})^T \beta^*$? If so, under what conditions? We present our results for the above aspects after a brief motivation of $\ell_2$ regularization. 

\noindent {\bf Combatting overfitting.} One weapon to combat overfitting is to constrain the learning algorithm to limit the effective model complexity by fitting the data under a simpler hypothesis. This technique is known as regularization, and it has been widely used in practice. To employ regularization, we introduce a complexity penalty term into the objective function \eqref{eq:ls}. For a general regularizer, the objective function takes the form
\begin{align}
	\hat{\beta}(\eta) &= \argmin_{v \in \mathbb{R}^{N-1}} \norm{Y_1^{-} - (\bhM^{-})^Tv}^2 + \eta \sum_{j = 1}^{N-1} \abs{v_j}^q, \label{eq:obj_reg}
\end{align}
for some choice of positive constants $\eta$ and $q$. The first term measures the empirical error of the model on the given dataset, while the second term penalizes models that are too ``complex'' by controlling the ``smoothness'' of the model in order to avoid overfitting. In general, the impact/trade-off of regularization can be controlled by the value of the regularization parameter $\eta$. Via the use of Lagrange multipliers, we note that minimizing \eqref{eq:obj_reg} is equivalent to minimizing \eqref{eq:ls} subject to the constraint that 
\begin{align*}
	\sum_{j = 1}^{N-1} \abs{v_j}^q \le c,
\end{align*}
for some appropriate value of $c$. When $q = 2$, \eqref{eq:obj_reg} corresponds to the classical setup known as \textit{ridge regression} or \textit{weight decay}. The case of $q =1$ is known as the LASSO in the statistics literature; the $\ell_1$-norm regularization of LASSO is a popular heuristic for finding a sparse solution. In either case, incorporating an additional regularization term encourages the learning algorithm to output a simpler model with respect to some measure of complexity, which helps the algorithm avoid overfitting to the idiosyncrasies within the observed dataset. Although the training error may suffer from the simpler model, empirical studies have demonstrated that the generalization error can be greatly improved under this new setting. Throughout this section, we will primarily focus our attention on the case of $q = 2$, which maintains our learning objective to be (convex) quadratic in the parameter $v$ so that its exact minimizer can be found in closed form: 
\begin{align}
	\hat{\beta}(\eta) &= \Big(\bhM^{-} (\bhM^{-})^T + \eta \bI \Big)^{-1} \bhM^{-} Y_1^{-} \label{eq:beta_ridge}.
\end{align} 

\subsection{Imputation analysis.} \label{sec:imputation}

In this section, we highlight the importance of our de-noising procedure and prescribe a universal threshold (similar to that of \cite{usvt}) that dexterously distinguishes signal from noise, enabling the algorithm to capture the appropriate amount of useful information (encoded in the singular values of $\bY$) while discarding out the randomness. Due to its universality, the threshold naturally adapts to the amount of structure within $\bM$ in a purely data-driven manner. Specifically, for any choice of $\omega \in (0.1,1)$, we find that choosing
\begin{align} \label{eq:goldilocks}
	\mu &= (2 + \omega) \sqrt{T (\hat{\sigma}^2 \hat{p} + \hat{p} (1 - \hat{p}))},
\end{align}
results in an estimator with strong theoretical properties for both interpolation and extrapolation (discussed in Section \ref{sec:goldilocks}). Here, $\hat{p}$ and $\hat{\sigma}^2$ denote the unbiased maximum likelihood estimates of $p$ and $\sigma^2$, respectively, and can be computed via \eqref{eq:p_hat} and \eqref{eq:sample_variance}. 

The following Theorems (adapted from Theorems 2.1 and 2.7 of \cite{usvt}) demonstrate that Step 1 of our algorithm (detailed in Section \ref{sec:robust_algo}) accurately imputes missing entries within our data matrix $\bX$ when the signal matrix $\bM$ is either low rank or generated by an $\mathcal{L}$-Lipschitz function. In particular, Theorems \ref{thm:imputation_lowrank} and \ref{thm:imputation_lipschitz}  states that Step 1 produces a consistent estimator of the underlying mean matrix $\bM$ with respect to the (matrix) mean-squared-error, which is defined as
\begin{align} \label{eq:matrix_mse}
	\text{MSE}(\bhM) &=  \dfrac{1}{(N-1)T} \mathbb{E} \Big[ \sum_{i=2}^N \sum_{j=1}^T (\hat{M}_{ij} - M_{ij})^2 \Big]. 
\end{align}
We say that $\bhM$ is a consistent estimator of $\bM$ if the right-hand side of \eqref{eq:matrix_mse} converges to zero as $N$ and $T$ grow without bound.

The following theorem demonstrates that $\bhM$ is a good estimate of $\bM$ when $\bM$ is a low rank matrix, particularly when the rank of $\bM$ is small compared to $(N-1)p$.  \\

\begin{thm} \label{thm:imputation_lowrank} {\em (\textbf{Theorem 2.1 of \cite{usvt}})}
Suppose that $\bM$ is rank $k$. Suppose that $p \ge \frac{T^{-1 + \zeta}}{\sigma^2 + 1}$ for some $\zeta > 0$. Then using $\mu$ as defined in \eqref{eq:goldilocks}, 
	\begin{align}
		\emph{MSE}(\bhM) &\le C_1 \sqrt{\dfrac{k}{(N-1) p}} + \mathcal{O}\Big( \frac{1}{(N-1)T} \Big),
	\end{align}
	where $C_1$ is a universal positive constant.
\end{thm}

Suppose that the latent row and column feature vectors, $\{\theta_i\}$ and $\{\rho_j\}$, belong to some bounded, closed sets $K \subset \mathbb{R}^d$, where $d$ is some arbitrary but fixed dimension. If we assume $f: K \times K \rightarrow [-1, 1]$ possesses desirable smoothness properties such as Lipschitzness, then $\bhM$ is again a good estimate of $\bM$. \\

\begin{thm} \label{thm:imputation_lipschitz} {\em (\textbf{Theorem 2.7 of \cite{usvt}})}
	Suppose $f$ is a $\mathcal{L}$-Lipschitz function. Suppose that $p \ge \frac{T^{-1 + \zeta}}{\sigma^2 + 1}$ for some $\zeta > 0$. Then using $\mu$ as defined in \eqref{eq:goldilocks}, 
	\begin{align}
		\emph{MSE}(\bhM) &\le C(K, d, \mathcal{L}) \dfrac{(N-1)^{-\frac{1}{d+2}}}{\sqrt{p}} + \mathcal{O}\Big( \frac{1}{(N-1)T} \Big),
	\end{align}
	where $C(K, d, \mathcal{L})$ is a constant depending on $K, d,$ and $\mathcal{L}$.
\end{thm}

It is important to observe that the models under consideration for both Theorems \ref{thm:imputation_lowrank} and \ref{thm:imputation_lipschitz} encompass the mean matrices, $\bM$, generated as per many of the popular Econometric factor models often considered in literature and assumed in practice. Therefore, de-noising the data serves as an important imputing and filtering procedure for a wide array of applications.

\subsection{Forecasting analysis: pre-intervention regime.} \label{sec:goldilocks}
Similar to the setting for interpolation, the prediction performance metric of interest is the average mean-squared-error in estimating $M^-_1$ using $\hat{M}^-_1$. Precisely, we define
\begin{align} \label{eq:mse_def}
	\text{MSE}(\hat{M}^-_1) &= \frac{1}{T_0}\mathbb{E}\Big[ \sum_{t=1}^{T_0} (M_{1t} - \hat{M}_{1t})^2 \Big]. 
\end{align}
If the right-hand side of \eqref{eq:mse_def} approaches zero in the limit as $T_0$ grows without bound, then we say that $\hat{M}_1^-$ is a consistent estimator of $M_1^-$ (note that our analysis here assumes that only $T_0 \rightarrow \infty$). 

In what follows, we first state the finite sample bound on the average MSE between $\hat{M}_1^-$ and $M_1^-$ for the most generic setup (Theorem \ref{thm:finite-sample}). As a main Corollary of the result, we specialize the bound in the case where we use our prescribed universal threshold. Finally, we discuss a minor variation of the algorithm where the data is pre-processed, and specialize the above result to establish the consistency of our estimator (Theorem \ref{thm:consistency}). 

\subsubsection{General result.} We provide a finite sample error bound for the most generic setting, i.e. for any choice of the threshold, $\mu$, and regularization hyperparameter, $\eta$. \\
\begin{thm} \label{thm:finite-sample} 
	For any $\eta \ge 0$ and $\mu \ge 0$, the pre-intervention error of the algorithm can be bounded as
	\begin{align} \label{eq:general noise model}
		\MSE(\hat{M}^-_1) & \le \dfrac{C_1}{p^2T_0} \mathbb{E}  \Big( \lambda^* + \norm{ \bY - p\bM} + \norm{ (\hat{p} - p) \bM^-} \Big)^2 + \dfrac{2\sigma^2 \abs{S}}{T_0}  + \dfrac{ \eta \norm{\beta^*}^2}{T_0} + C_2 e^{-cp(N-1)T}.
	\end{align}
	Here, $\lambda_1, \dots, \lambda_{N-1}$ are the singular values of $p\bM$ in decreasing order and repeated by multiplicities, with $\lambda^* = \max_{i \notin S} \lambda_i$; $C_1, C_2$ and $c$ are universal positive constants.
\end{thm}

{\bf Bias-variance tradeoff.} Let us interpret the result by parsing the terms in the error bound. The last term decays exponentially with $(N-1)T$, as long as the fraction of observed entries is such that, on average, we see a super-constant number of entries, i.e. $p (N-1) T \gg 1$. More interestingly, the first two terms highlight the ``bias-variance tradeoff'' of the algorithm with respect to the singular value threshold $\mu$. Precisely, the size of the set $S$ increases with a decreasing value of the hyperparameter $\mu$, causing the second error term to increase. Simultaneously, however, this leads to a decrease in $\lambda^*$. Note that $\lambda^*$ denotes the aspect of the ``signal'' within the matrix $\bM$ that is not captured due to the thresholding through $S$. On the other hand, the second term, $|S| \sigma^2/ T_0$, represents the amount of ``noise'' captured by the algorithm, but wrongfully interpreted as a signal, during the thresholding process. In other words, if we use a large threshold, then our model may fail to capture pertinent information encoded in $\bM$; if we use a small threshold, then the algorithm may overfit the spurious patterns in the data. Thus, the hyperparameter $\mu$ provides a way to trade-off ``bias'' (first term) and ``variance'' (second term). 


\subsubsection{Goldilocks principle: a universal threshold.} Using the universal threshold defined in \eqref{eq:goldilocks}, we now highlight the prediction power of our estimator for any choice of $\eta$, the regularization hyperparameter. As described in Section \ref{sec:imputation}, the prescribed threshold automatically captures the ``correct'' level of information encoded in the (noisy) singular values of $\bY$ in a data-driven manner, dependent on the structure of $\bM$. However, unlike the statements in Theorems \ref{thm:imputation_lowrank} and \ref{thm:imputation_lipschitz}, the following bound does not require $\bM$ to be low rank or $f$ to be Lipschitz. \\

\begin{corollary} \label{corollary:universal_threshold}
	Suppose $p \ge \frac{T^{-1 + \zeta}}{\sigma^2 + 1}$ for some $\zeta > 0$. Let $T \le \alpha T_0$ for some constant $\alpha > 1$. Then for any $\eta \ge 0$ and using $\mu$ as defined in \eqref{eq:goldilocks}, the pre-intervention error is bounded above by
	\begin{align}
		\emph{MSE}(\hat{M}_1^-) &\le \dfrac{C_1}{p}(\sigma^2 + (1-p)) + \mathcal{O}(1 / \sqrt{T_0}),
	\end{align}
	where $C_1$ is a universal positive constant. 
\end{corollary}

As an implication, if $p = (1 + \vartheta) \sqrt{T_0} / (1 + \sqrt{T_0})$ and $\sigma^2 \le \vartheta$, we have that $\text{MSE}(\hat{M}^-_1) = \mathcal{O}(1/\sqrt{T_0})$. More generally, Corollary \ref{corollary:universal_threshold} shows that by adroitly capturing the signal, the resulting error bound simply depends on the variance of the noise terms, $\sigma^2$, and the error introduced due to missing data. Ideally, one would hope to overcome the error term when $T_0$ is sufficiently large. This motivates the following setup.  

\subsubsection{Consistency.} We present a straightforward pre-processing step that leads to the consistency of our algorithm. The pre-processing step simply involves replacing the columns of $\bX$ by the averages of subsets of its columns. This admits the same setup as before, but with the variance for each noise term reduced. An implicit side benefit of this approach is that required SVD step in the algorithm is now applied to a matrix of smaller dimensions. 

To begin, partition the $T_0$ columns of the pre-intervention data matrix $\bX^{-}$ into $\Delta$ blocks, each of size $\tau = \lfloor T_0 / \Delta \rfloor$ except potentially the last block, which we shall ignore for theoretical purposes; in practice, however, the remaining columns can be placed into the last block. Let $B_j = \{ (j-1) \tau + \ell : 1\leq \ell \le \tau \}$ denote the column indices of $\bX^{-}$ within partition $j \in [\Delta]$. Next, we replace the $\tau$ columns within each partition by their average, and thus create a new matrix, $\bbX^{-}$, with $\Delta$ columns and $N-1$ rows. Precisely, $\bbX^{-} = [\bar{X}_{ij}]_{2\leq i\leq N, j \in [\Delta]}$ with 
\begin{align}
	\bar{X}_{ij} &= \dfrac{1}{\tau} \sum_{t \in B_j} X_{it} \cdot D_{it},
\end{align}
where 
\begin{align*}
	D_{it} = \begin{cases}
		1 & \text{if}\ X_{it} \text{ is observed}, \\
		0 & \text{otherwise.}
		\end{cases}
\end{align*}
For the treatment row, let $\bar{X}_{1j} = \frac{\hat{p}}{\tau} \sum_{t \in B_j} X_{1t}$ for all $j \in [\Delta]$\footnote{Although the statement in Theorem \ref{thm:consistency} assumes that an oracle provides the true $p$, we prescribe practitioners to use $\hat{p}$ since $\hat{p}$ converges to $p$ almost surely by the Strong Law of Large Numbers.}. Let $\bar{\bM}^{-} = [\bar{M}_{ij}]_{2\leq i\leq N, j \in [\Delta]}$ with
\begin{align}
	\bar{M}_{ij} & = \mathbb{E}[\bar{X}_{ij}] ~ = \frac{p}{\tau} \sum_{t \in B_j} M_{it}.
\end{align}
We apply the algorithm to $\bbX^{-}$ to produce the estimate $\hat{\bar{\bM}}^{-}$ of $\bar{\bM}^{-}$, 
which is sufficient to produce $\hat{\beta}(\eta)$. This $\hat{\beta}(\eta)$ can be used to produce the post-intervention synthetic control means
$\hat{M}_{1}^{+} = [\hat{M}_{1t}]_{T_0 < t \leq T}$ in a similar manner as before \footnote{In practice, one can first de-noise $\bX^+$ via step one of Section \ref{sec:algorithm}, and use the entries of $\bhM^+$ in \eqref{eq:post_avg}.}: for $T_0 < t \leq T$,
\begin{align} \label{eq:post_avg}
	\hat{M}_{1t} & = \sum_{i=2}^N \hat{\beta}_i(\eta) X_{it}.
\end{align}
For the pre-intervention period, we produce the estimator $\hat{\bar{M}}_1^{-} = [\hat{\bar{M}}_{1j}]_{j \in [\Delta]}$: for $j \in [\Delta]$,
\begin{align}
	\hat{\bar{M}}_{1j} & =  \sum_{i=2}^N \hat{\beta}_i(\eta) \hat{\bar{M}}_{ij}.
\end{align}
Our measure of estimation error is defined as 
\begin{align}
	\text{MSE}(\hat{\bar{M}}_1^{-}) & = \dfrac{1}{\Delta} \mathbb{E} \Big[ \sum_{j=1}^{\Delta} (\bar{M}_{1j} -  \hat{\bar{M}}_{1j})^2\Big].
\end{align}
For simplicity, we will analyze the case where each block contains at least one entry such that $\bbX^-$ is completely observed. We now state the following result. \\

\begin{thm} \label{thm:consistency} 
		Fix any $\gamma \in (0,1/2)$ and $\omega \in (0.1, 1)$. Let $\Delta = T_0^{\frac{1}{2} + \gamma}$ and $\mu = (2 + \omega) \sqrt{ T_0^{2 \gamma} (\hat{\sigma}^2 \hat{p} + \hat{p}(1 - \hat{p}))}$. Suppose $p \ge \frac{T_0^{-2 \gamma}}{\sigma^2 + 1}$ is known. Then for any $\eta \ge 0$, 
	\begin{align}
		\emph{MSE}(\hat{\bar{M}}_1^{-}) &= \mathcal{O}(T_0^{-1/2 + \gamma}).
	\end{align}
\end{thm} 

We note that the method of \cite[Sec 2.3]{abadie3} learns the weights (here $\hat{\beta}(0)$) by pre-processing the data. One common pre-processing proposal is to also aggregate the columns, but the aggregation parameters are chosen by solving an optimization problem to minimize the resulting prediction error of the observations. In that sense, the above averaging of column is a simple, data agnostic approach to achieve a similar effect, and potentially more effectively.

\subsection{Forecasting analysis: post-intervention regime.}
For the post-intervention regime, we consider the average root-mean-squared-error in measuring the performance of our algorithm. Precisely, we define
\begin{align} \label{eq:mse_def}
	\text{RMSE}(\hat{M}^+_1) &= \frac{1}{\sqrt{T - T_0}}\mathbb{E} \Big[ \Big( \sum_{t>T_0}^{T} (M_{1t} - \hat{M}_{1t})^2 \Big)^{1/2} \Big]. 
\end{align}
The key assumption of our analysis is that the treatment unit signal can be written as a linear combination of donor pool signals. Specifically, we assume that this relationship holds in the pre-intervention regime, i.e. $M_1^{-} = (\bM^{-})^T\beta^*$ for some $\beta^* \in {\mathbb R}^{N-1}$ as stated in \eqref{eq:3}. However, the question still remains: does the same relationship hold for the post-intervention regime and if so, under what conditions does it hold? We state a simple linear algebraic fact to this effect, justifying the approach of synthetic control. It is worth noting that this important aspect has been amiss in the literature, potentially implicitly believed or assumed starting in the work by \cite{abadie3}. \\

\begin{thm} \label{thm:post}
	Let \eqref{eq:3} hold for some $\beta^*$. Let $\rank(\bM^{-}) = \rank(\bM)$. Then $M_{1}^{+} = (\bM^{+})^T \beta^*$. 
\end{thm}

If we assume that the linear relationship prevails in the post-intervention period, then we arrive at the following error bound. \\



\begin{thm} \label{thm:post_rmse}
Suppose $p \ge \frac{T^{-1 + \zeta}}{\sigma^2 + 1}$ for some $\zeta > 0$. Suppose $\norm{\hat{\beta}(\eta)}_{\infty} \le \psi$ for some $\psi > 0$. Let $\alpha' T_0 \le T \le \alpha T_0$ for some constants $\alpha', \alpha > 1$. Then for any $\eta \ge 0$ and using $\mu$ as defined in \eqref{eq:goldilocks}, the post-intervention error is bounded above by
	\begin{align*}
		\emph{RMSE}(\hat{M}_1^+)  &\le \dfrac{C_1}{\sqrt{p}} (\sigma^2 + (1-p))^{1/2} + \dfrac{C_2 \norm{\bM}}{\sqrt{T_0}} \cdot \E\norm{\hat{\beta}(\eta) - \beta^*} + \mathcal{O}(1/ \sqrt{T_0}),
	\end{align*}
	where $C_1$ and $C_2$ are universal positive constants.
\end{thm}

\textbf{Benefits of regularization.} In order to motivate the use of regularization, we analyze the error bounds of Theorems \ref{thm:finite-sample} and \ref{thm:post_rmse} to observe how the pre- and post-intervention errors react to regularization. As seen from Theorem \ref{thm:finite-sample}, the pre-intervention error {\em increases} linearly with respect to the choice of $\eta$. Intuitively, this increase in pre-intervention error derives from the fact that regularization reduces the model complexity, which biases the model and handicaps its ability to fit the data. At the same time, by restricting the hypothesis space and controlling the ``smoothness'' of the model, regularization prevents the model from overfitting to the data, which better equips the model to generalize to unseen data. Therefore, a larger value of $\eta$ {\em reduces} the post-intervention error. This can be seen by observing the second error term of Theorem \ref{thm:post_rmse}, which is controlled by the expression $\norm{\hat{\beta}(\eta) - \beta^*}$. In words, this error is a function of the learning algorithm used to estimate $\beta^*$. Interestingly, \cite{farebrother} demonstrates that there exists an $\eta > 0$ such that 
\begin{align*}
	\norm{\hat{\beta}(\eta) - \beta^*} &\le \norm{\hat{\beta}(0) - \beta^*},
\end{align*}
without any assumptions on the rank of $\bhM^{-}$. In other words, \cite{farebrother} demonstrates that regularization can decrease the MSE between $\hat{\beta}(\eta)$ and the true $\beta^*$, thus reducing the overall error. Ultimately, employing ridge regression introduces extraneous bias into our model, yielding a higher pre-intervention error. In exchange, regularization reduces the post-intervention error (due to smaller variance).

\subsection{Bayesian analysis.}
We turn our attention to a Bayesian treatment of synthetic control. By operating under a Bayesian framework, we allow practitioners to naturally encode domain knowledge into prior distributions while simultaneously avoiding the problem of overfitting. In addition, rather than making point estimates, we can now quantitatively express the uncertainty in our estimates with posterior probability distributions.

We begin by treating $\beta^*$ as a random variable as opposed to an unknown constant. In this approach, we specify a prior distribution, $p(\beta)$, that expresses our apriori beliefs and preferences about the underlying parameter (synthetic control). Given some new observation for the donor units, our goal is to make predictions for the counterfactual treatment unit on the basis of a set of pre-intervention (training) data. For the moment, let us assume that the noise parameter $\sigma^2$ is a known quantity and that the noise is drawn from a Gaussian distribution with zero-mean; similarly, we temporarily assume $\bM^{-}$ is also given. Let us denote the vector of donor estimates as $M_{\cdot t} = [M_{it}]_{2 \le i \le N}$; we define $X_{\cdot t}$ similarly. Denoting the pre-intervention data as $D = \{ (Y_{1t}, M_{\cdot t}): t \in [T_0] \}$, the likelihood function $p(Y_1^{-} \given \beta, \bM^{-})$ is expressed as
\begin{align}
	p(Y_1^{-} \given \beta, \bhM^{-}) &= \mathcal{N}( (\bM^{-})^T \beta, \sigma^2 \bI),
\end{align}
an exponential of a quadratic function of $\beta$. The corresponding conjugate prior, $p(\beta)$, is therefore given by a Gaussian distribution, i.e. $\beta \distas{} \mathcal{N}( \beta \given \beta_0, \bSigma_0)$ with mean $\beta_0$ and covariance $\Sigma_0$. By using a conjugate Gaussian prior, the posterior distribution, which is proportional to the product of the likelihood and the prior, will also be Gaussian. Applying Bayes' Theorem (derivation unveiled in the Appendix), we have that the posterior distribution is $p(\beta \given D) = \mathcal{N}( \beta_D, \bSigma_D)$ where
\begin{align}
	\bSigma_D &= \Big( \bSigma_0^{-1} + \dfrac{1}{\sigma^2} \bM^{-} (\bM^{-})^T \Big)^{-1}
	\\ \beta_D &= \bSigma_D \Bigg( \dfrac{1}{\sigma^2} \bM^{-} Y_1^{-}+ \bSigma_0^{-1} \beta_0 \Bigg).
\end{align}
For the remainder of this section, we shall consider a popular form of the Gaussian prior. In particular, we consider a zero-mean isotropic Gaussian with the following parameters: $\beta_0 = 0$ and $\bSigma_0 = \alpha^{-1} \bI$ for some choice of $\alpha > 0$. Since $\bM^{-}$ is unobserved by the algorithm, we use the estimated $\bhM^{-}$, computed as per step one of Section \ref{sec:algorithm}, as a proxy; therefore, we redefine our data as $D = \{ (Y_{1t}, \hat{M}_{\cdot t}): t \in [T_0] \}$. Putting everything together, we have that $p(\beta \given D) = \mathcal{N}( \beta_D, \bSigma_D)$ whereby
\begin{align}
	\bSigma_D &= \Big( \alpha \bI + \dfrac{1}{\sigma^2} \bhM^{-} (\bhM^{-})^T \Big)^{-1} \label{eq:post_cov}
	\\ \beta_D &= \dfrac{1}{\sigma^2} \bSigma_D\bhM^{-} Y_1^{-}
	\\ &= \dfrac{1}{\sigma^2} \Bigg( \dfrac{1}{\sigma^2} \bhM^{-} (\bhM^{-})^T + \alpha \bI \Bigg)^{-1} \bhM^{-} Y_1^{-}\label{eq:post_mean}.
\end{align} 

\subsubsection{Maximum a posteriori (MAP) estimation.}
By using the zero-mean, isotropic Gaussian conjugate prior, we can derive a point estimate of $\beta^*$ by maximizing the log posterior distribution, which we will show is equivalent to minimizing the regularized objective function of \eqref{eq:ls} for a particular choice of $\eta$. In essence, we are determining the optimal $\hat{\beta}$ by finding the most probable value of $\beta^*$ given the data and under the influence of our prior beliefs. The resulting estimate is known as the maximum a posteriori (MAP) estimate. 

We begin by taking the log of the posterior distribution, which gives the form
\begin{align*}
	\ln{ p(\beta \given D) } &= -\dfrac{1}{2 \sigma^2} \norm{Y_1^{-}- (\bhM^{-})^T \beta}^2 - \dfrac{\alpha}{2} \norm{\beta}^2 + \text{const.}
\end{align*}
Maximizing the above log posterior then equates to minimizing the quadratic regularized error \eqref{eq:ls} with $\eta = \alpha \sigma^2$. We define the MAP estimate, $\hat{\beta}_{\text{MAP}}$, as
\begin{align}
	\hat{\beta}_{\text{MAP}} &= \argmax_{\beta \in \mathbb{R}^{N-1}} \ln{ p (\beta \given D) } \nonumber
	\\ &= \argmin_{\beta \in \mathbb{R}^{N-1}} \dfrac{1}{2} \norm{Y_1^{-}- (\bhM^{-})^T \beta}^2 + \dfrac{\alpha \sigma^2}{2} \norm{\beta}^2 \nonumber
	\\ &= \Big( \bhM^{-} (\bhM^{-})^T + \alpha \sigma^2 \bI \Big)^{-1} \bhM^{-} Y_1^{-} \label{eq:beta_map}.
\end{align}
With the MAP estimate at hand, we then make predictions of the counterfactual as
\begin{align}
	\hat{M}_1 &= \bhM^T \hat{\beta}_{\text{MAP}}.
\end{align}
Therefore, we have seen that the MAP estimation is equivalent to ridge regression since the introduction of an appropriate prior naturally induces the additional complexity penalty term.


\subsubsection{Fully Bayesian treatment.}
Although we have treated $\beta^*$ as a random variable attached with a prior distribution, we can venture beyond point estimates to be fully Bayesian. In particular, we will make use of the posterior distribution over $\beta^*$ to marginalize over all possible values of $\beta^*$ in evaluating the predictive distribution over $Y_1^{-}$. We will decompose the regression problem of predicting the counterfactual into two separate stages: the \textit{inference} stage in which we use the pre-intervention data to learn the predictive distribution (defined shortly), and the subsequent \textit{decision} stage in which we use the predictive distribution to make estimates. By separating the inference and decision stages, we can readily develop new estimators for different loss functions without having to relearn the predictive distribution, providing practitioners tremendous flexibility with respect to decision making. 

Let us begin with a study of the inference stage. We evaluate the predictive distribution over $Y_{1t}$, which is defined as
\begin{align}
	p(Y_{1t} \given \hat{M}_{\cdot t}, D) &= \int p(Y_{1t} \given \hat{M}_{\cdot t}, \beta) \,\, p(\beta \given D) \, d\beta \nonumber
	\\ &= \mathcal{N}( \hat{M}_{\cdot t}^T \, \beta_D, \sigma^2_D),
\end{align}
where
\begin{align}
	\sigma^2_D &= \sigma^2 + \hat{M}_{\cdot, t}^T \bSigma_D \hat{M}_{\cdot, t}.
\end{align} 
Note that $p(\beta \given D)$ is the posterior distribution over the synthetic control parameter and is governed by \eqref{eq:post_cov} and \eqref{eq:post_mean}. With access to the predictive distribution, we move on towards the decision stage, which consists of determining a particular estimate $\hat{M}_{1t}$ given a new observation vector $X_{\cdot t}$ (used to determine $\hat{M}_{\cdot t}$). Consider an arbitrary loss function $L(Y_{1t}, g(\hat{M}_{\cdot t}))$ for some function $g$. The expected loss is then given by
\begin{align}
	\mathbb{E}[L] &= \int \int L( Y_{1t}, g(\hat{M}_{\cdot t})) \cdot p(Y_{1t}, \hat{M}_{\cdot t}) \, d Y_{1t} \, d \hat{M}_{\cdot t} \nonumber
	\\ &= \int \Bigg( \int L( Y_{1t}, g(\hat{M}_{\cdot t})) \cdot p(Y_{1t} \given \hat{M}_{\cdot t}) \, d Y_{1t} \Bigg) p(\hat{M}_{\cdot t}) \,d \hat{M}_{\cdot t} \label{eq:loss},
\end{align}
and we choose our estimator $\hat{g}(\cdot)$ as the function that minimizes the average cost, i.e.,
\begin{align}
	\hat{g}(\cdot) &= \argmin_{g(\cdot)} \mathbb{E}[ L(Y_{1t}, g(\hat{M}_{\cdot t}))].
\end{align}
Since $p(\hat{M}_{\cdot t}) \ge 0$, we can minimize \eqref{eq:loss} by selecting $\hat{g}(\hat{M}_{\cdot t})$ to minimize the term within the parenthesis for each individual value of $Y_{1t}$, i.e., 
\begin{align}
	\hat{M}_{1t} &= \hat{g}(\hat{M}_{\cdot t}) \nonumber
	\\ &= \argmin_{g(\cdot)} \int L( Y_{1t}, g(\hat{M}_{\cdot t})) \cdot p(Y_{1t} \given \hat{M}_{\cdot t}) \, d Y_{1t} \label{eq:loss_opt}.
\end{align}
As suggested by \eqref{eq:loss_opt}, the optimal estimate $\hat{M}_{1t}$ for a particular loss function depends on the model only through the predictive distribution $p(Y_{1t} \given \hat{M}_{\cdot t}, D)$. Therefore, the predictive distribution summarizes all of the necessary information to construct the desired Bayesian estimator for any given loss function $L$. 

\subsubsection{Bayesian least-squares estimate.}
We analyze the case for the squared loss function (MSE), a common cost criterion for regression problems. In this case, we write the expected loss as
\begin{align*}
	\mathbb{E}[L] &= \int \Bigg( \int ( Y_{1t} - g(\hat{M}_{\cdot t}))^2 \cdot p(Y_{1t} \given \hat{M}_{\cdot t}) \, d Y_{1t} \Bigg) p(\hat{M}_{\cdot t}) \,d \hat{M}_{\cdot t}.
\end{align*}
Under the MSE cost criterion, the optimal estimate is the mean of the predictive distribution, also known as the Bayes' least-squares (BLS) estimate:
\begin{align}
	\hat{M}_{1t} &= \mathbb{E}[ Y_{1t} \given  \hat{M}_{\cdot t}, D] \nonumber
	\\ &= \int Y_{1t} \,\, p (Y_{1t} \given \hat{M}_{\cdot t}, D) dY_{1t} \nonumber
	\\ &= \hat{M}_{\cdot t}^T \, \beta_D.
\end{align}

\begin{remark}
	Since the noise variance $\sigma^2$ is usually unknown in practice, we can introduce another conjugate prior distribution $p(\beta, 1/\sigma^2)$ given by the Gaussian-gamma distribution. This prior yields a Student's $t$-distribution for the predictive probability distribution. Alternatively, one can estimate $\sigma^2$ via \eqref{eq:sample_variance}. 
\end{remark}


\section{Experiments} \label{sec:experiments}

We begin by exploring two real-world case studies discussed in \cite{abadie1, abadie2, abadie3} that demonstrate the ability of the original synthetic control's algorithm to produce a reliable counterfactual reality. We use the same case-studies to showcase the ``robustness'' property of our proposed algorithm. Specifically, we demonstrate that our algorithm reproduces similar results even in presence of missing data, and without knowledge of the extra covariates utilized by prior works. We find that our approach, surprisingly, also discovers a few subtle effects that seem to have been overlooked in prior studies. In the following empirical studies, we will employ three different learning procedures as described in the robust synthetic control algorithm: (1) linear regression ($\eta = 0$), (2) ridge regression ($\eta > 0$), and (3) LASSO ($\zeta > 0$). 


As described in \cite{abadie1, abadie2, abadie4}, the synthetic control method allows a practitioner to evaluate the reliability of his or her case study results by running placebo tests. One such placebo test is to apply the synthetic control method to a donor unit. Since the control units within the donor pool are assumed to be unaffected by the intervention of interest (or at least much less affected in comparison), one would expect that the estimated effects of intervention for the placebo unit should be less drastic and divergent compared to that of the treated unit. Ideally, the counterfactuals for the placebo units would show negligible effects of intervention. Similarly, one can also perform exact inferential techniques that are similar to permutation tests. This can be done by applying the synthetic control method to every control unit within the donor pool and analyzing the gaps for every simulation, and thus providing a distribution of estimated gaps. In that spirit, we present the resulting placebo tests (for only the case of linear regression) for the Basque Country and California Prop. 99 case studies below to assess the significance of our estimates.

We will also analyze both case studies under a Bayesian setting. From our results, we see that our predictive uncertainty, captured by the standard deviation of the predictive distribution, is influenced by the number of singular values used in the de-noising process. Therefore, we have plotted the eigenspectrum of the two case study datasets below. Clearly, the bulk of the signal contained within the datasets is encoded into the top few singular values -- in particular, the top two singular values. Given that the validation errors computed via forward chaining are nearly identical for low-rank settings (with the exception of a rank-1 approximation), we shall use a rank-2 approximation of the data matrix. In order to exhibit the role of thresholding in the interplay between bias and variance, we also plot the cases where we use threshold values that are too high (bias) or too low (variance). For each figure, the dotted blue line will represent our posterior predictive means while the shaded light blue region spans one standard deviation on both sides of the mean. As we shall see, our predictive uncertainty is smallest in the neighborhood around the pre-intervention period. However, the level of uncertainty increases as we deviate from the the intervention point, which appeals to our intuition.

In order to choose an appropriate choice of the prior parameter $\alpha$, we first use forward-chaining for the ridge regression setting to find the optimal regularization hyperparameter $\eta$. By observing the expressions of \eqref{eq:beta_ridge} and \eqref{eq:beta_map}, we see that $\eta = \alpha \sigma^2$ since ridge regression is closely related to MAP estimation for a zero-mean, isotropic Gaussian prior. Consequently, we choose $\alpha = \eta / \hat{\sigma}^2$ where $\eta$ is the value obtained via forward chaining.

\begin{figure}[h]
	\centering
	\begin{subfigure}[b]{0.45\textwidth}
		\centering
		\includegraphics[width=\textwidth]{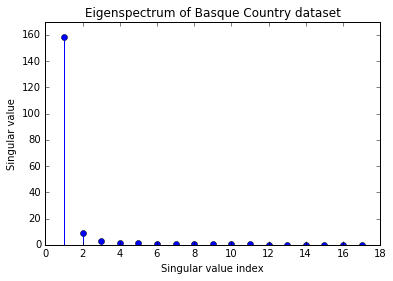}
		\caption{Eigenspectrum of Basque data.}
		\label{fig:basque_eig}
	\end{subfigure}
	\begin{subfigure}[b]{0.45\textwidth}
		\centering
		\includegraphics[width=\textwidth]{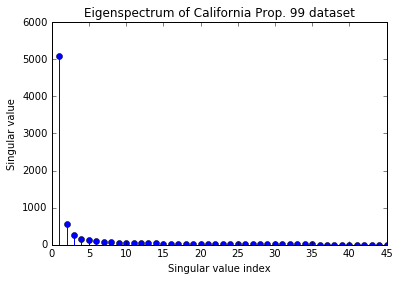}
		\caption{Eigenspectrum of California data.}
		\label{fig:cali_eig}
	\end{subfigure}
\end{figure}

\subsection{Basque Country}
The goal of this case-study is to investigate the effects of terrorism on the economy of Basque Country using the neighboring Spanish regions as the control group. 
In 1968, the first Basque Country victim of terrorism was claimed; however, it was not until the mid-1970s did the terrorist activity 
become more rampant \cite{abadie3}. To study the economic ramifications of terrorism on Basque Country, we only use as data the per-capita GDP (outcome variable) 
of 17 Spanish regions from 1955-1997. We note that in \cite{abadie3}, 13 additional predictor variables for 
each region were used including demographic information pertaining to one's educational status, and average shares for six industrial sectors.

\textbf{Results.} Figure \ref{fig:basque_a} shows that our method (for all three estimators) produces a very similar qualitative synthetic control 
to the original method even though we do not utilize additional predictor variables. Specifically, the synthetic control resembles the observed 
GDP in the pre-treatment period between 1955-1970. However, due to the large-scale terrorist activity in the mid-70s, there is a noticeable 
economic divergence between the synthetic and observed trajectories beginning around 1975. This deviation suggests that terrorist activity 
negatively impacted the economic growth of Basque Country. 

One subtle difference between our synthetic control -- for the case of linear and ridge regression -- and that of \cite{abadie3} is between 1970-75: our approach suggests that
there was a small, but noticeable economic impact starting just prior to 1970, potentially due to first terrorist attack in 1968. Notice, however, that the original synthetic control of \cite{abadie3} diverges only after 1975. Our LASSO estimator's trajectory also agrees with that of the original synthetic control method's, which is intuitive since both estimators seek sparse solutions.

To study the robustness of our approach with respect to missing entries, we discard each data point uniformly at random with probability $1-p$. The resulting control for different values of $p$ is presented in Figure \ref{fig:basque_b} suggesting the robustness of our (linear) algorithm. 
Finally, we produce Figure \ref{fig:basque_c} by applying our algorithm without the de-noising step. As evident from the Figure, the resulting predictions suffer drastically, reinforcing the value of de-noising. Intuitively, using an appropriate threshold $\mu$ equates to selecting the correct model complexity, which helps safeguard the algorithm from potentially overfitting to the training data. 

\begin{figure}[H]
	\centering
	\begin{subfigure}[b]{0.325\textwidth}
		\centering
		\includegraphics[width=\textwidth]{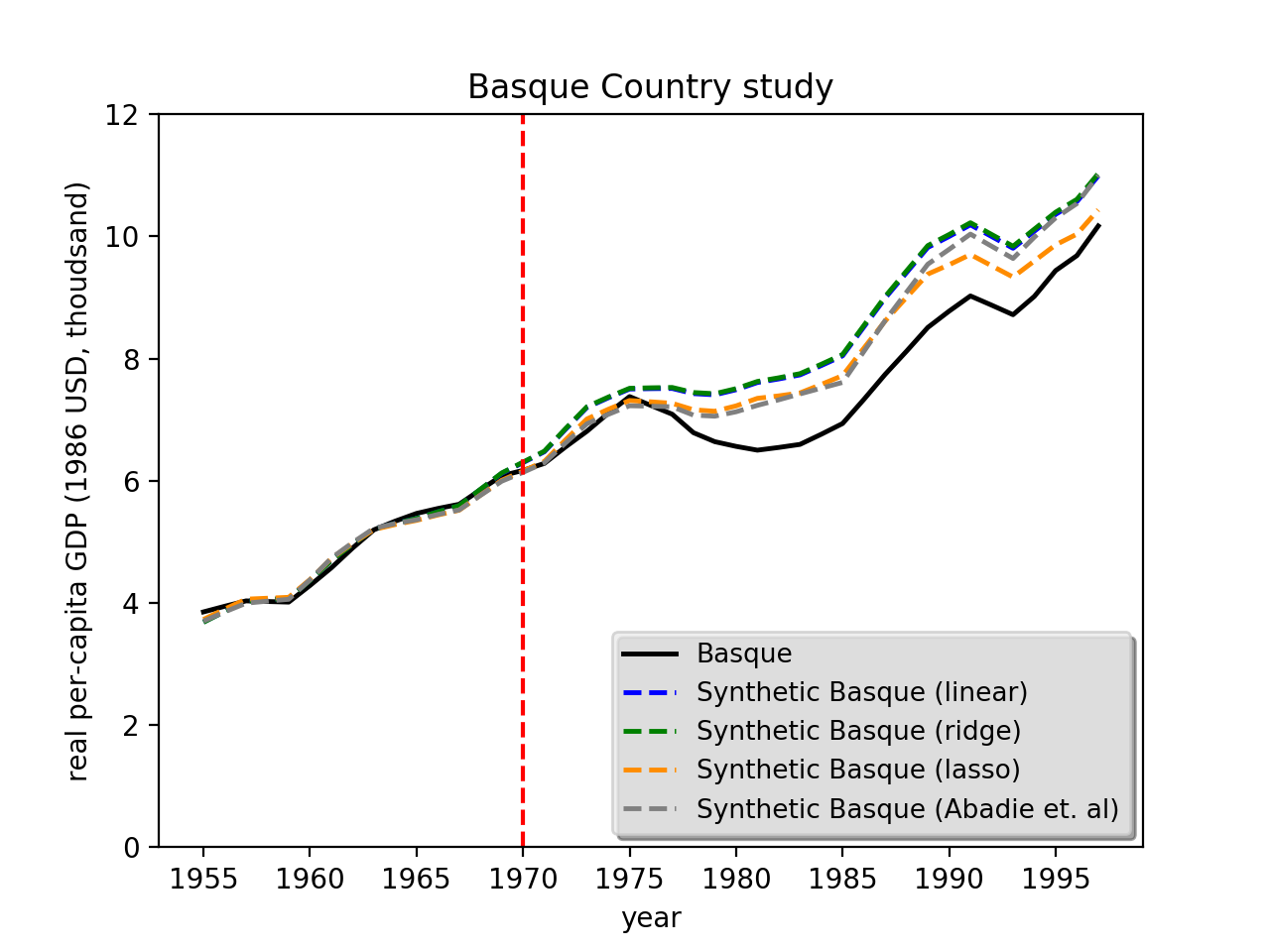}
		\caption{Comparison of methods.}
		\label{fig:basque_a}
	\end{subfigure}
	\begin{subfigure}[b]{0.325\textwidth}
		\centering
		\includegraphics[width=\textwidth]{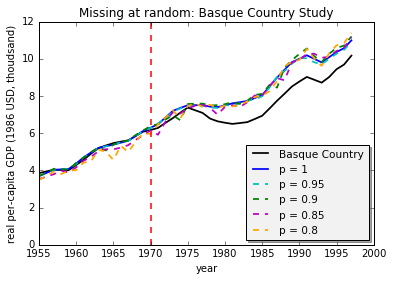}
		\caption{Missing data.}
		\label{fig:basque_b}
	\end{subfigure}
	\begin{subfigure}[b]{0.325\textwidth}
		\centering
		\includegraphics[width=\textwidth]{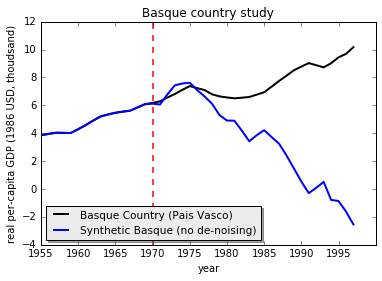}
		\caption{Impact of de-noising.}
		\label{fig:basque_c}
	\end{subfigure}
	\caption{Trends in per-capita GDP between Basque Country vs. synthetic Basque Country. }
\end{figure}

\textbf{Placebo tests.}
We begin by applying our robust algorithm to the Spanish region of Cataluna, a control unit that is not only similar to Basque Country, but also exposed to a much lower level of terrorism \cite{abadie2}. Observing both the synthetic and observed economic evolutions of Cataluna in Figure \ref{fig:cataluna}, we see that there is no identifiable treatment effect, especially compared to the divergence between the synthetic and observed Basque trajectories. We provide the results for the regions of Aragon and Castilla Y Leon in Figures \ref{fig:aragon} and \ref{fig:castilla}.

\begin{figure}[H]
	\centering
	\begin{subfigure}[b]{0.325\textwidth}
		\centering
		\includegraphics[width=\textwidth]{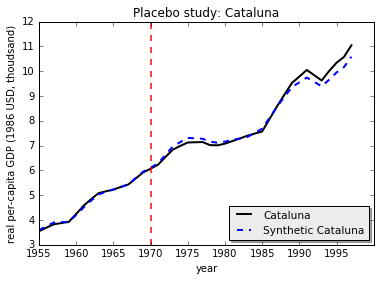}
		\caption{Cataluna.}
		\label{fig:cataluna}
	\end{subfigure}
	\begin{subfigure}[b]{0.325\textwidth}
		\centering
		\includegraphics[width=\textwidth]{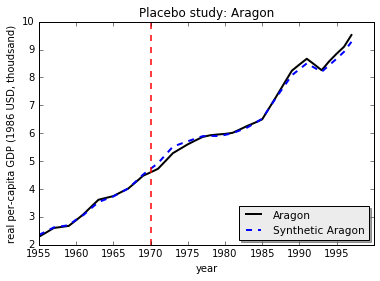}
		\caption{Aragon.}
		\label{fig:aragon}
	\end{subfigure}
	\begin{subfigure}[b]{0.325\textwidth}
		\centering
		\includegraphics[width=\textwidth]{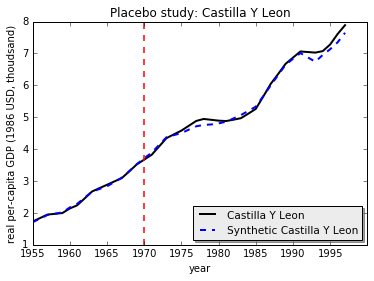}
		\caption{Castilla Y Leon.}
		\label{fig:castilla}
	\end{subfigure}
	\caption{Trends in per-capita GDP for placebo regions. }
\end{figure}

Finally, similar to \cite{abadie2}, we plot the differences between our estimates and the observations for Basque Country and all other regionals, individually, as placebos. Note that \cite{abadie2} excluded five regions that had a poor pre-intervention fit but we keep all regions. Figure  \ref{fig:basque_placebo1} shows the resulting plot for all regions with the solid black line being Basque Country. This plot helps visualize the extreme post-intervention divergence between the predicted means and the observed values for Basque. Up until about 1990, the divergence for Basque Country is the most extreme compared to all other regions (placebo studies) lending credence to the belief that the effects of terrorism on Basque Country were indeed significant. Refer to Figure \ref{fig:basque_placebo2} for the same test but with Madrid and Balearic Islands excluded, as per \cite{abadie2}. The conclusions drawn should remain the same, pointing to the robustness of our approach.

\begin{figure}[H]
	\centering
	\begin{subfigure}[b]{0.35\textwidth}
		\centering
		\includegraphics[width=\textwidth]{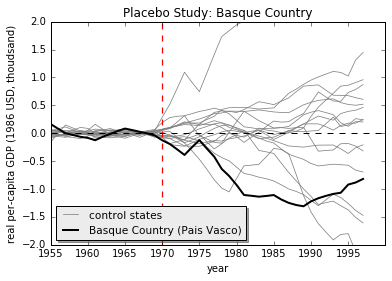}
		\caption{Includes all control regions.}
		\label{fig:basque_placebo1}
	\end{subfigure}
	\begin{subfigure}[b]{0.35\textwidth}
		\centering
		\includegraphics[width=\textwidth]{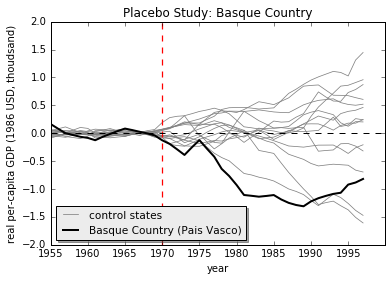}
		\caption{Excludes 2 regions.}
		\label{fig:basque_placebo2}
	\end{subfigure}
	\caption{Per-capita GDP gaps for Basque Country and control regions.}
\end{figure}


\textbf{Bayesian approach.}
We plot the resulting Bayesian estimates in the figures below under varying thresholding conditions. It is interesting to note that our uncertainty grows dramatically once we include more than two singular values in the thresholding process. This confirms what our theoretical results indicated earlier: choosing a smaller threshold, $\mu$, would lead to a greater number of singular values retained which results in higher variance. On the other hand, notice that just selecting 1 singular value results in an apparently biased estimate which is overestimating the synthetic control. It appears that selecting the top two singular values balance the bias-variance tradeoff the best and is also agrees with our earlier finding that the data matrix appears to be of rank 2 or 3. Note that in this setting, we would find it hard to reject the null-hypothesis because the observations for the treated unit lie within the uncertainty band of the estimated synthetic control.

\begin{figure}[H]
	\centering
	\begin{subfigure}[b]{0.325\textwidth}
		\centering
		\includegraphics[width=\textwidth]{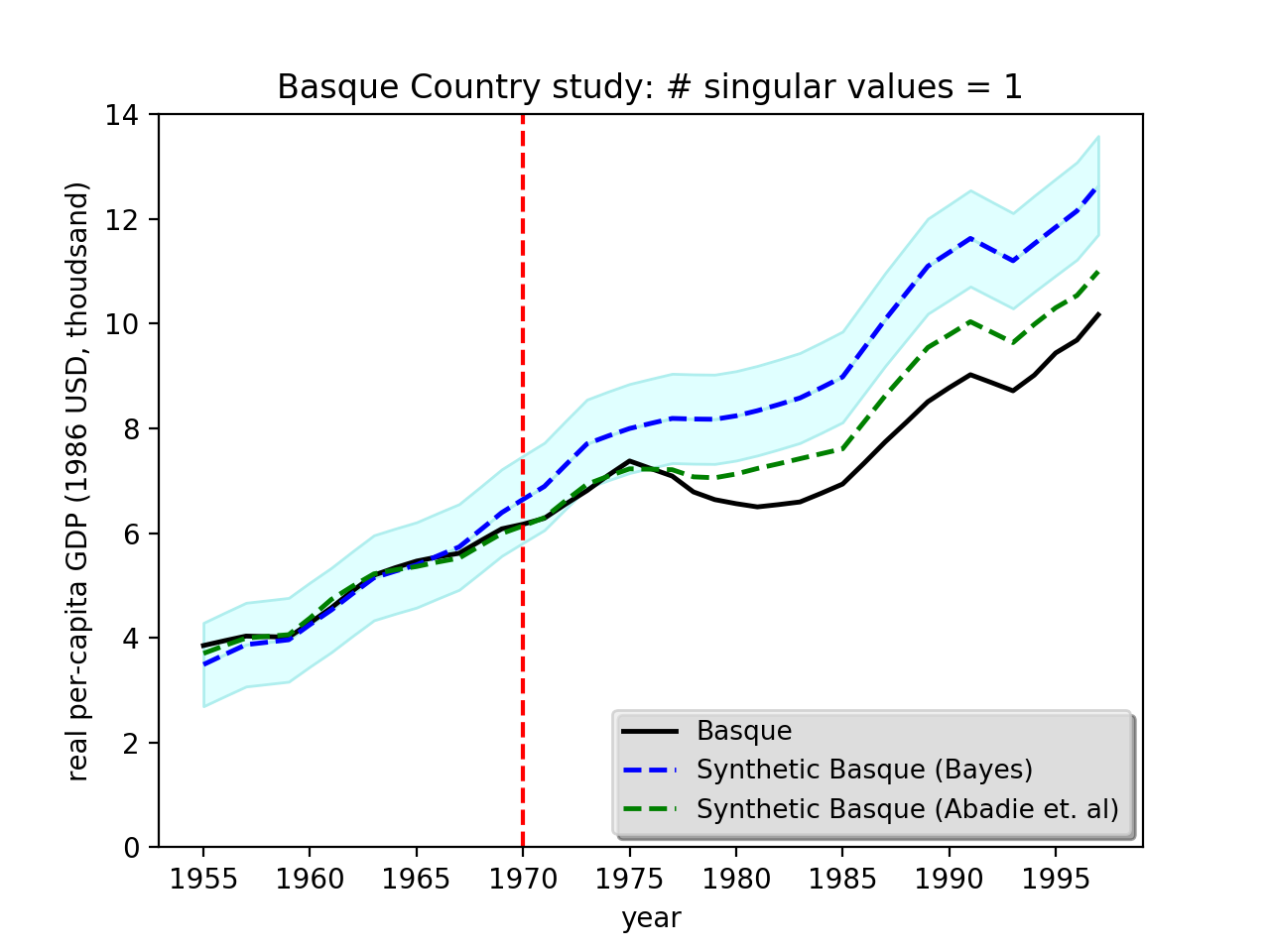}
		\caption{Top singular value.}
		\label{fig:basque_bayes1}
	\end{subfigure}
	\begin{subfigure}[b]{0.325\textwidth}
		\centering
		\includegraphics[width=\textwidth]{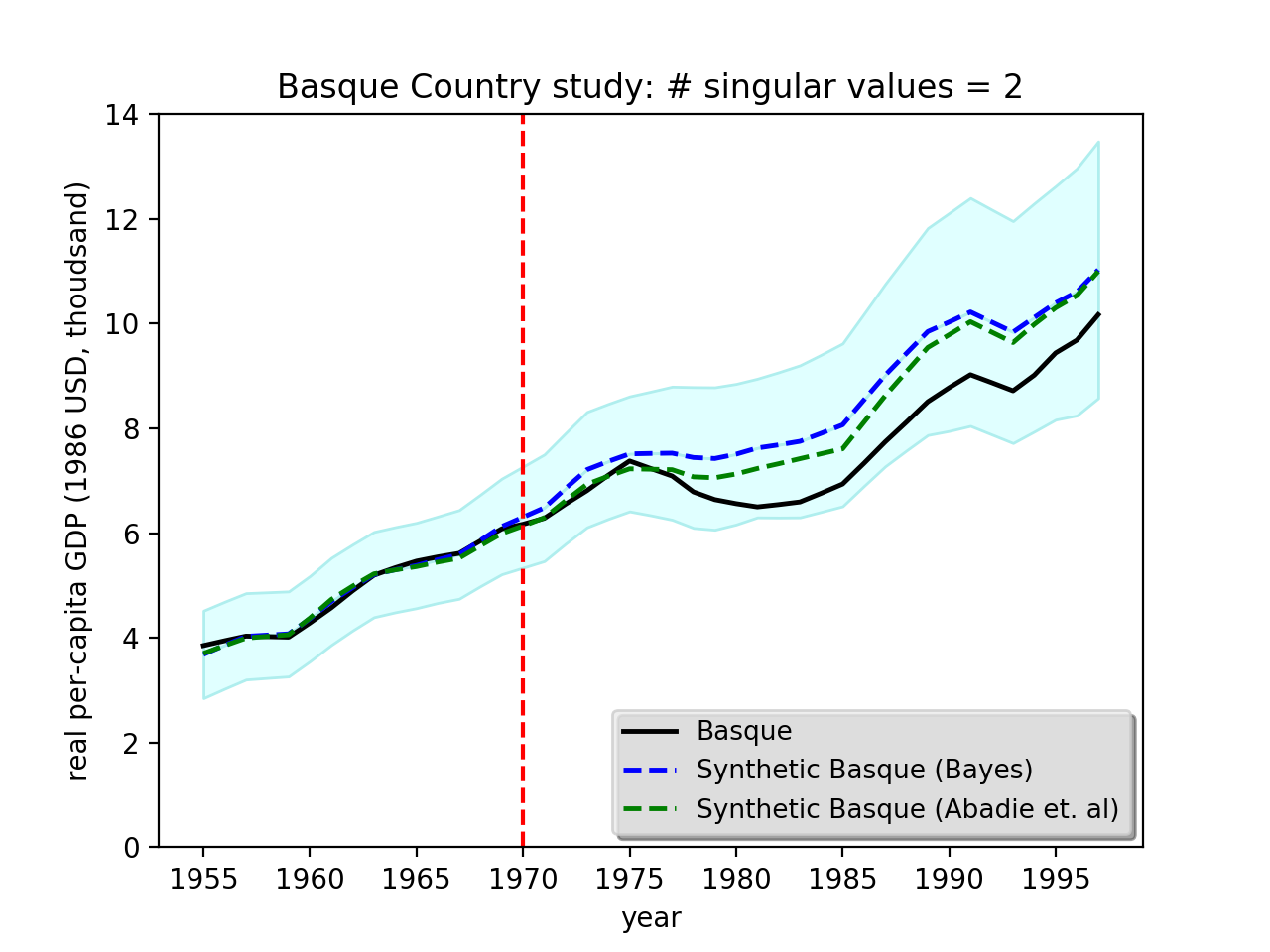}
		\caption{Top two singular values.}
		\label{fig:basque_bayes2}
	\end{subfigure}
	\begin{subfigure}[b]{0.325\textwidth}
		\centering
		\includegraphics[width=\textwidth]{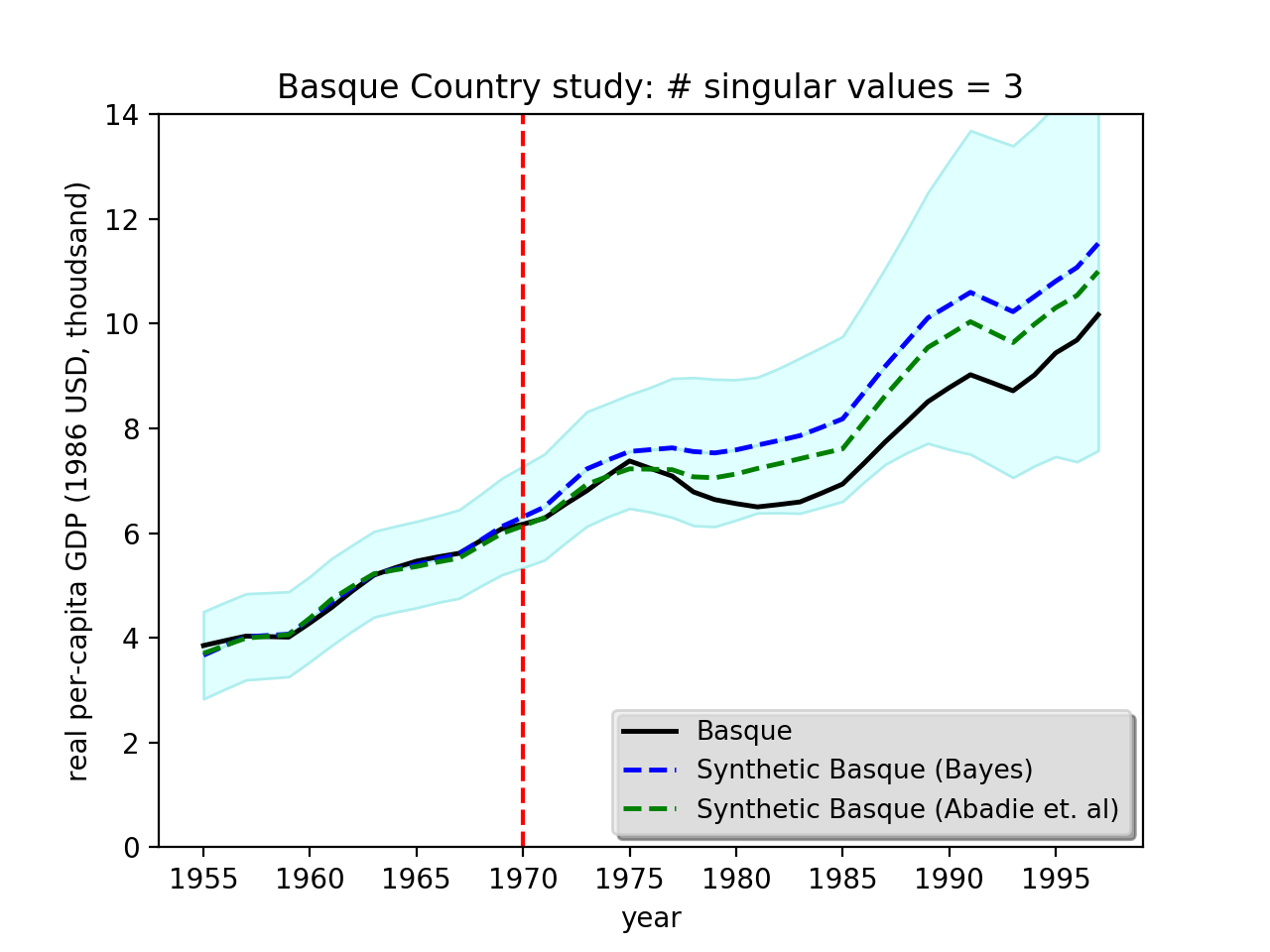}
		\caption{Top three singular values.}
		\label{fig:basque_bayes3}
	\end{subfigure}

	\centering
	\begin{subfigure}[b]{0.325\textwidth}
		\centering
		\includegraphics[width=\textwidth]{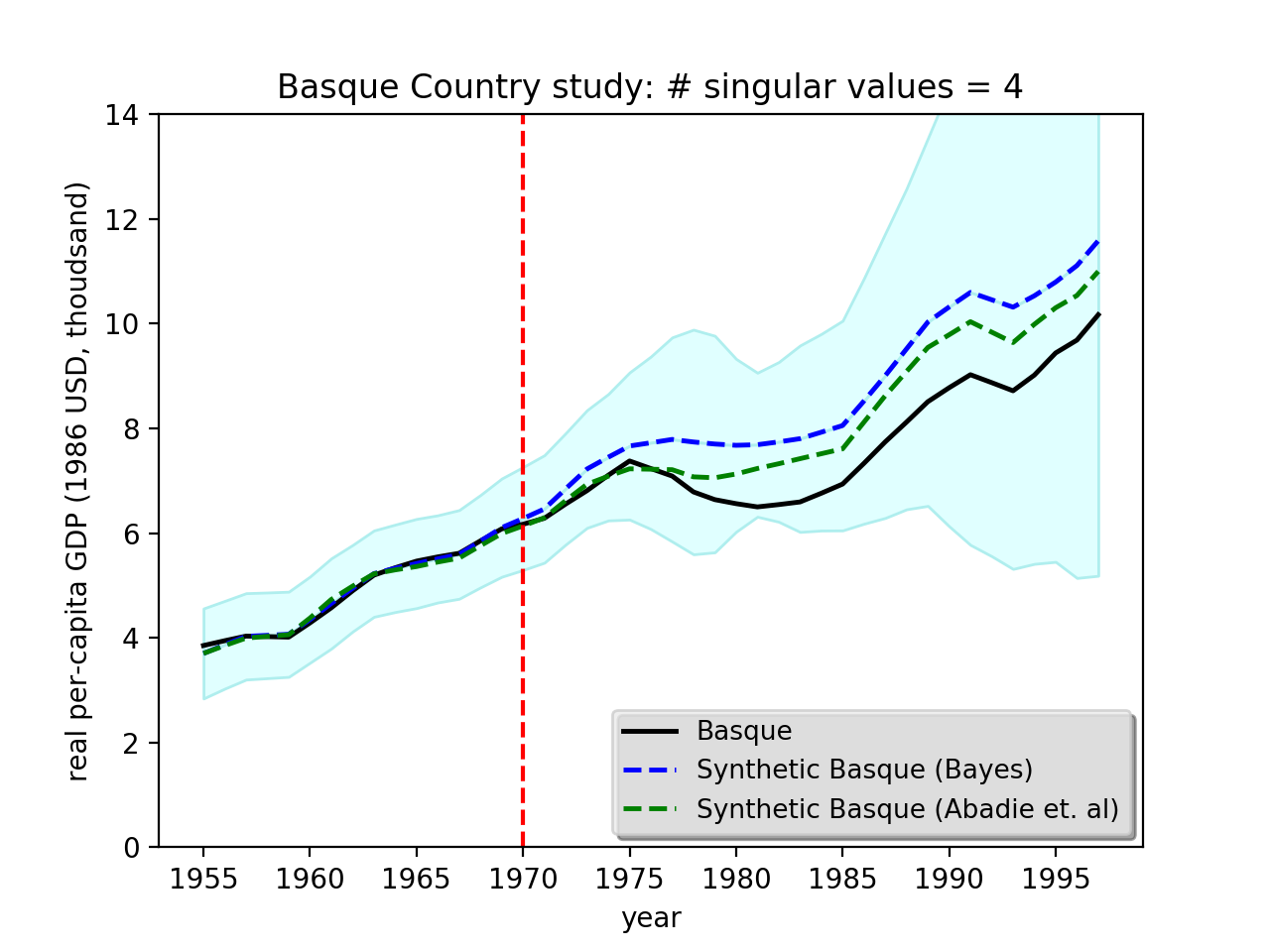}
		\caption{Top four values.}
		\label{fig:basque_bayes4}
	\end{subfigure}
	\begin{subfigure}[b]{0.325\textwidth}
		\centering
		\includegraphics[width=\textwidth]{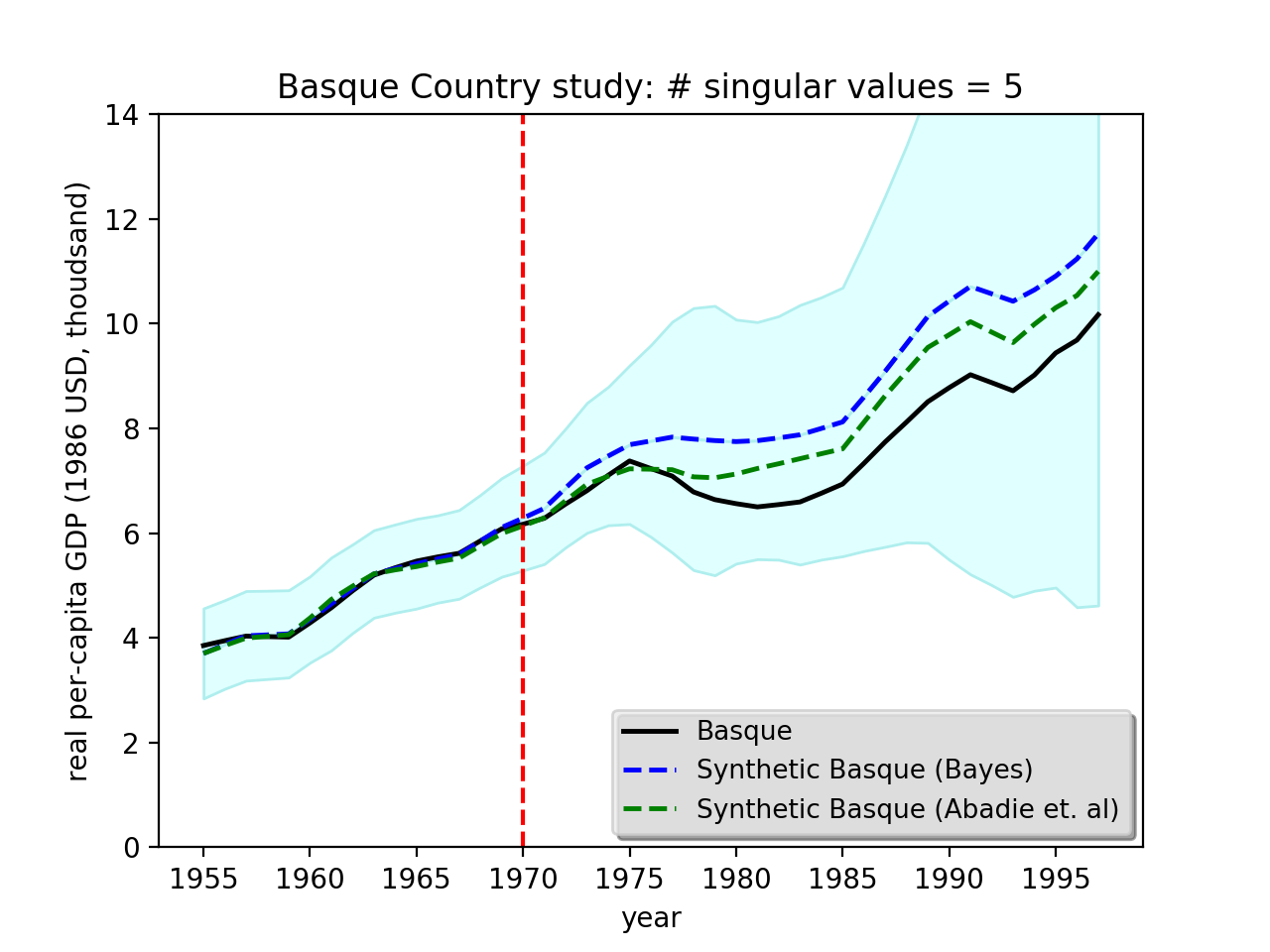}
		\caption{Top five singular values.}
		\label{fig:basque_bayes5}
	\end{subfigure}
	\begin{subfigure}[b]{0.325\textwidth}
		\centering
		\includegraphics[width=\textwidth]{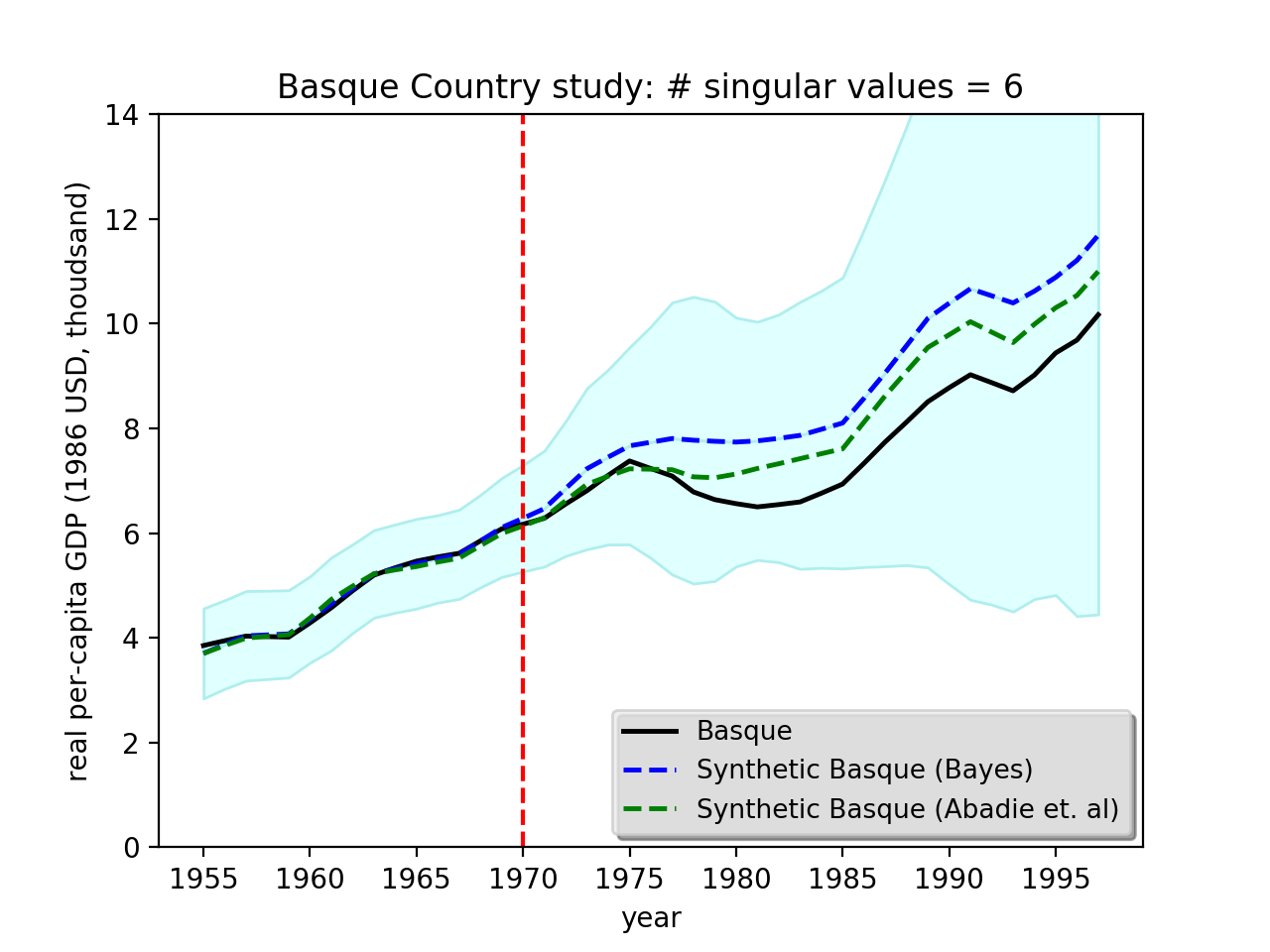}
		\caption{Top six singular values.}
		\label{fig:basque_bayes6}
	\end{subfigure}
	
	\caption{Trends in per-capita GDP between Basque Country vs. synthetic Basque Country. }
\end{figure}

\subsection{California Anti-tobacco Legislation}
We study the impact of California's anti-tobacco legislation, Proposition 99, on the per-capita cigarette consumption of California. In 1988, California introduced the first modern-time large-scale anti-tobacco legislation in the United States \cite{abadie1}. To analyze the effect of California's anti-tobacco legislation, we use the annual per-capita cigarette consumption at the state-level for all 50 states in the United States, as well as the District of Columbia, from 1970-2015. Similar to the previous case study, \cite{abadie3} uses 6 additional observable covariates per state, e.g. retail price, beer consumption per capita, and percentage of individuals between ages of 15-24, to predict their synthetic California. Furthermore, \cite{abadie3} discarded 12 states from the donor pool since some of these states also adopted anti-tobacco legislation programs or raised their state cigarette taxes, and discarded data after the year 2000 since many of the control units had implemented anti-tobacco measures by this point in time.

\textbf{Results.} As shown in Figure \ref{fig:california_a}, in the pre-intervention period of 1970-88, our control matches the observed trajectory. Post 1988, however, there is a significant divergence suggesting that the passage of Prop. 99 helped reduce cigarette consumption. 
 Similar to the Basque case-study, our estimated effect is similar to that of \cite{abadie3}. As seen in Figure \ref{fig:california_b}, our algorithm is again robust to randomly missing data. 

\begin{figure}[H]
	\centering
	\begin{subfigure}[b]{0.35\textwidth}
		\centering
		\includegraphics[width=\textwidth]{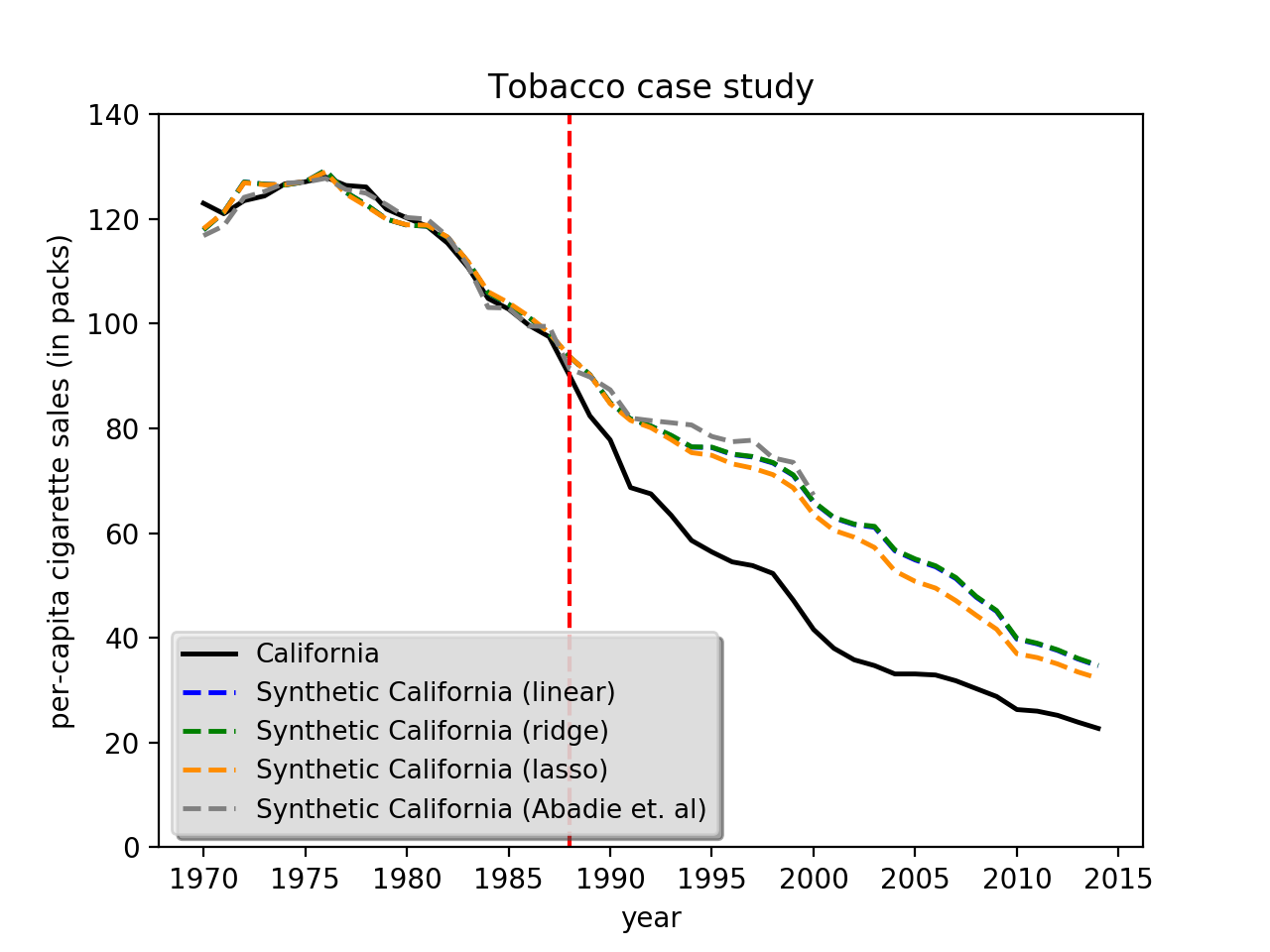}
		\caption{Comparison of methods.}
		\label{fig:california_a}
	\end{subfigure}
	\begin{subfigure}[b]{0.35\textwidth}
		\centering
		\includegraphics[width=\textwidth]{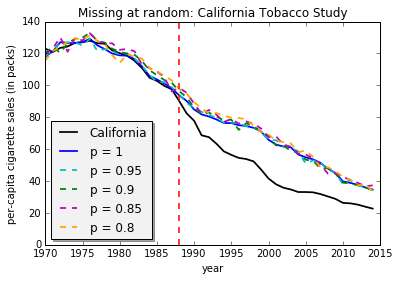}
		\caption{Missing data.}
		\label{fig:california_b}
	\end{subfigure}
	\caption{Trends in per-capita cigarette sales between California vs. synthetic California.}
\end{figure}

\textbf{Placebo tests.}
We now proceed to apply the same placebo tests to the California Prop 99 dataset. Figures \ref{fig:colorado}, \ref{fig:iowa}, and \ref{fig:wyoming} are three examples of the applied placebo tests on the remaining states (including District of Columbia) within the United States. 
\begin{figure}[H]
	\centering
	\begin{subfigure}[b]{0.325\textwidth}
		\centering
		\includegraphics[width=\textwidth]{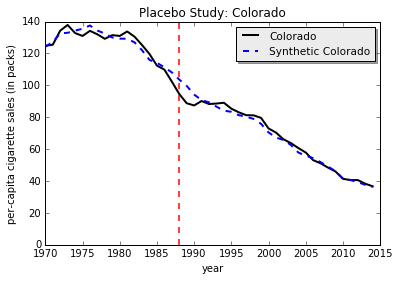}
		\caption{Colorado.}
		\label{fig:colorado}
	\end{subfigure}
	\begin{subfigure}[b]{0.325\textwidth}
		\centering
		\includegraphics[width=\textwidth]{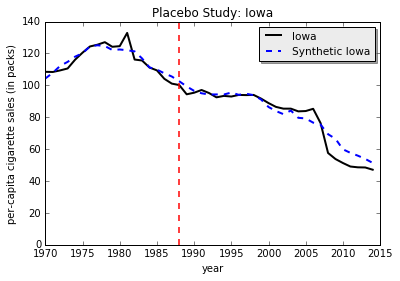}
		\caption{Iowa.}
		\label{fig:iowa}
	\end{subfigure}	
	\begin{subfigure}[b]{0.325\textwidth}
		\centering
		\includegraphics[width=\textwidth]{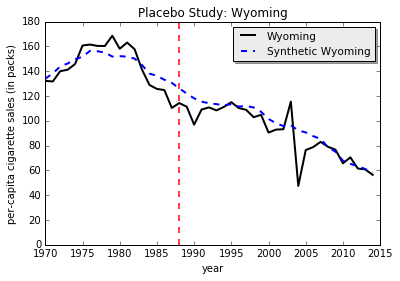}
		\caption{Wyoming.}
		\label{fig:wyoming}
	\end{subfigure}
	\caption{Placebo Study: trends in per-capita cigarette sales for Colorado, Iowa, and Wyoming.}
	\label{fig:tobacco_placebo}
\end{figure}

Finally, similar to \cite{abadie1}, we plot the differences between our estimates and the actual observations for California and all other states, individually, as placebos. Note that \cite{abadie1} excluded twelve states but we keep all states. Figure \ref{fig:cali_p1} shows the resulting plot for all states with the solid black line being California. This plot helps visualize the  extreme post-intervention divergence between the predicted means and the observed values for California. Up until about 1995, the divergence for California was clearly the most significant and consistent outlier compared to all other regions (placebo studies) lending credence to the belief that the effects of Proposition 99 were indeed significant. Refer to Figure \ref{fig:cali_p2} for the same test but with the same twelve states excludes as in \cite{abadie1}. Just like the Basque Country case study, the exclusion of states should not affect the conclusions drawn.

\begin{figure}[H]
	\centering
	\begin{subfigure}[b]{0.35\textwidth}
		\centering
		\includegraphics[width=\textwidth]{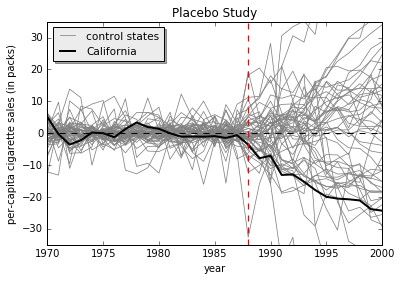}
		\caption{Includes all donors.}
		\label{fig:cali_p1}
	\end{subfigure}
	\begin{subfigure}[b]{0.35\textwidth}
		\centering
		\includegraphics[width=\textwidth]{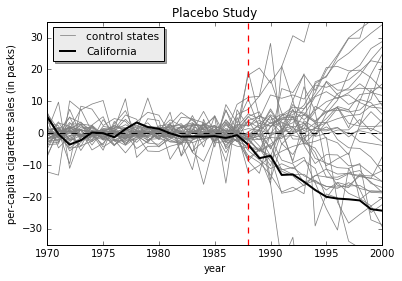}
		\caption{Excludes 12 states.}
		\label{fig:cali_p2}
	\end{subfigure}	
	\caption{Per-capita cigarette sales gaps in California and control regions.}
\end{figure}


\textbf{Bayesian approach.}
Similar to the Basque Country case study, our predictive uncertainty increases as the number of singular values used in the learning process exceeds two. In order to gain some new insight, however, we will focus our attention to the resulting figure associated with three singular values, which is particularly interesting. Specifically, we observe that our predictive means closely match the counterfactual trajectory produced by the classical synthetic control method in both the pre- and post-intervention periods (up to year 2000), and yet our uncertainty for this estimate is significantly greater than our uncertainty associated with the estimate produced using two singular values. As a result, it may be possible that the classical synthetic control method overestimated the effect of Prop. 99, even though the legislation did probably discourage the consumption of cigarettes -- a conclusion reached by both our robust approach and the classical approach.

\begin{figure}[H]

	\begin{subfigure}[b]{0.325\textwidth}
		\centering
		\includegraphics[width=\textwidth]{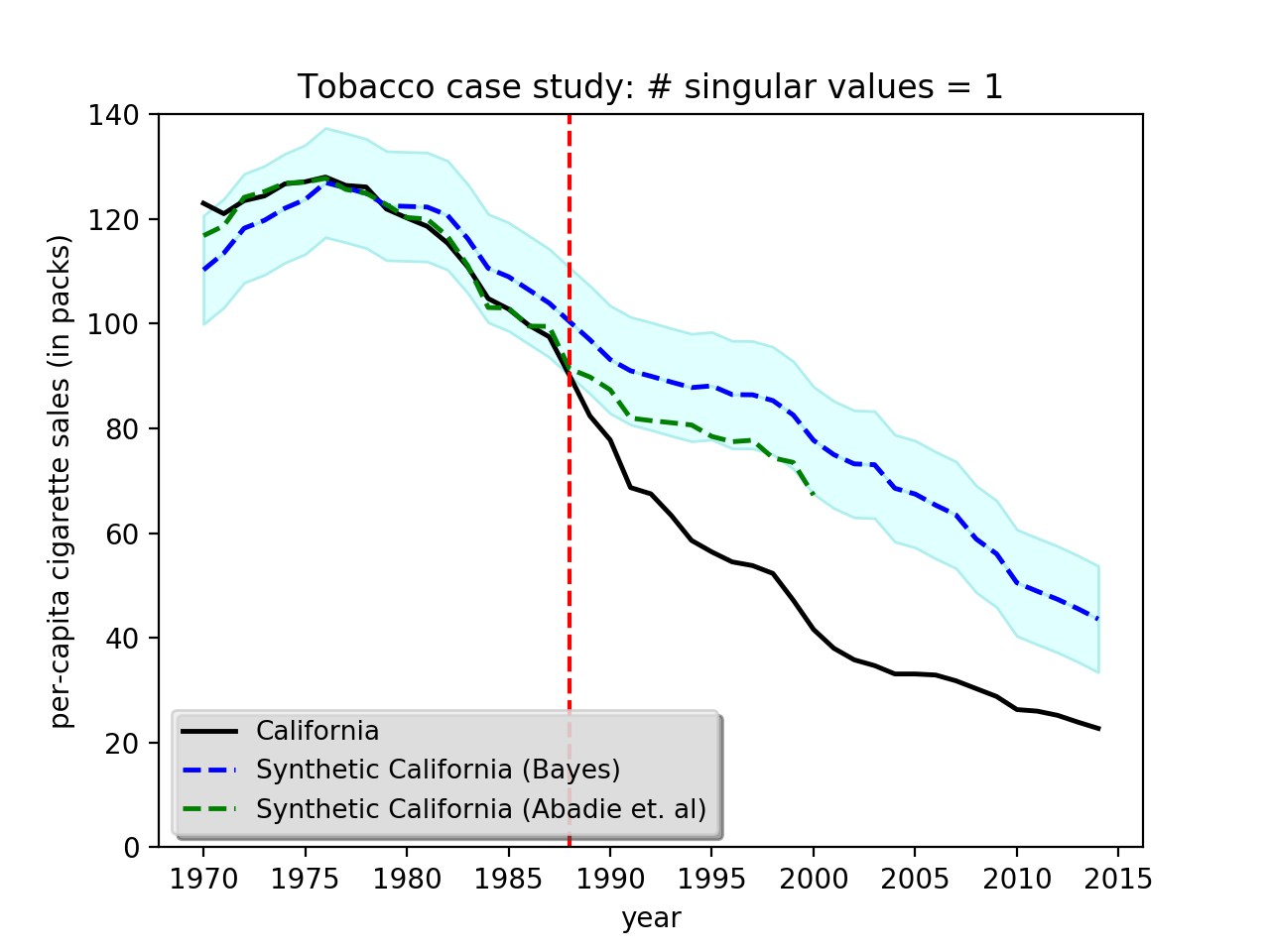}
		\caption{Top singular value.}
		\label{fig:cali_bayes11}
	\end{subfigure}
	\begin{subfigure}[b]{0.325\textwidth}
		\centering
		\includegraphics[width=\textwidth]{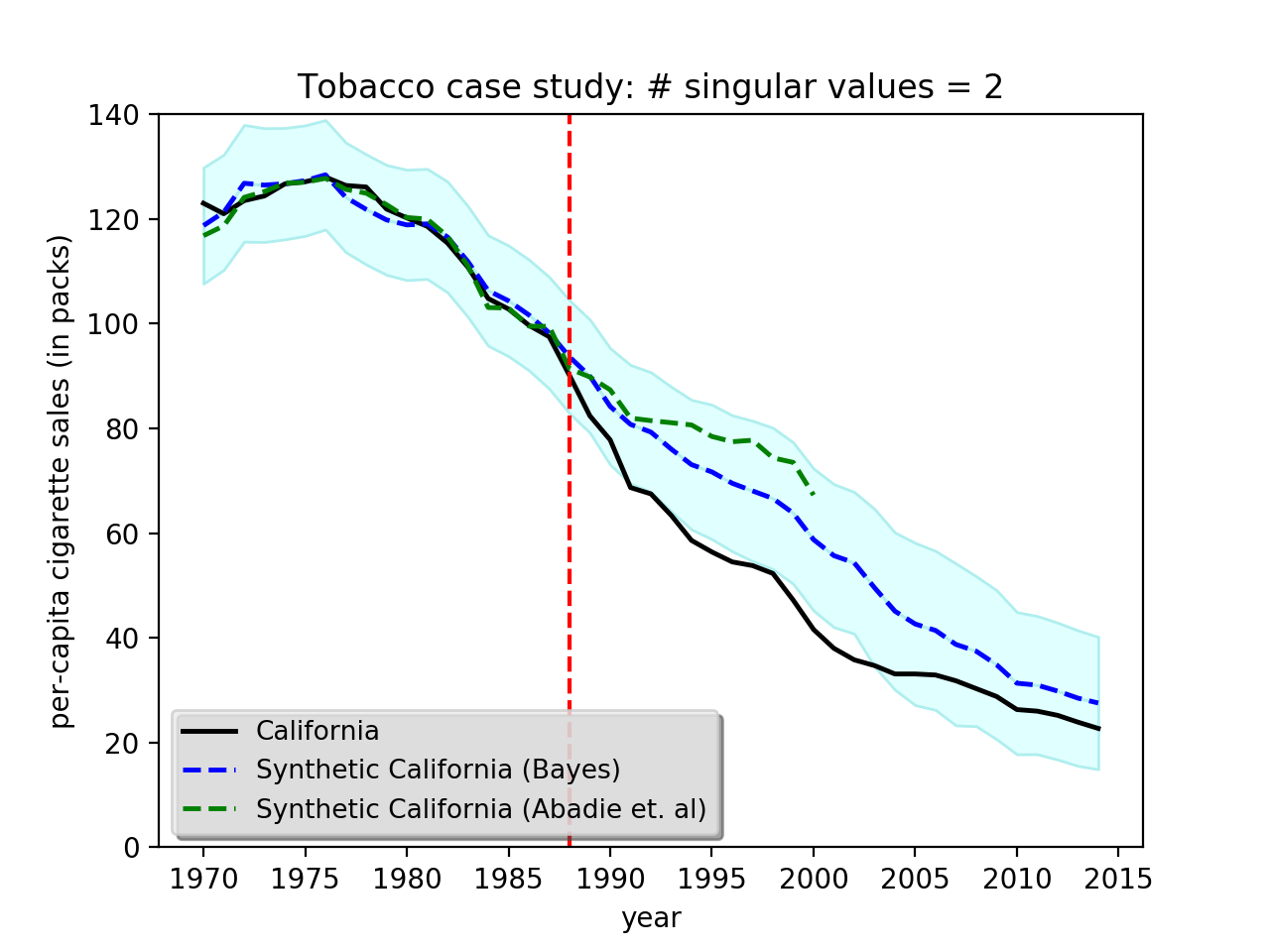}
		\caption{Top two singular values.}
		\label{fig:cali_bayes22}
	\end{subfigure}
	\begin{subfigure}[b]{0.325\textwidth}
		\centering
		\includegraphics[width=\textwidth]{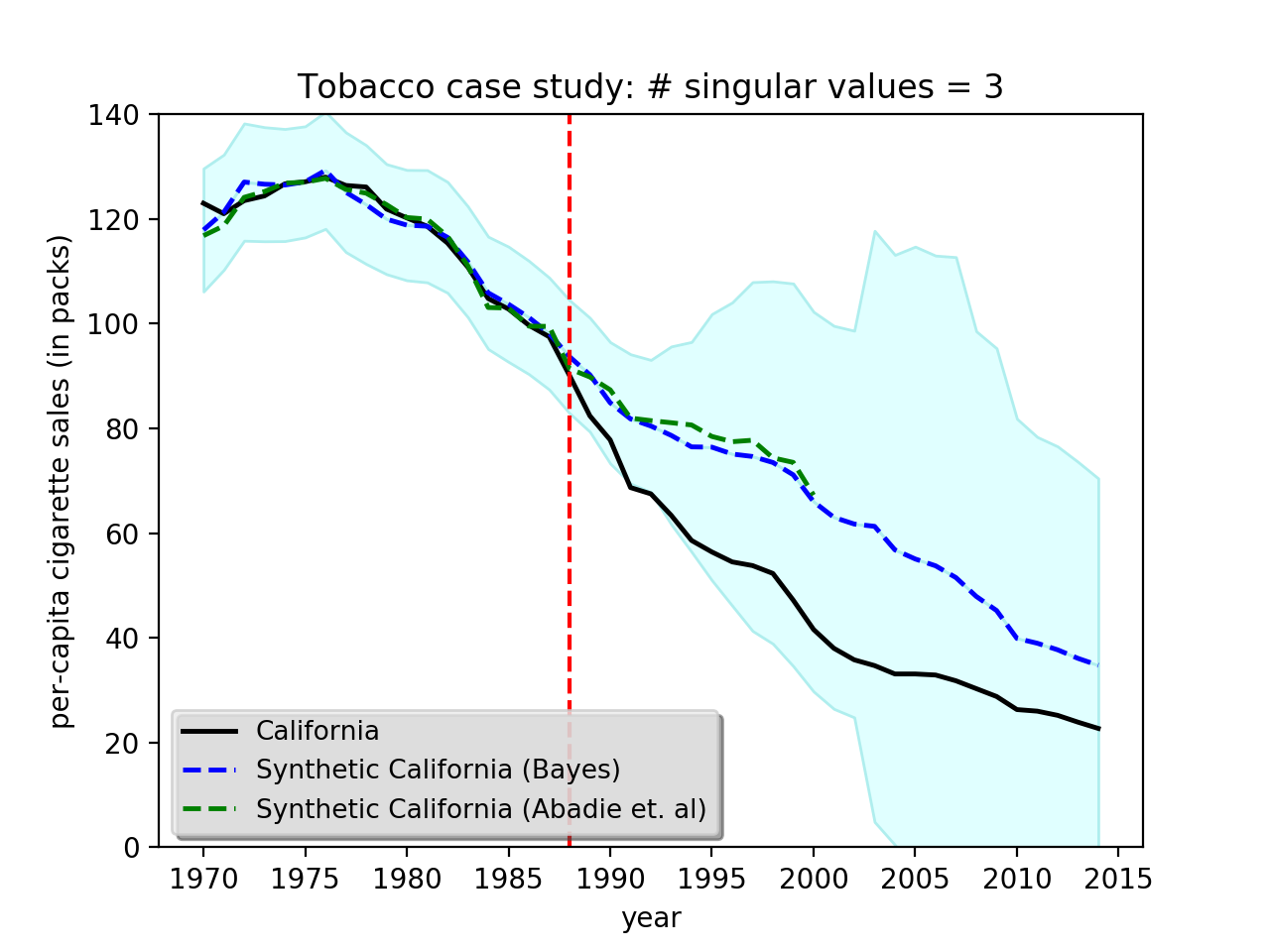}
		\caption{Top three singular values.}
		\label{fig:cali_bayes33}
	\end{subfigure}
	
	\caption{Trends in per-capita cigarette sales between California vs. synthetic California.}
\end{figure}

\begin{remark}
	We note that in \cite{abadie4}, the authors ran two robustness tests to examine the sensitivity of their results (produced via the original synthetic control method) to alterations in the estimated convex weights -- recall that the original synthetic control estimator produces a $\beta^*$ that lies within the simplex. In particular, the authors first iteratively reproduced a new synthetic West Germany by removing one of the countries that received a positive weight in each iteration, demonstrating that their synthetic model is fairly robust to the exclusion of any particular country with positive weight. Furthermore, \cite{abadie4} examined the trade-off between the original method's ability to produce a good estimate and the sparsity of the given donor pool. In order to examine this tension, the authors restricted their synthetic West Germany to be a convex combination of only four, three, two, and a single control country, respectively, and found that, relative to the baseline synthetic West Germany (composed of five countries), the degradation in their goodness of fit was moderate. 
\end{remark}

\subsection{Synthetic simulations}
We conduct synthetic simulations to establish the various properties of the estimates in both the pre- and post-intervention stages. 

\textbf{Experimental setup.}
For each unit $i \in [N]$, we assign latent feature $\theta_i$ by drawing a number uniformly at random in $[0,1]$. For each time $t \in [T]$, we assign latent variable
$\rho_t = t$. The mean value $m_{it} = f(\theta_i, \rho_t)$. In the experiments described in this section, we use the following:
\begin{align*}
f(\theta_i, \rho_t) =& \theta_i + (0.3 \cdot \theta_i \cdot \rho_t/T)*(\exp^{\rho_t/T}) + \\
				& \cos( f_1  \pi / 180) + 0.5 \sin( f_2  \pi / 180) + 
				1.5 \cos(f_3 \pi / 180) - 0.5 sin(f_4* \pi /180) \\
\end{align*}
where $f_1, f_2, f_3, f_4$ define the periodicities: $f_1 = \rho_t \bmod(360), f_2 = \rho_t \bmod(180), f_3 = 2\cdot \rho_t \bmod(360), f_4 = 2.0 \cdot \rho_t \bmod(180)$. The observed value $X_{it}$ is produced by adding i.i.d. Gaussian noise to
mean with zero mean and variance $\sigma^2$. For this set of experiments, we use $N = 100, T = 2000$, while assuming the treatment was performed at $t = 1600$. 

\begin{figure}[H]
\centering
\includegraphics[scale=0.4]{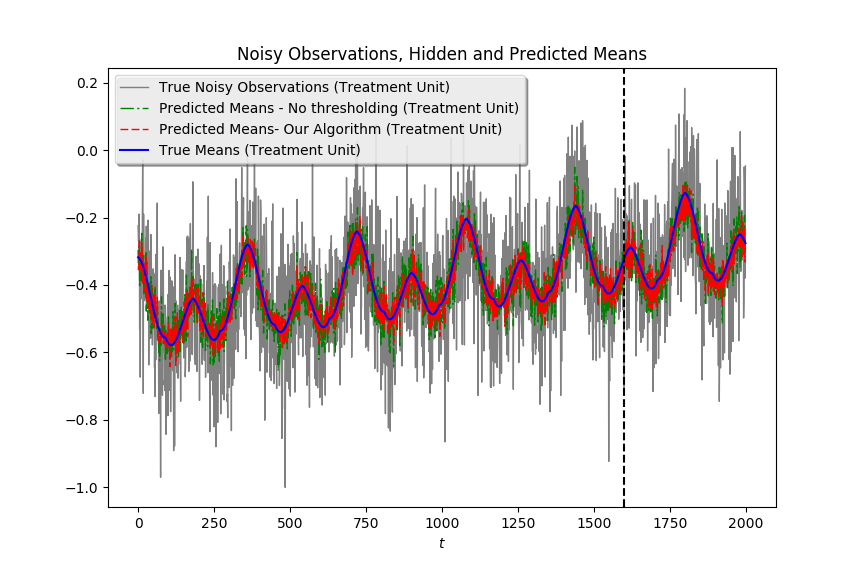}
\caption{Treatment unit: noisy observations (gray) and true means (blue) and the estimates from our algorithm (red) and one where no singular value thresholding is performed (green). The plots show all entries normalized to lie in range $[-1, 1]$. Notice that the estimates in red generated by our model are much better at estimating the true underlying mean (blue) when compared to an algorithm which performs no singular value thresholding.} \label{fig:treatment1}
\end{figure}

\begin{figure}[H]
\centering
\includegraphics[scale=0.5]{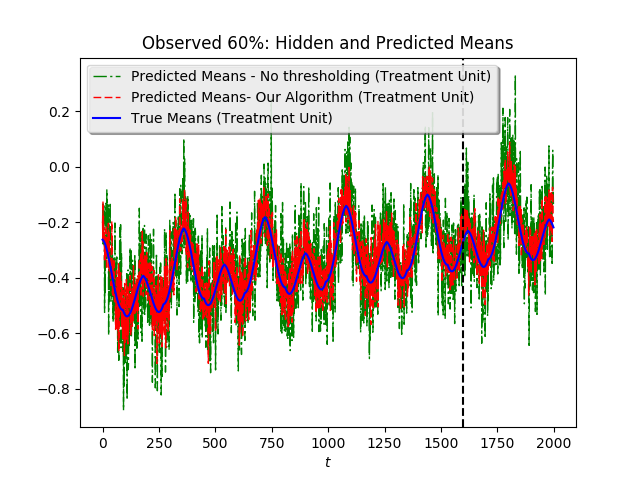}
\caption{Same dataset as shown in Figure \ref{fig:treatment1} but with 40\% data missing at random. Treatment unit: not showing the noisy observations for clarity; plotting true means (blue) and the estimates from our algorithm (red) and one where no singular value thresholding is performed (green). The plots show all entries normalized to lie in range $[-1, 1]$. } \label{fig:treatment2}
\end{figure}

\textbf{Training error approximates generalization error.}
For the first experimental study, we analyze the relationship between the pre-intervention MSE (training error) and the post-intervention MSE (generalization error). As seen in Table \ref{table:generalization}, the post-intervention MSE closely matches that of the pre-intervention MSE for varying noise levels, $\sigma^2$. Thus suggesting efficacy of our algorithm. Figures \ref{fig:treatment1} and \ref{fig:treatment2} plot the estimates
of algorithm with no missing data (Figure \ref{fig:treatment1}) and with 40\% randomly missing data (Figure \ref{fig:treatment2}) on the same underlying dataset. All entries in the plots were normalized to lie within $[-1, 1]$. These plots confirm the robustness of our algorithm. Our algorithm outperforms the algorithm with no singular value thresholding under all proportions of missing data. The estimates from the algorithm which performs no singular value thresholding (green) degrade significantly with missing data while our algorithm remains robust.

\begin{table}[H]
		\caption{Training vs. generalization error}
		\label{table:generalization}
		\centering
		\begin{tabular}{lll}
			\toprule
			\cmidrule{1-3}
			Noise   & Training error  &  Generalization error  \\
			\midrule
			3.1		&  0.48		&  0.53 	\\
			2.5		&  0.31		&  0.34 	\\
			1.9		&  0.19		&  0.22 	\\
			1.3		&  0.09		&  0.1 		\\
			0.7		&  0.027	&  0.03 	\\
			0.4		&  0.008	&  0.009 	\\
			0.1		&  0.0005	&  0.0006 	\\
			\bottomrule 
		\end{tabular}	
	\end{table}

\textbf{Benefits of de-noising.}
We now analyze the benefit of de-noising the data matrix, which is the main contribution of this work compared to the prior work. Specifically, we study the generalization
error of method using de-noising via thresholding and without thresholding as in prior work. The results summarized in Table \ref{table:threshold} show that for range of parameters the generalization error with de-noising is consistency better than that without de-noising. 

\begin{table}[H]
		\caption{Impact of thresholding}
		\label{table:threshold}
		\centering
		\begin{tabular}{lll}
			\toprule
			\cmidrule{1-3}
			Noise    & De-noising error   &  No De-noising  error\\
			\midrule
			3.1	&  0.122	&  0.365 	\\
			2.5	&  0.079	&  0.238 	\\
			1.9	&  0.046	&  0.138 	\\
			1.6	&  0.032	&  0.098 	\\
			1		&  0.013	&  0.038 	\\
			0.7	&  0.006	&  0.018 	\\
			0.4	&  0.002	&  0.005 	\\
			\bottomrule 
		\end{tabular}		
\end{table}

\textbf{Bayesian approach.}
From the synthetic simulations (figures below), we see that the number of singular values included in the thresholding process plays a crucial role in the model's prediction capabilities. If not enough singular values are used, then there is a significant loss of information (high bias) resulting in a higher MSE. On the other hand, if we include too many singular values, then the model begins to overfit to the dataset by misinterpreting noise for signal (high variance). As emphasized before, the goal is to find the simplest model that both fits the data and is also plausible, which is achieved when four singular values are employed.  

\begin{figure}[H]
	\centering
	\begin{subfigure}[b]{0.325\textwidth}
		\centering
		\includegraphics[width=\textwidth]{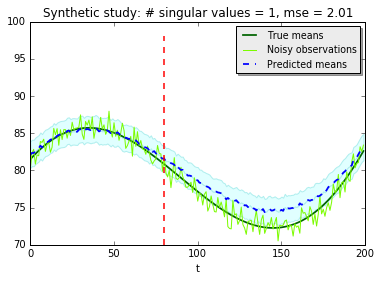}
		\caption{Top singular value.}
		\label{fig:synth_bayes1}
	\end{subfigure}
	\begin{subfigure}[b]{0.325\textwidth}
		\centering
		\includegraphics[width=\textwidth]{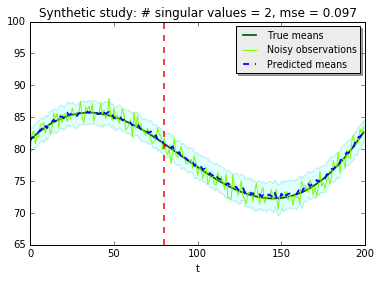}
		\caption{Top two singular values.}
		\label{fig:synth_bayes2}
	\end{subfigure}
	\begin{subfigure}[b]{0.325\textwidth}
		\centering
		\includegraphics[width=\textwidth]{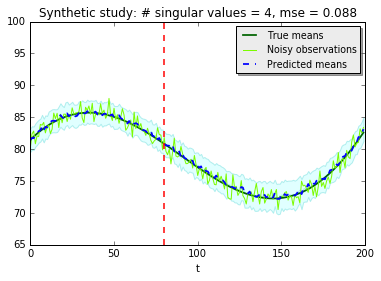}
		\caption{Top four singular values.}
		\label{fig:synth_bayes3}
	\end{subfigure}
	\begin{subfigure}[b]{0.325\textwidth}
		\centering
		\includegraphics[width=\textwidth]{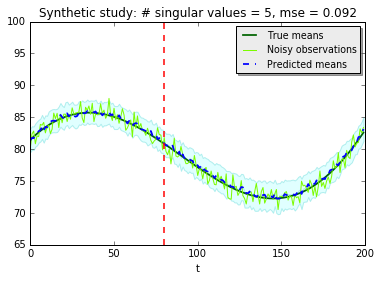}
		\caption{Top five singular values.}
		\label{fig:synth_bayes4}
	\end{subfigure}
	\begin{subfigure}[b]{0.325\textwidth}
		\centering
		\includegraphics[width=\textwidth]{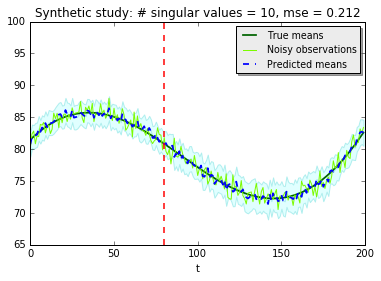}
		\caption{Top 10 singular values.}
		\label{fig:synth_bayes5}
	\end{subfigure}
	\begin{subfigure}[b]{0.325\textwidth}
		\centering
		\includegraphics[width=\textwidth]{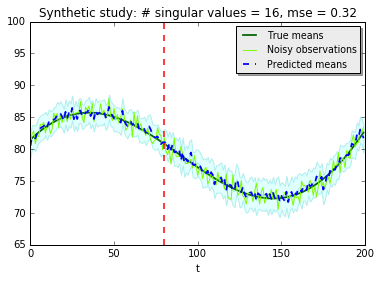}
		\caption{Top 16 singular values.}
		\label{fig:synth_bayes6}
	\end{subfigure}
\end{figure}


\section{Conclusion} \label{sec:conclusion}

The classical synthetic control method is recognized as a powerful and effective technique for causal inference for comparative case studies. In this work, we motivate a robust synthetic control algorithm, which attempts to improve on the classical method in the following regimes: (a) randomly missing data and (b) large levels of noise. We also demonstrate that the algorithm performs well even in the absence of covariate or expert information, but do \textit{not} propose ignoring information which may eliminate ``bad'' donors. Our data-driven algorithm, and its Bayesian counterpart, uses singular value thresholding to impute missing data and ``de-noise'' the observations. Once ``de-noised'', we use regularized linear regression to determine the synthetic control. Motivating our algorithm is a modeling framework, specifically the Latent Variable Model, which is a generalization of the various factor models used in related work. We establish finite-sample bounds on the MSE between the estimated ``synthetic'' control and the latent \textit{true} means of the treated unit of interest. In situations with plentiful data, we show that a simple data aggregation method can lead to an asymptotically consistent estimator. Experiments on synthetically generated data (where the \textit{truth} is known) and on real-world case-studies allow us to demonstrate the promise of our approach, which is an improvement over the classical method.

\newpage
\bibliographystyle{abbrv}
\bibliography{mainbib}

\newpage
\appendix
\section{Useful Theorems}

We present useful theorems that we will frequently employ in our proofs. \\

\begin{thm} \label{thm:singular values} {\bf Perturbation of singular values.} \\
	Let $\bA$ and $\bB$ be two $m \times n$ matrices. Let $k = \min\{m,n\}$. Let $\lambda_1,\dots, \lambda_k$ be the singular values of $\bA$ in decreasing order and repeated by multiplicities, and let $\tau_1, \dots, \tau_k$ be the singular values of $\bB$ in decreasing order and repeated by multiplicities. Let $\delta_1, \dots, \delta_k$ be the singular values of $\bA - \bB$, in any order but still repeated by multiplicities. Then,
	\begin{align*}
	\max_{1 \le i \le k} \abs{ \lambda_i - \tau_i} &\le \max_{1 \le i \le k} \abs{ \delta_i}.
	\end{align*}
\end{thm}

\noindent References for the proof of the above result can be found in \cite{usvt}, for example. \\

\begin{thm} \label{thm:cauchy} {\bf Poincar\'{e} separation Theorem.} \\
	Let $\bA$ be a symmetric $n \times n$ matrix. Let $\bB$ be the $m \times m$ matrix with $m \le n$, where $\bB = \bP^T \bA \bP$ for some orthogonal projection matrix $\bP$. If the eigenvalues of $\bA$ are $\sigma_1 \le \dots \le \sigma_n$, and those of $\bB$ are $\tau_1 \le \dots \le \tau_m$, then for all $j < m+1$,
	\[ \sigma_j \le \tau_j \le \sigma_{n-m+j}. \]	
\end{thm}
\begin{remark} 
In the case where $\bB$ is the principal submatrix of $\bA$ with dimensions $(n-1) \times (n-1)$, the above Theorem is also known as Cauchy's interlacing law. \\
\end{remark}

\begin{thm} \label{thm:bernstein} {\bf Bernstein's Inequality.} \\
	Suppose that $X_1, \dots, X_n$ are independent random variables with zero mean, and M is a constant such that $\abs{X_i} \le M$ with probability one for each $i$. Let $S := \sum_{i=1}^n X_i$ and $v := \text{Var}(S)$. Then for any $t \ge 0$,
	\begin{align*}
		\mathbb{P}(\abs{S} \ge t) &\le 2 \exp(- \dfrac{3 t^2}{6v + 2Mt} ).
	\end{align*}
\end{thm}

\begin{thm} \label{thm:hoeffding} {\bf Hoeffding's Inequality.} \\
	Suppose that $X_1, \dots, X_n$ are independent random variables that are strictly bounded by the intervals $[a_i, b_i]$. Let $S := \sum_{i=1}^n X_i$. Then for any $t > 0$, 
	\begin{align*}
		\Pb (\abs{S - \E[S]} \ge t) &\le 2 \exp( - \frac{2 t^2}{\sum_{i=1}^n (b_i - a_i)^2}). 
	\end{align*}
\end{thm}

\begin{thm} \label{thm:talagrand} {\bf Theorem 3.4 of \cite{usvt}} \\
	Take any two numbers $m$ and $n$ such that $1 \le m \le n$. Suppose that $\bA = [A_{ij}]_{1 \le i \le m, 1 \le j \le n}$ is a matrix whose entries are independent random variables that satisfy, for some $\delta^2 \in [0,1]$,
	\[ \mathbb{E}[A_{ij}] = 0, \quad \mathbb{E}[A_{ij}^2] \le \delta^2, \quad \text{and} \quad \abs{A_{ij}} \le 1 \quad \text{a.s.}  \]
	Suppose that $\delta^2 \ge n^{-1 + \zeta}$ for some $\zeta > 0$. Then, for any $\omega \in (0,1)$, 
	\[ \mathbb{P}( \norm{\bA} \ge (2 + \omega) \delta \sqrt{n}) \le C(\zeta) e^{-c\delta^2n},  \]
	where $C(\zeta)$ depends only on $\omega$ and $\zeta$, and $c$ depends only on $\omega$. The same result is true when $m = n$ and $A$ is symmetric or skew-symmetric, with independent entries on and above the diagonal, all other assumptions remaining the same. Lastly, all results remain true if the assumption $\delta^2 \ge n^{-1 + \zeta}$ is changed to $\delta^2 \ge n^{-1} (\log{n})^{6 + \zeta}$. \\
\end{thm}

\begin{remark}
The proof of Theorem \ref{thm:talagrand} can be found in \cite{usvt} under Theorem 3.4. 
\end{remark}

\section{Useful Lemmas} 
We begin by proving (and providing) a series of useful lemmas that we will frequently use to derive our desired results. \\

\begin{lemma} \label{lemma:spectral_norm}
	Suppose $\bC$ is an $m \times n$ matrix composed of an $m \times p$ submatrix $\bA$ and an $m \times (n-p)$ submatrix $\bB$, i.e., $\bC = \left[ \begin{array}{c|c}
	\bA & \bB
	\end{array}
	\right]$. Then, the spectral (operator) norms of $\bA$ and $\bB$ are bounded above by the spectral norm of $\bC$,
	\begin{align*}
	\max \{\norm{\bA}, \norm{\bB} \} &\le \norm{\bC}.
	\end{align*}
\end{lemma}

\begin{proof}
	Without loss of generality, we prove the case for $\norm{\bA} \le \norm{\bC}$, since the same argument applies for $\norm{\bB}$. By definition, 
	\begin{align*}
	\bC^T \bC &= \left[
	\begin{array}{c c}
	\bA^T \bA & \bA^T \bB \\
	\bB^T \bA & \bB^T \bB
	\end{array}
	\right].
	\end{align*}
	Let $\sigma_1, \dots, \sigma_n$ be the eigenvalues of $\bC^T \bC$ in increasing order and repeated by multiplicities. Let $\tau_1, \dots, \tau_p$ be the eigenvalues of $\bA^T \bA$ in increasing order and repeated by multiplicities. By the Poincar\'{e} separation Theorem \ref{thm:cauchy}, we have for all $j < p + 1$,
	\begin{align*}
	\sigma_j &\le \tau_j \le \sigma_{n - p + j}.
	\end{align*}
	Thus, $\tau_p \le \sigma_n$. Since the eigenvalues of $\bC^T \bC$ and $\bA^T \bA$ are the squared singular values of $\bC$ and $\bA$ respectively, we have
	\begin{align*}
	\sqrt{\tau_p} &= \norm{\bA} \le \norm{\bC} = \sqrt{\sigma_n}.
	\end{align*}
	We complete the proof by applying an identical argument for the case of $\norm{\bB}$. 
\end{proof}

\begin{lemma} \label{lemma:projection}
	Let $\bA$ be any $m$ by $n$ matrix, and let $\bA^{\dagger}$ be its corresponding pseudoinverse. Then, the matrices $\bP_1 = \bA \bA^{\dagger}$ and $\bP_2 = \bA^{\dagger} \bA$ are projection matrices. 
\end{lemma}

\begin{proof}
	We first prove that $\bP_1$ is a projection matrix. In order to show $\bP_1$ is a projection matrix, we must demonstrate that $\bP_1$ satisfies two properties: namely, (1) $\bP_1$ is symmetric, i.e. $\bP_1^T = \bP_1$, and (2) $\bP_1$ is idempotent, i.e. $\bP_1^2 = \bP_1$. 
	
	Let $\bA = \bQ_1 \bSigma \bQ_2^T$ represent the SVD of $\bA$, with the pseudoinverse expressed as $\bA^{\dagger} = \bQ_2 \bSigma^{+} \bQ_1^T$. As a result, 
	\begin{align*}
		\bP_1 &= \bA \bA^{\dagger}
		\\ &= \bQ_1 \bSigma \bQ_2^T \bQ_2 \bSigma^+ \bQ_1^T
		\\ &= \bQ_1 \bSigma \bSigma^+ \bQ_1^T.
	\end{align*}
	Note that
	\begin{align*}
		\bP_1^T &= ( \bQ_1 \bSigma \bSigma^+ \bQ_1^T) ^T
		\\ &=  \bQ_1 \bSigma \bSigma^+ \bQ_1^T
		\\ &= \bP_1,
	\end{align*}
	which proves that $\bP_1$ is symmetric. Furthermore, 
	\begin{align*}
		\bP_1^2 &= ( \bQ_1 \bSigma \bSigma^+ \bQ_1^T) \cdot ( \bQ_1 \bSigma \bSigma^+ \bQ_1^T) 
		\\ &=  \bQ_1 \bSigma \bSigma^+ \bSigma \bSigma^+ \bQ_1^T 
		\\ &= \bQ_1 \bSigma \bSigma^+ \bQ_1^T 
		\\ &= \bP_1,
	\end{align*}
	which proves that $\bP_1$ is idempotent. The same argument can be applied for $\bP_2$. 
\end{proof}

\begin{lemma} \label{lemma:projection_eigenvalues}
	The eigenvalues of a projection matrix are $1$ or $0$. 
\end{lemma}

\begin{proof}
	Let $\lambda$ be an eigenvalue of the projection matrix $\bP$ for some eigenvector $v$. Then, by definition of eigenvalues,
	\begin{align*}
		\bP v &= \lambda v.
	\end{align*}
	However, by the idempotent property of projection matrices ($\bP^2 = \bP$), if we multiply the above equality by $\bP$ on the left, then we have
	\begin{align*}
		\bP (\bP v) &= \bP (\lambda v) 
		\\ &= \lambda^2 v.
	\end{align*}
	Since $v \neq 0$, the eigenvalues of $\bP$ can only be members $\mathbb{R}$ whereby $\lambda^2 = \lambda$. Ergo, we must have that $\lambda \in \{0, 1\}$. 
\end{proof}

\begin{lemma} \label{lemma:general_threshold}
	Let $\bA = \sum_{i=1}^{m} \sigma_i x_i y_i^T$ be the singular value decomposition of $\bA$ with $\sigma_1, \dots, \sigma_m$ in decreasing order and with repeated multiplicities. For any choice of $\mu \ge 0$, let $S = \{i : \sigma_i \ge \mu \}$. Define
	\begin{align*}
		\bhB &= \sum_{i \in S} \sigma_i x_i y_i^T.
	\end{align*}
	Let $\tau_1, \dots, \tau_m$ be the singular values of $\bB$ in decreasing order and repeated by multiplicities, with $\tau^* = \max_{i \notin S} \tau_i$. Then
	\begin{align*}
		\norm{ \bhB - \bB } &\le \tau^* + 2\norm{\bA - \bB}. 
	\end{align*}
\end{lemma}

\begin{proof}
	By Theorem \ref{thm:singular values}, we have that $\sigma_i \le \tau_i + \norm{\bA - \bB}$ for all $i$. Applying triangle inequality, we obtain
	\begin{align*}
		\norm{ \bhB - \bB} &\le \norm{\bhB - \bA} + \norm{\bA - \bB}
		\\ &= \max_{i \notin S} \sigma_i + \norm{\bA - \bB}
		\\ &\le \max_{i \notin S} \Big(\tau_i + \norm{\bA - \bB} \Big) + \norm{\bA - \bB}
		\\ &= \tau^* + 2 \norm{\bA - \bB}. 
	\end{align*} 
\end{proof}

\begin{lemma} \label{lemma:prescription_threshold}
	Let $\bA = \sum_{i=1}^{m} \sigma_i x_i y_i^T$ be the singular value decomposition of $\bA$. Fix any $\delta >0$ such that $\mu = (1+\delta) \norm{\bA - \bB}$, and let $S = \{i : \sigma_i \ge \mu \}$. Define
	\begin{align*}
		\bhB &= \sum_{i \in S} \sigma_i x_i y_i^T.
	\end{align*}
	Then
	\begin{align*}
		\norm{ \bhB - \bB } &\le (2 + \delta) \norm{\bA - \bB}. 
	\end{align*}
\end{lemma}

\begin{proof}
	By the definition of $\mu$ and hence the set of singular values $S$, we have that
	\begin{align*}
		\norm{ \bhB - \bB} &\le \norm{\bhB - \bA} + \norm{\bA - \bB}
		\\ &= \max_{i \notin S} \sigma_i + \norm{\bA - \bB}
		\\ &\le (1 + \delta)\norm{\bA - \bB} + \norm{\bA - \bB}
		\\ &= (2 + \delta) \norm{\bA - \bB}. 
	\end{align*} 
\end{proof}

\begin{lemma}  \label{lemma:usvt_key_lemma} {\textbf{Lemma 3.5 of \cite{usvt}}}
	Let $\bA = \sum_{i=1}^{m} \sigma_i x_i y_i^T$ be the singular value decomposition of $\bA$. Fix any $\delta > 0$ and define $S = \{ i : \sigma_i \ge (1 + \delta) \norm{\bA - \bB} \}$ such that
	\begin{align*}
		\bhB &= \sum_{i \in S} \sigma_i x_i y_i^T.
	\end{align*} 
	Then
	\begin{align*}
		\norm{\bhB - \bB}_F &\le K(\delta) ( \norm{\bA - \bB} \norm{\bB}_*)^{1/2},
	\end{align*}
	where $K(\delta) = (4 + 2\delta) \sqrt{2 / \delta} + \sqrt{2 + \delta}$. 
\end{lemma}

\begin{proof}
	The proof can be found in \cite{usvt}. 
\end{proof}


\section{Preliminaries.}
To simplify the following exposition, we assume that $\abs{M_{ij}} \le 1$ and $\abs{X_{ij}} \le 1$. Recall that all entries of the pre-intervention treatment row are observed such that $Y_1^{-} = X_1^{-} = M_1^- + \epsilon_1^{-}$. On the other hand, every entry within the pre- and post-intervention periods for the donor units are observed independently of the other entries with some arbitrary probability $p$. Specifically, for all $2 \le i \le N$ and $j \in [T]$, we define $Y_{ij} = X_{ij}$ if $X_{ij}$ is observed, and $Y_{ij} = 0$ otherwise. Consequently, observe that for all $i > 1$ and $j$, 
\begin{align*}
	\mathbb{E}[Y_{ij}] = p M_{ij}
\end{align*}
and
\begin{align*}
	\text{Var}(Y_{ij}) &= \mathbb{E}[Y_{ij}^2] - ( \mathbb{E}[Y_{ij}] )^2
	\\ &= p \mathbb{E}[X_{ij}^2] - (p M_{ij})^2
	\\ &\le p (\sigma^2 + M_{ij}^2) - (p M_{ij})^2
	\\ &= p \sigma^2 + p M_{ij}^2 (1 - p)
	\\ &\le p \sigma^2 + p (1-p).
\end{align*}
Recall that $\hat{p}$ denotes the proportion of observed entries in $\bX$ and $\hat{\sigma}^2$ represents the (unbiased) sample variance computed from the pre-intervention treatment row \eqref{eq:sample_variance}. Given the information above, we define, for any $\omega \in (0.1, 1)$, three events $E_1, E_2,$ and $E_3$ as
\begin{align*}
	E_1 &:= \{ \abs{\hat{p} - p} \le \omega p / z \},
	\\ E_2 &:= \{ \abs{\hat{\sigma}^2 - \sigma^2} \le \omega \sigma^2 / z \},
	\\ E_3 &:= \{ \norm{\bY - p \bM} \le (2 + \omega/2) \sqrt{T q} \},
\end{align*}
where $q = \sigma^2 p + p(1-p)$; for reasons that will be made clear later, we choose $z = 60 (\frac{\sigma^2 + 1}{\sigma^2})$. By Bernstein's Inequality, we have that
\begin{align*}
	\Pb(E_1) &\ge 1 - 2e^{-c_1(N-1)Tp},
\end{align*}
for appropriately defined constant $c_1$. By Hoeffding's Inequality, we obtain
\begin{align*}
	\Pb(E_2) &\ge 1 - 2e^{-c_2 T \sigma^2}
\end{align*}
for some positive constant $c_2$. Moreover, by Theorem \ref{thm:talagrand},
\begin{align*}
	\mathbb{P}(E_3) &\ge 1 - C e^{-c_3 T q}
\end{align*}
as long as $q = \sigma^2 p + p(1-p) \ge T^{-1 + \zeta}$ for some $\zeta > 0$. In other words, 
\begin{align*}
	p (\sigma^2 + 1) &\ge p( \sigma^2 + (1-p))
	\\ & \ge T^{-1 + \zeta}.
\end{align*}
Consequently, assuming the event $E_3$ occurs, we require that $p \ge \frac{T^{-1 + \zeta}}{\sigma^2 + 1}$ for some $\zeta > 0$. 

Finally, as previously discussed, we will assume that both $N$ and $T$ grow without bound in our imputation analysis. However, in our forecasting analysis, only $T_0 \rightarrow \infty$. 


\section{Imputation Analysis}

In this section, we prove that our key de-noising procedure produces a consistent estimator of the underlying mean matrix, thereby adroitly imputing the missing entries and filtering corrupted observations within our data matrix. \\

\begin{lemma} \label{lemma:usvt_lipschitz} 
Let $\bM = [M_{ij}]$ be defined as before. Suppose $f$ is a Lipschitz function with Lipschitz constant $\mathcal{L}$ and the latent row and column feature vectors come from a compact space $K$ of dimension $d$. Then for any small enough $\delta > 0$, 
\begin{align*}
	\norm{\bM}_* &\le \delta (N-1) \sqrt{T} + C(K, d, \mathcal{L}) \sqrt{(N-1)T \delta^{-d}},
\end{align*}
where $C(K, d, L)$ is a constant that depends on $K$, $d$, and $\mathcal{L}$. 
\end{lemma}

\begin{proof}
The proof is a straightforward adaptation of the arguments from [\cite{usvt}, Lemma 3.6]; however, we provide it here for completeness. By the Lipschitzness assumption, every entry in $\bM = [M_{ij}] = [f(\theta_i, \rho_j)]$ is Lipschitz in both its arguments, space ($i$) and time ($j$). For any $\delta > 0$, it is not hard to see that one can find a finite covering $P_1 (\delta)$ and $P_2(\delta)$ of $K$ so that for any $\theta, \rho \in K$, there exists $\theta' \in P_1(\delta)$ and $\rho' \in P_2(\delta)$ such that
\begin{align*}
	\abs{ f(\theta, \rho) - f(\theta', \rho')} &\le \delta.
\end{align*}
Without loss of generality, let us consider the case where $P(\delta) = P_1(\delta) = P_2(\delta)$. For every latent row feature $\theta_i$, let $p_1(\theta_i)$ be the unique element in $P(\delta)$ that is closest to $\theta_i$. Similarly, for the latent column feature $\rho_j$, find the corresponding closest element in $P(\delta)$ and denote it by $p_2(\rho_j)$. Let $\bB = [B_{ij}]$ be the matrix where $B_{ij} = f(p_1(\theta_i), p_2(\rho_j))$. Using the arguments from above, we have that for all $i$ and $j$,
\begin{align*}
	\norm{\bM - \bB}_F^2 &= \sum_{i, j} (f(\theta_i, \rho_j) - f(p_1(\theta_i), p_2(\rho_j)))^2 ~\le (N-1) T \delta^2. 
\end{align*}
Therefore,
\begin{align*}
	\norm{\bM}_* &\le \norm{\bM - \bB}_* + \norm{\bB}_*
	\\ &\stackrel{(a)}{\le} \sqrt{N-1} \norm{\bM - \bB}_F + \norm{\bB}_*
	\\ &\le \delta (N-1) \sqrt{T} + \norm{\bB}_*,
\end{align*}
where (a) follows from the fact that $\norm{\bQ}_* \le \sqrt{ \text{rank}(\bQ)} \norm{\bQ}_F$ for any real-valued matrix $\bQ$. In order to bound the nuclear norm of $\bB$, note that (by its construction) for any two columns, say $j, j' \in [N-1]$, if $p_2(\rho_j) = p_2(\rho_j')$ then it follows that the columns of $j$ and $j'$ of $\bB$ are identical. Thus, there can be at most $\abs{P(\delta)}$ distinct columns (and rows) of $\bB$. In other words, $\text{rank}(\bB) \le \abs{P(\delta)}$. Ergo,
\begin{align*}
	\norm{\bB}_* &\le \sqrt{ \abs{P(\delta)}} \norm{\bB}_F
	\\ &\le  \sqrt{ \abs{P(\delta)}} \sqrt{(N-1)T}.
\end{align*}
Due to the Lipschitzness property of $f$ and the compactness of the latent space, it can be shown that $\abs{P(\delta)} \le C(K, d, \mathcal{L}) \delta^{-d}$ where $ C(K, d, \mathcal{L})$ is a constant that depends only on $K, d,$ and $\mathcal{L}$ (the Lipschitz constant of $f$).
\end{proof}

\begin{lemma} \label{lemma:imputation} {\em (\textbf{Theorem 1.1 of \cite{usvt}})}
	Let $\bhM$ and $\bM$ be defined as before. Suppose that $p \ge \frac{T^{-1 + \zeta}}{\sigma^2 + 1}$ for some $\zeta > 0$. Then using $\mu$ as defined in \eqref{eq:goldilocks}
	\begin{align*}
		\emph{MSE}(\bhM) &\le \dfrac{C_1 \norm{\bM}_*}{(N-1) \sqrt{T p}} + \mathcal{O}\Big( \frac{1}{(N-1)T} \Big),
	\end{align*}
	where $C_1$ is a universal positive constant. 
\end{lemma}

\begin{proof}
Let $\delta > 0$ be defined by the relation
\begin{align*}
	(1 + \delta) \norm{\bY - p \bM} &= (2 + \omega) \sqrt{T \hat{q}}, 
\end{align*}
where $\hat{q} = \hat{\sigma}^2 \hat{p} + \hat{p} (1 - \hat{p})$. Observe that if $E_1, E_2,$ and $E_3$ happen, then
\begin{align*}
	1 + \delta &\ge \dfrac{(2 + \omega) \sqrt{T(\hat{\sigma}^2 \hat{p} + \hat{p}(1 - \hat{p}))}} {(2 + \omega/2) \sqrt{T (\sigma^2 p + p(1-p))}}
	\\ &\ge \dfrac{(2 + \omega) \sqrt{1 - \omega / z} \sqrt{ (1 - \omega / z) \sigma^2 p + p (1 - p - \omega p / z)}}{(2 + \omega/2) \sqrt{ \sigma^2 p + p(1-p)}}
	\\ &= \dfrac{(2 + \omega)\sqrt{1 - \omega / z}}{2 + \omega/2}  \sqrt{ 1 - \frac{\omega}{z} \Big( \frac{\sigma^2 + p} {\sigma^2 + 1 - p} \Big) }
	\\ &\ge \dfrac{(2 + \omega)\sqrt{1 - \omega / z}}{2 + \omega/2}  \sqrt{ 1 - \frac{\omega}{z} \Big( \frac{\sigma^2 + 1} {\sigma^2} \Big) }
	\\ &= \dfrac{(2 + \omega)\sqrt{1 - \omega / z}}{2 + \omega/2}  \sqrt{ 1 - \frac{\omega}{60} }
	\\ &\ge \dfrac{2 + \omega}{2 + \omega/2} \Big( 1 - \frac{\omega}{60} \Big)
	\\ &\ge \Big( 1 + \dfrac{ \omega}{5} \Big) \Big( 1 - \frac{\omega}{60} \Big)
	\\ &\ge 1 + \frac{\omega}{5} - \frac{1}{50}. 
\end{align*}
Let $K(\delta)$ be the constant defined in Lemma \ref{lemma:usvt_key_lemma}. Since $\omega \in (0.1, 1)$, $\delta \ge \frac{10 \omega - 1}{50} > 0 $ and 
\begin{align*}
	K(\delta) &= (4 + 2 \delta) \sqrt{2 / \delta} + \sqrt{2 + \delta}
	\\ &\le 4 \sqrt{1 + \delta} \sqrt{2 / \delta} + 2 \sqrt{2 ( 1 + \delta)} + \sqrt{2 ( 1 + \delta)}
	\\ &= ( 4 \sqrt{2 / \delta} + 3 \sqrt{2} ) \sqrt{ 1 + \delta}
	\\ &\le C_1 \sqrt{ 1 + \delta}
\end{align*}
where $C_1$ is a constant that depends only on the choice of $\omega$. By Lemma \ref{lemma:usvt_key_lemma}, if $E_1, E_2$ and $E_3$ occur, then
\begin{align*}
	\norm{\hat{p} \bhM - p \bM}_F^2 &\le C_2 (1 + \delta) \norm{\bY - p\bM} \norm{p \bM}_*
	\\ &\le C_3 \sqrt{T \hat{q}} \norm{p \bM}_*
	\\ &\le C_4 \sqrt{T q} \norm{p \bM}_*
\end{align*} 
for an appropriately defined constant $C_4$. Therefore, 
	\begin{align*}
		p^2\norm{ \bhM - \bM }_F^2 &\le C_5 \hat{p}^2 \norm{ \bhM - \bM}_F^2
		\\ &\le C_5  \norm{ \hat{p} \bhM - p\bM}_F^2 +   C_5 (\hat{p} - p)^2 \norm{\bM}_F^2
		\\ &\le C_6 \sqrt{T q} \norm{p \bM}_* + C_5 (\hat{p} - p)^2 (N-1)T,
	\end{align*}
where the last inequality follows from the boundedness assumption of $\bM$. In general, since $\abs{M_{ij}}$ and $\abs{Y_{ij}} \le 1$,
\begin{align*}
	\norm{ \bhM - \bM }_F &\le \norm{\bhM}_F + \norm{\bM}_F
	\\ &\le \sqrt{\abs{S}} \norm{\bhM} + \norm{\bM}_F
	\\ &=  \frac{\sqrt{\abs{S}}}{\hat{p}} \norm{\bY} + \norm{\bM}_F
	\\ &\le (N-1)^{3/2} T \norm{\bY} + \norm{\bM}_F
	\\ &\le (N-1)^{3/2} T \sqrt{(N-1) T} + \sqrt{(N-1) T}
	\\ &\le  2 (N-1)^2 T^{3/2}.
\end{align*}
Let $E := E_1 \cap E_2 \cap E_3$. Applying DeMorgan's Law and the Union Bound,
\begin{align} \label{eq:union_bound}
	\Pb(E^c) &= \Pb(E_1^c \cup E_2^c \cup E_e^c) \nonumber
	\\ &\le \Pb(E_1^c) + \Pb(E_2^c) + \Pb(E_3^c) \nonumber
	\\ &\le C_7 e^{-c_8 T(p(N-1) + \sigma^2 + q)} \nonumber
	\\ &= C_7 e^{-c_8 \phi T},
\end{align}
where we define $\phi := p(N-1) + \sigma^2 + q$ and $C_7, c_8$ are appropriately defined. Observe that $\mathbb{E}(\hat{p} - p)^2 = \frac{p(1-p)}{(N-1)T}$. Thus, by the law of total probability and noting that $\Pb(E) \le 1$ (for appropriately defined constants), 
\begin{align*}
	\mathbb{E} \norm{\bhM - \bM}_F^2 &\le \mathbb{E} \Big[ \norm{\bhM - \bM}^2_F \given E \Big] + \mathbb{E} \Big[ \norm{\bhM - \bM}^2_F \given E^c \Big] \mathbb{P}(E^c)
	\\ &\le C_6 p^{-1} \sqrt{T q} \norm{ \bM}_* + C_5 p^{-1} (1 - p) + C_9 (N-1)^4 T^3 e^{-c_8 \phi T}
	\\ &= C_6 p^{-1/2} T^{1/2} (\sigma^2 + (1-p))^{1/2} \norm{\bM}_* + C_5 p^{-1}(1-p) + C_9 (N-1)^4 T^3 e^{-c_8 \phi T}. 
\end{align*}
Normalizing by $(N-1)T$, we obtain
\begin{align*}	
	\text{MSE}(\bhM) &\le \dfrac{C_{12} \norm{\bM}_* }{(N-1)\sqrt{Tp}} + \dfrac{C_5 (1 -p)}{(N-1)Tp} + C_{10} e^{-c_{11} \phi T}.
\end{align*}
The proof is complete assuming constants are re-named.  
\end{proof}

\subsection{Proof of Theorem \ref{thm:imputation_lowrank}} 

\begin{thm*} [\ref{thm:imputation_lowrank}] {\em (\textbf{Theorem 2.1 of \cite{usvt}})}
Suppose that $\bM$ is rank $k$. Suppose that $p \ge \frac{T^{-1 + \zeta}}{\sigma^2 + 1}$ for some $\zeta > 0$. Then using $\mu$ as defined in \eqref{eq:goldilocks}, 
	\begin{align*}
		\emph{MSE}(\bhM) &\le C_1 \sqrt{\dfrac{k}{(N-1) p}} + \mathcal{O}\Big( \frac{1}{(N-1)T} \Big),
	\end{align*}
	where $C_1$ is a universal positive constant.  
\end{thm*}

\begin{proof}
	By the low rank assumption of $\bM$, we have that
	\begin{align*}
		\norm{\bM}_* &\le \sqrt{ \text{rank}(\bM) } \norm{\bM}_F
		\\ &\le \sqrt{ k (N-1) T}. 
	\end{align*}
	The proof follows from a simple application of Lemma \ref{lemma:imputation}. 
\end{proof}

\subsection{Proof of Theorem \ref{thm:imputation_lipschitz}} 

\begin{thm*} [\ref{thm:imputation_lipschitz}] {\em (\textbf{Theorem 2.7 of \cite{usvt}})}
	Suppose $f$ is a $\mathcal{L}$-Lipschitz function. Suppose that $p \ge \frac{T^{-1 + \zeta}}{\sigma^2 + 1}$ for some $\zeta > 0$. Then using $\mu$ as defined in \eqref{eq:goldilocks}, 
	\begin{align*}
		\emph{MSE}(\bhM) &\le C(K, d, \mathcal{L}) \dfrac{(N-1)^{-\frac{1}{d+2}}}{\sqrt{p}} + \mathcal{O}\Big( \frac{1}{(N-1)T} \Big),
	\end{align*}
	where $C(K, d, \mathcal{L})$ is a constant depending on $K, d,$ and $\mathcal{L}$.
\end{thm*}

\begin{proof}
	Since $f$ is Lipschitz, we invoke Lemmas \ref{lemma:usvt_lipschitz} and \ref{lemma:imputation} and choose $\delta = (N-1)^{-1/(d+2)}$. This completes the proof. 
\end{proof}

\section{Forecasting Analysis: Pre-Intervention Regime}

Here, we will bound the pre-intervention $\ell_2$ error of our estimator in order to measure its prediction power.

\subsection{Linear Regression}
In this section, we will analyze the performance of our algorithm when learning $\beta^*$ via linear regression, i.e. $\eta = 0$. As a result, we will temporarily drop the dependency on $\eta$ in this subsection such that $\hat{\beta} = \hat{\beta}(0)$. 
To ease the notational complexity of the following Lemma \ref{lemma:mse} proof, 
we will make use of the following notations for {\bf only} in this subsection:
\begin{align}
	\bQ &:= (\bM^{-})^T
	\\ \bhQ & := (\bhM^{-})^T
\end{align}
such that
\begin{align}
	M_1^- & := \bQ \beta^*
	\\ \hat{M}_1^- & := \bhQ \hat{\beta}.
\end{align}

\begin{lemma} \label{lemma:mse}
	Suppose $Y_1^- = M_1^- + \epsilon_1^-$ with $\mathbb{E}[\epsilon_{1j}] = 0$ and $\text{Var}(\epsilon_{1j}) \le \sigma^2$ for all $j \in [T_0]$. Let $\beta^*$ be defined as in \eqref{eq:3} and let $\hat{\beta}$ be the minimizer of \eqref{eq:ls}. Then for any $\mu \ge 0$ and $\eta = 0$, 
	\begin{equation} \label{eq:x.20}
		\mathbb{E}\norm{ M_1^- - \hat{M}_1^-}^2 \le \mathbb{E} \norm{(\bM^- - \bhM^-)^T  \beta^*}^2  + 2\sigma^2 \abs{S}.
	\end{equation}
\end{lemma}

\begin{proof}
	Recall that for the treatment row, $Y_1^{-} = M_1^- + \epsilon_1^{-}$ with $M_1^- = \bQ \beta^*$. Since $\hat{\beta}$, by definition, minimizes $\norm{Y_1^{-} - \bhQ v}$ for any $v \in \mathbb{R}^{N-1}$, we subsequently have
	\begin{align*}
	\norm{M_1^- - \hat{M}_1^-}^2 &= \norm{ (Y_1^{-} - \epsilon_1^{-}) - \bhQ \hat{\beta}}^2
	\\ &= \norm{(Y_1^{-} - \bhQ\hat{\beta}) + (- \epsilon_1^{-}) }^2
	\\ &= \norm{Y_1^{-} - \bhQ \hat{\beta}}^2 + \norm{ \epsilon_1^{-}}^2 + 2 \langle -\epsilon_1^{-}, Y_1^{-}- \bhQ \hat{\beta} \rangle
	\\ &\le \norm{Y_1^{-} - \bhQ \beta^*}^2 + \norm{ \epsilon_1^{-}}^2 + 2 \langle -\epsilon_1^{-}, Y_1^{-}- \bhQ \hat{\beta} \rangle
	\\ &= \norm{(\bQ\beta^* + \epsilon_1^{-}) - \bhQ \beta^*}^2 + \norm{ \epsilon_1^{-}}^2 + 2 \langle -\epsilon_1^{-}, Y_1^{-}- \bhQ \hat{\beta} \rangle
	\\ &= \norm{(\bQ - \bhQ) \beta^* + \epsilon_1^{-}}^2 + \norm{ \epsilon_1^{-}}^2 + 2 \langle -\epsilon_1^{-}, Y_1^{-}- \bhQ \hat{\beta} \rangle
	\\ &= \norm{(\bQ - \bhQ)\beta^*}^2 + 2\norm{ \epsilon_1^{-}}^2 +   2 \langle \epsilon_1^{-}, (\bQ - \bhQ)\beta^* \rangle + 
	 2 \langle -\epsilon_1^{-}, Y_1^{-}- \bhQ \hat{\beta} \rangle.
	\end{align*}
	Taking expectations, we arrive at the inequality
	\begin{align}\label{eq:mse}
	\mathbb{E}\norm{ \hat{M}_1^- - M_1^-}^2  & \le \mathbb{E} \norm{(\bQ - \bhQ) \beta^*}^2 + 2\mathbb{E}\norm{ \epsilon_1^{-}}^2 + 2\mathbb{E}[ \langle \epsilon_1^{-}, (\bQ - \bhQ)\beta^* \rangle]   + 2\mathbb{E}[ \langle -\epsilon_1^{-}, Y_1^{-}- \bhQ \hat{\beta} \rangle].
	\end{align}
	We will now deal with the two inner products on the right hand side of equation (\ref{eq:mse}). First, observe that
	\begin{align*}
	\mathbb{E}[ \langle \epsilon_1^{-}, (\bQ - \bhQ)\beta^* \rangle] &=  \mathbb{E}[(\epsilon_1^{-})^T] \bQ \beta^* - \mathbb{E}[ (\epsilon_1^{-})^T \bhQ]\beta^*
	\\ &=  -\mathbb{E}[( \epsilon_1^{-})^T] \mathbb{E}[\bhQ]\beta^*
	\\ &= 0,
	\end{align*}
	since the additive noise terms are independent random variables that satisfy $\mathbb{E}[\epsilon_{ij}] = 0$ for all $i$ and $j$ by assumption, and $\bhQ := (\bhM^{-})^T$ depends only on the noise terms for $i \neq 1$; i.e., the construction of $\bhQ := (\bhM^{-})^T$ excludes the first row (treatment row), and thus depends solely on the donor pool. 
	
For the other inner product term, we begin by recognizing that $(\epsilon_1^{-})^T \bhQ \bhQ^{\dagger} \epsilon_1^{-}$ is a scalar random variable, which allows us to replace the random variable by its own trace. This is useful since the trace operator is a linear mapping and is invariant under cyclic permutations, i.e., $\text{tr}(\bA \bB) = \text{tr}(\bB \bA)$. As a result,
\begin{align*}
\mathbb{E}[(\epsilon_1^{-})^T \bhQ \bhQ^{\dagger}\epsilon_1^{-}] &=  \mathbb{E}[ \text{tr}((\epsilon_1^{-})^T \bhQ \bhQ^{\dagger}\epsilon_1^{-})]
	\\ &=  \mathbb{E}[ \text{tr}(\bhQ \bhQ^{\dagger}\epsilon_1^{-}(\epsilon_1^{-})^T )]
	\\ &=  \text{tr}\Big(\mathbb{E}[ \bhQ \bhQ^{\dagger}\epsilon_1^{-}(\epsilon_1^{-})^T ]\Big)
	\\ &= \text{tr}\Big(\mathbb{E}[ \bhQ \bhQ^{\dagger}] \mathbb{E}[\epsilon_1^{-}(\epsilon_1^{-})^T ]\Big)
	\\ &\le  \text{tr}\Big(\mathbb{E}[ \bhQ \bhQ^{\dagger}] \sigma^2 I \Big)
	\\ &= \sigma^2 \mathbb{E}[ \text{tr}(\bhQ \bhQ^{\dagger}) ]
	\\ &\stackrel{(a)}{=} \sigma^2 \mathbb{E}[\text{rank}(\bhQ)]
	\\ &\le \sigma^2 \abs{S},
\end{align*}
where $(a)$ follows from the fact that $\bhQ \bhQ^{\dagger}$ is a projection matrix by Lemma \ref{lemma:projection}. As a result, $\bhQ \bhQ^{\dagger}$ has rank$(\bhQ)$ eigenvalues equal to 1 and all other eigenvalues equal to 0 (by Lemma \ref{lemma:projection_eigenvalues}), and since the trace of a matrix is equal to the sum of its eigenvalues, $\text{tr}(\bhQ \bhQ^{\dagger}) = \text{rank}(\bhQ)$. Simultaneously, by the definition of $\bhQ := (\bhM^{-})^T$, we have that the rank of $\bhQ := (\bhM^-)^T$ is at most $\abs{S}$. Returning to the second inner product and recalling $\hat{\beta} = \bhQ^{\dagger} Y_1^{-}$,
	\begin{align*}
	 \mathbb{E}[\langle -\epsilon_1^{-}, Y_1^{-} - \bhQ\hat{\beta} \rangle ] 
	& \quad = \mathbb{E}[(\epsilon_1^{-})^T \bhQ \hat{\beta}] -\mathbb{E}[(\epsilon_1^{-})^T Y_1^{-}]
	\\ & \quad =  \mathbb{E}[(\epsilon_1^{-})^T \bhQ \bhQ^{\dagger} Y_1^{-}] -\mathbb{E}[(\epsilon_1^{-})^T]M_1^- - \mathbb{E}[(\epsilon_1^{-})^T \epsilon_1^{-}]
	\\ & \quad  = \mathbb{E}[(\epsilon_1^{-})^T \bhQ \bhQ^{\dagger}]M_1^- + \mathbb{E}[(\epsilon_1^{-})^T \bhQ \bhQ^{\dagger}\epsilon_1^{-}]  - \mathbb{E}[(\epsilon_1^{-})^T \epsilon_1^{-}]
	\\ &\quad  \stackrel{(a)}{=} \mathbb{E}[(\epsilon_1^{-})^T] \mathbb{E}[\bhQ \bhQ^{\dagger}]M_1^- + \mathbb{E}[(\epsilon_1^{-})^T \bhQ \bhQ^{\dagger}\epsilon_1^{-}]  - \mathbb{E}[(\epsilon_1^{-})^T \epsilon_1^{-}]
	\\ &\quad  = \mathbb{E}[(\epsilon_1^{-})^T \bhQ \bhQ^{\dagger}\epsilon_1^{-}] - \mathbb{E}\norm{ \epsilon_1^{-}}^2
	\\ &\quad \le \sigma^2 \abs{S} - \mathbb{E}\norm{ \epsilon_1^{-}}^2,
\end{align*}
where $(a)$ follows from the same independence argument used in evaluating the first inner product. Finally, we incorporate the above results to (\ref{eq:mse}) to arrive at the inequality
	\begin{align*}
	\mathbb{E}\norm{ \hat{M}_1^- - M_1^-}^2 &\le \mathbb{E} \norm{(\bQ - \bhQ) \beta^*}^2 + 2\mathbb{E} \norm{\epsilon_1^{-}}^2 + 2(\sigma^2 \abs{S} - \mathbb{E} \norm{ \epsilon_1^{-}}^2)
	\\ &=  \mathbb{E} \norm{(\bQ - \bhQ) \beta^*}^2  + 2\sigma^2 \abs{S}.
	\end{align*}
\end{proof}

\begin{lemma} \label{lemma:mse_linear}
For $\eta = 0$ and any $\mu \ge 0$, the pre-intervention error of the algorithm can be bounded as
	\begin{align} 
	\MSE(\hat{M}^-_1) & \le \dfrac{C_1}{p^2T_0} \mathbb{E}  \Big( \lambda^* + \norm{ \bY - p\bM} + \norm{ (\hat{p} - p) \bM^-} \Big)^2   + \dfrac{2\sigma^2 \abs{S}}{T_0} 
	+ C_2 e^{-cp(N-1)T}.
	\end{align}
	Here, $\lambda_1, \dots, \lambda_{N-1}$ are the singular values of $p\bM$ in decreasing order and repeated by multiplicities, with $\lambda^* = \max_{i \notin S} \lambda_i$; $C_1, C_2$ and $c$ are universal positive constants.
\end{lemma}

\begin{proof}
Recall that $E_1 := \{ \abs{\hat{p} - p} \le \frac{\omega p}{z} \}$ for some choice of $\omega \in (0.1,1)$. Thus, under the event $E_1$,
\begin{align*}
	p\norm{ \bhM^- - \bM^- } &\le C_1 \hat{p} \norm{ \bhM^- - \bM^-}
	\\ &\le C_1 \Big( \norm{ \hat{p} \bhM^- - p\bM^-} +   \norm{ (\hat{p} - p) \bM^-} \Big)
	\\ &\stackrel{(a)}{\le} C_1 \Big( \norm{ \hat{p} \bhM - p\bM} +   \norm{ (\hat{p} - p) \bM^-} \Big)
	\\ &\stackrel{(b)}{\le} C_1 \Big( \lambda^* + 2 \norm{\bY - p\bM} + \norm{ (\hat{p} - p) \bM^- } \Big)
\end{align*}
where (a) follows from Lemma \ref{lemma:spectral_norm} and (b) follows from Lemma \ref{lemma:general_threshold}. In general, since $\abs{M_{ij}}$ and $\abs{Y_{ij}} \le 1$,
\begin{align} \label{eq:reference_again}
	\norm{ \bhM^- - \bM^- } &\stackrel{(a)}{\le} \norm{\bhM} + \norm{\bM^-}  \nonumber
	\\ &= \dfrac{1}{\hat{p}} \norm{\bY} + \norm{\bM^-}  \nonumber
	\\ &\le (N-1)T \norm{\bY} + \norm{\bM^-}  \nonumber
	\\ &\le (N-1)T \sqrt{(N-1) T} + \sqrt{(N-1) T_0}   \nonumber
	\\ &\le  2 ((N-1)T)^{3/2}.
\end{align}
(a) follows from a simple application of Lemma \ref{lemma:spectral_norm} and the triangle inequality of operator norms. 
By the law of total probability and $\mathbb{P}(E_1) \leq 1$,
\begin{align*}
	\mathbb{E} \norm{ (\bhM^{-} - \bM^-)^T \beta^*}^2 
	&\le \mathbb{E} \Big[ \norm{(\bhM^{-} - \bM^-)^T \beta^*}^2 \given E_1 \Big] + \mathbb{E} \Big[\norm{(\bhM^{-} - \bM^-)^T \beta^*}^2 \given E_1^c \Big] \mathbb{P}(E_1^c)
	\\ &\stackrel{(a)}{\le} \mathbb{E} \Big[ \norm{\bhM^{-} - \bM^-}^2 \given E_1 \Big] \norm{\beta^*}^2 + \mathbb{E} \Big[\norm{\bhM^{-} - \bM^-}^2 \given E_1^c \Big] \norm{\beta^*}^2 \mathbb{P}(E_1^c)
	\\ &\le \dfrac{C_2}{p^2} \mathbb{E} \Big[ \Big( \lambda^* + 2 \norm{\bY - p\bM} + \norm{ (\hat{p} - p) \bM^- } \Big)^2 \given E_1 \Big] + C_3((N-1)T)^{3/2} e^{-cp(N-1)T},
\end{align*}
where (a) follows because the spectral norm is an induced norm, and the last inequality makes use of the results from above. Note that $C_2$ and $C_3$ are appropriately defined to depend on $\beta^*$. Moreover, for any non-negative valued random variable $X$ and event $E$ with $\mathbb{P}(E) \geq 1/2$,
\begin{align}\label{eq:trick}
	\mathbb{E}[X \given E] & \leq \frac{\mathbb{E}[X]}{\mathbb{P}(E)}   ~\le 2 \mathbb{E}[X].
\end{align}
Using the fact that $\mathbb{P}(E_1) \geq 1/2$ for large enough $T, N$, we apply Lemma \ref{lemma:mse} to obtain (with appropriately defined constants $C_4, C_5, c_6$)
\begin{align} \label{eq:x.3}
	\text{MSE}(\hat{M}_1^-) &\le \dfrac{1}{T_0} \mathbb{E} \norm{ (\bM^- - \bhM^-)^T \beta^*}^2 + \dfrac{2 \sigma^2 \abs{S}}{T_0} \nonumber
	\\ &\le \dfrac{C_4}{p^2 T_0} \mathbb{E} \Big( \lambda^* +  \norm{\bY - p\bM} + \norm{ (\hat{p} - p) \bM^- } \Big)^2 + \dfrac{2 \sigma^2 \abs{S}}{T_0} + C_5 e^{-c_6p(N-1)T}.
\end{align}
The proof is completed assuming we re-label constants $C_4, C_5, c_6$ as $C_1, C_2,$ and $c$, respectively. 
\end{proof}

\subsection{Ridge Regression}
In this section, we will prove our results for the ridge regression setting, i.e. $\eta > 0$. Let us begin by deriving the closed form expression of $\hat{\beta}(\eta)$. 

\noindent
\textbf{Derivation of $\hat{\beta}(\eta)$.}
We derive the closed form solution for $\hat{\beta}(\eta)$ under the new objective function with the additional complexity penalty term:
\begin{align*} 
	 \norm{ Y_1^{-} - (\bhM^{-})^T v}^2 + \eta \norm{v}^2 &= (Y_1^{-})^T Y_1^{-} - 2 v^T \bhM^{-} Y_1^{-} + v^T \bhM^{-} (\bhM^{-})^T v + \eta v^T v. 
\end{align*}
Setting the gradient of the above expression to zero and solving for $v$, we obtain
\begin{align*}
	\nabla_v \Bigg\{ \norm{ Y_1^{-} - (\bhM^{-})^T v}^2 + \eta \norm{v}^2 \Bigg\}_{v = \hat{\beta}(\eta)} 
	&= -2 \bhM^{-} Y_1^{-} + 2 \bhM^{-} (\bhM^{-})^T v + 2 \eta v = 0.
\end{align*}
Therefore, 
\begin{align*}
	 \hat{\beta}(\eta) &= \Big( \bhM^{-} (\bhM^{-})^T + \eta \bI \Big)^{-1} \bhM^{-} Y_1^{-}.
\end{align*}
\begin{remark}
To ease the notational complexity of the following Lemmas \ref{lemma:projection_matrix_regularization} and \ref{lemma:mse_ridge} proofs, we will make use of the following notations for {\bf only} this derivation: Let
\begin{align}
	\bQ &:= (\bM^{-})^T
	\\ \bhQ & := (\bhM^{-})^T
\end{align}
such that
\begin{align}
	M_1^- & := \bQ \beta^*
	\\ \hat{M}_1^- & := \bhQ \hat{\beta}.
\end{align}
\end{remark}

\begin{lemma} \label{lemma:projection_matrix_regularization}
	Let $\bP_{\eta} = \bhQ (\bhQ^T \bhQ + \eta \bI)^{-1} \bhQ^T$ denote the projection matrix under the quadratic regularization setting. Then, the non-zero singular values of $\bP_{\eta}$ are $s_{i}^2 / (s_{i}^2 + \eta)$ for all $i \in S$. 
\end{lemma}

\begin{proof}
Recall that the singular values of $\bY$ are $s_{i}$, while the singular values of $\bhQ$ are those $s_i \ge \mu$. Let $\bhQ = \bU \bSigma \bV^T$ be the singular value decomposition of $\bhQ$. Since $\bV \bV^T = \bI$, we have that
\begin{align*}
	\bP_{\eta} &= \bhQ (\bhQ^T \bhQ + \eta \bI)^{-1} \bhQ^T
	\\ &= \bU \bSigma \bV^T (\bV \bSigma^2 \bV^T + \eta \bI)^{-1} V \bSigma \bU^T
	\\ &= \bU \bSigma \bV^T (\bV \bSigma^2 \bV^T + \eta \bV \bV^T)^{-1} \bV \bSigma \bU^T
	\\ &= \bU \bSigma \bV^T \bV (\bSigma^2 + \eta \bI)^{-1} \bV^T \bV \bSigma \bU^T
	\\ &= \bU \bSigma (\bSigma^2 + \eta \bI)^{-1} \bSigma \bU^T
	\\ &= \bU \bD \bU^T,
\end{align*}
where
\begin{align*}
	\bD &= \text{diag}\Bigg(\dfrac{s_1^2}{s_1^2 + \eta}, \dots, \dfrac{s_{\abs{S}}^2}{s_{\abs{S}}^2 + \eta}, 0, \dots, 0 \Bigg).
\end{align*}
\end{proof}

\begin{lemma} \label{lemma:mse_ridge_intermediate}
	Suppose $Y_1^- = M_1^- + \epsilon_1^-$ with $\mathbb{E}[\epsilon_{1j}] = 0$ and $\text{Var}(\epsilon_{1j}) \le \sigma^2$ for all $j \in [T_0]$. Let $\beta^*$ be defined as in \eqref{eq:3}, i.e. $M_1^- = (\bM^{-})^T \beta^*$, and let $\hat{\beta}(\eta)$ be the minimizer of \eqref{eq:ls}. Then for any $\mu \ge 0$ and $\eta > 0$, 
	\begin{equation} \label{eq:x.21}
		\mathbb{E}\norm{ M_1^- - \hat{M}_1^-}^2 \le \mathbb{E} \norm{(\bM^- - \bhM^-)^T \beta^*}^2  + \eta \norm{\beta^*}^2 - \eta \mathbb{E}\norm{\hat{\beta}(\eta)}^2 +  2\sigma^2 \abs{S}.
	\end{equation}
\end{lemma}

\begin{proof}
The following proof is a slight modification for the proof of Lemmas \ref{lemma:mse}. In particular, observe that $\hat{\beta}(\eta)$ minimizes $\norm{Y_1^{-} - \bhQ v} + \eta \norm{v}^2 $ for any $v \in \mathbb{R}^{N-1}$. As a result, 
\begin{align*}
	& \norm{M_1^- - \hat{M}_1^-}^2 + \eta \norm{\hat{\beta}(\eta)}^2 \\
	& \quad  =   \norm{ (Y_1^{-} - \epsilon_1^{-}) - \bhQ \hat{\beta}(\eta)}^2 +  \eta \norm{\hat{\beta}(\eta)}^2
	\\ &\quad = \norm{(Y_1^{-} - \bhQ\hat{\beta}(\eta))  + (- \epsilon_1^{-}) }^2 +   \eta \norm{\hat{\beta}(\eta)}^2
	\\ &\quad = \norm{Y_1^{-} - \bhQ \hat{\beta}(\eta)}^2   + \eta \norm{\hat{\beta}(\eta)}^2 + \norm{ \epsilon_1^{-}}^2 + 2 \langle -\epsilon_1^{-}, Y_1^{-}- \bhQ \hat{\beta}(\eta) \rangle
	\\ &\quad \le \norm{Y_1^{-} - \bhQ \beta^*}^2 + \eta \norm{\beta^*}^2 + \norm{ \epsilon_1^{-}}^2  + 2 \langle -\epsilon_1^{-}, Y_1^{-}- \bhQ \hat{\beta}(\eta) \rangle
	\\ &\quad = \norm{(\bQ\beta^* + \epsilon_1^{-}) - \bhQ \beta^*}^2 + \eta \norm{\beta^*}^2 + \norm{ \epsilon_1^{-}}^2 + 2 \langle -\epsilon_1^{-}, Y_1^{-}- \bhQ \hat{\beta}(\eta) \rangle
	\\ &\quad = \norm{(\bQ - \bhQ) \beta^* + \epsilon_1^{-}}^2 + \eta \norm{\beta^*}^2 + \norm{ \epsilon_1^{-}}^2 + 2 \langle -\epsilon_1^{-}, Y_1^{-}- \bhQ \hat{\beta}(\eta) \rangle
	\\ &\quad = \norm{(\bQ - \bhQ)\beta^*}^2 + \eta \norm{\beta^*}^2 + 2\norm{ \epsilon_1^{-}}^2 + 2 \langle \epsilon_1^{-}, (\bQ - \bhQ)\beta^* \rangle  + 2 \langle -\epsilon_1^{-}, Y_1^{-}- \bhQ \hat{\beta}(\eta) \rangle
\end{align*}
Taking expectations, we have
\begin{align*}
	 & \mathbb{E}\norm{ \hat{M}_1^- - M_1^-}^2  \\
	& \quad \le \mathbb{E}  \norm{(\bQ - \bhQ)\beta^*}^2 + \eta \Big(\norm{\beta^*}^2 - \mathbb{E}\norm{\hat{\beta}(\eta)}^2\Big) + 2 \mathbb{E} \norm{ \epsilon_1^{-}}^2 + 2 \mathbb{E} \langle \epsilon_1^{-}, (\bQ - \bhQ)\beta^* \rangle + 2 \mathbb{E} \langle -\epsilon_1^{-}, Y_1^{-}- \bhQ \hat{\beta}(\eta) \rangle.
\end{align*}
As before, we have that $\mathbb{E} \langle \epsilon_1^{-}, (\bQ - \bhQ)\beta^* \rangle = 0$ by the zero-mean and independence assumptions of the noise random variables. Similarly, note that
\begin{align*}
\mathbb{E}[(\epsilon_1^{-})^T \bhQ \hat{\beta}(\eta)] 
&=  \mathbb{E}[(\epsilon_1^{-})^T \bhQ ( \bhQ^T \bhQ + \eta \bI)^{-1} \bhQ^T Y_1^{-}]
\\ &= \mathbb{E}[(\epsilon_1^{-})^T \bhQ ( \bhQ^T \bhQ + \eta \bI)^{-1} \bhQ^T] M_1^- + \mathbb{E}[(\epsilon_1^{-})^T \bhQ ( \bhQ \bhQ^T + \eta \bI)^{-1} \bhQ^T \epsilon_1^{-}]
\\ &=  \mathbb{E}[(\epsilon_1^{-})^T \bhQ ( \bhQ^T \bhQ + \eta \bI)^{-1} \bhQ^T \epsilon_1^{-}]
\\ &= \mathbb{E}[ \text{tr}( (\epsilon_1^{-})^T \bhQ ( \bhQ^T \bhQ + \eta \bI)^{-1} \bhQ^T \epsilon_1^{-})]
\\ &= \mathbb{E}[ \text{tr}( \bhQ ( \bhQ^T \bhQ + \eta \bI)^{-1} \bhQ^T \epsilon_1^{-} (\epsilon_1^{-})^T)]
\\ &= \text{tr}( \mathbb{E} [  \bhQ ( \bhQ^T \bhQ + \eta \bI)^{-1} \bhQ^T \epsilon_1^{-} (\epsilon_1^{-})^T])
\\ &= \text{tr}( \mathbb{E} [  \bhQ ( \bhQ^T \bhQ + \eta \bI)^{-1} \bhQ^T] \mathbb{E}[ \epsilon_1^{-} (\epsilon_1^{-})^T])
\\ &\le \sigma^2 \mathbb{E}[ \text{tr}(\bhQ ( \bhQ^T \bhQ + \eta \bI)^{-1} \bhQ^T )] 
\\ &\stackrel{(a)}{\le}  \sigma^2 \mathbb{E}[ \text{tr}(\bhQ \bhQ^{\dagger})] 
\\ &\stackrel{(b)}{=} \sigma^2 \text{rank}(\bhQ)
\\ &\le \sigma^2 \abs{S},
\end{align*}
where $(a)$ follows from Lemma \ref{lemma:projection_matrix_regularization}, and as before, $(b)$ follows because $\bhQ \bhQ^{\dagger}$ is a projection matrix. 
\end{proof}

\begin{lemma} \label{lemma:mse_ridge}
	For any $\eta > 0$ and $\mu \ge 0$, the pre-intervention error of the regularized algorithm can be bounded as
	\begin{align*} 
		\MSE(\hat{M}^-_1) & \le \dfrac{C_1}{p^2T_0} \mathbb{E}  \Big( \lambda^* + \norm{ \bY - p\bM} + \norm{ (\hat{p} - p) \bM^-} \Big)^2 + \dfrac{2\sigma^2 \abs{S}}{T_0}
		+ \dfrac{\eta \norm{\beta^*}^2} {T_0} + C_2 e^{-cp(N-1)T} .
	\end{align*}
	Here, $\lambda_1, \dots, \lambda_{N-1}$ are the singular values of $p\bM$ in decreasing order and repeated by multiplicities, with $\lambda^* = \max_{i \notin S} \lambda_i$; $C_1, C_2$ and $c$ are universal positive constants.
\end{lemma}

\begin{proof}
The proof follows the same arguments as that of Lemma \ref{lemma:mse_linear}. 
\end{proof}

\subsection{Combining linear and ridge regression.}
\subsubsection{Proof of Theorem \ref{thm:finite-sample}}

\begin{thm*} [\ref{thm:finite-sample}]
For any $\eta \ge 0$ and $\mu \ge 0$, the pre-intervention error of the algorithm can be bounded as
\begin{align*} 
	\MSE(\hat{M}^-_1) & \le \dfrac{C_1}{p^2T_0} \mathbb{E}  \Big( \lambda^* + \norm{ \bY - p\bM} + \norm{ (\hat{p} - p) \bM^-} \Big)^2  + \dfrac{2\sigma^2 \abs{S}}{T_0} + \dfrac{\eta \norm{\beta^*}^2}{T_0}
	+ C_2 e^{-cp(N-1)T}.
	\end{align*}
	Here, $\lambda_1, \dots, \lambda_{N-1}$ are the singular values of $p\bM$ in decreasing order and repeated by multiplicities, with $\lambda^* = \max_{i \notin S} \lambda_i$; $C_1, C_2$ and $c$ are universal positive constants.
\end{thm*}

\begin{proof}
	The proof follows from a simple amalgamation of Lemmas \ref{lemma:mse_linear} and \ref{lemma:mse_ridge}. 
\end{proof}

\subsubsection{Proof of Corollary \ref{corollary:universal_threshold}}
\begin{corollary*} [\ref{corollary:universal_threshold}]
	Suppose $p \ge \frac{T^{-1 + \zeta}}{\sigma^2 + 1}$ for some $\zeta > 0$. Let $T \le \alpha T_0$ for some constant $\alpha > 1$. Then for any $\eta \ge 0$ and using $\mu$ as defined in \eqref{eq:goldilocks}, the pre-intervention error is bounded above by

	\begin{align*}
		\emph{MSE}(\hat{M}_1^-) &\le \dfrac{C_1}{p}(\sigma^2 + (1-p)) + \mathcal{O}(1 / \sqrt{T_0}),
	\end{align*}
	where $C_1$ is a universal positive constant. 
\end{corollary*}

\begin{proof}
Since the singular value threshold $\mu = (2 + \omega) \sqrt{T \hat{q}}$, let us define $\delta$ so that
	\begin{align*}
		(1 + \delta) \norm{\bY - p \bM} &= (2 + \omega) \sqrt{T \hat{q}}, 
	\end{align*}
	where $\hat{q} = \hat{\sigma}^2 \hat{p} + \hat{p} (1 - \hat{p})$; recall that $q = \sigma^2 p + p(1-p)$. If $E_3$ happens, then we know
	that $\delta \geq 0$. 
	Therefore, assuming $E_1, E_2,$ and $E_3$ happens, Lemma \ref{lemma:prescription_threshold} states that
	\begin{align} \label{eq:thresh_1}
		\norm{ \hat{p} \bhM - p\bM} &\le (2 + \delta) \norm{\bY - p\bM} \nonumber
		\\ &\le 2 (1 + \delta) \norm{\bY - p\bM} \nonumber
		\\ &= (4 + 2 \omega) \sqrt{T \hat{q}} \nonumber
		\\ &\le C_1 \sqrt{T q}
	\end{align}
	for an appropriately defined constant $C_1$. Therefore, 
	\begin{align} \label{eq:thresh_2}
		p\norm{ \bhM^- - \bM^- } &\le C_2 \hat{p} \norm{ \bhM^- - \bM^-} \nonumber
		\\ &\le C_2 \Big( \norm{ \hat{p} \bhM^- - p\bM^-} +   \norm{ (\hat{p} - p) \bM^-} \Big) \nonumber
		\\ &\stackrel{(a)}{\le} C_2 \Big( \norm{ \hat{p} \bhM - p\bM} +   \norm{ (\hat{p} - p) \bM^-} \Big) \nonumber
		\\ &\le C_2 \Big( C_1 \sqrt{T q} + \norm{ (\hat{p} - p) \bM^-} \Big)
	\end{align}
	where (a) follows from Lemma \ref{lemma:spectral_norm}. Applying the logic that led to \eqref{eq:reference_again}, we have that, in general, 
\begin{align} 
	\norm{ \bhM^- - \bM^- } &\le 2 ((N-1)T)^{3/2}.
\end{align}
Let $E := E_1 \cap E_2 \cap E_3$. Further, using the same argument that led to \eqref{eq:union_bound}, we have
\begin{align*}
	\Pb(E^c) &\le C_3 e^{-c_4 \phi T}
\end{align*}
where we define $\phi := p(N-1) + \sigma^2 + q$ and $C_3, c_4$ are appropriately defined. Thus, by the law of total probability and noting that $\Pb(E) \le 1$,
\begin{align} \label{eq:thresh_3}
	\mathbb{E} \norm{ (\bhM^{-} - \bM^-)^T \beta^*}^2 
	&\le \mathbb{E} \Big[ \norm{(\bhM^{-} - \bM^-)^T \beta^*}^2 \given E \Big] + \mathbb{E} \Big[\norm{(\bhM^{-} - \bM^-)^T \beta^*}^2 \given E^c \Big] \mathbb{P}(E^c) \nonumber
	\\ &\stackrel{(a)}{\le} \mathbb{E} \Big[ \norm{\bhM^{-} - \bM^-}^2 \given E \Big] \norm{\beta^*}^2 + \mathbb{E} \Big[\norm{\bhM^{-} - \bM^-}^2 \given E^c \Big] \norm{\beta^*}^2 \mathbb{P}(E^c) \nonumber
	\\ &\le \dfrac{C_5}{p^2} \mathbb{E} \Big[ \Big( \sqrt{T q} + \norm{ (\hat{p} - p) \bM^- } \Big)^2 \given E \Big] + C_6((N-1)T)^{3/2} e^{-c_4 \phi T},
\end{align}
where (a) follows because the spectral norm is an induced norm and the last inequality makes use of the results from above. Note that $C_5$ and $C_6$ are appropriately defined to depend on $\beta^*$. Using the fact that $\mathbb{P}(E) \geq 1/2$ for large enough $T, N$, we apply Lemmas \ref{lemma:mse} and \ref{lemma:mse_ridge} as well as \eqref{eq:trick} to obtain (with appropriately defined constants $C_7, C_8, c_9$)
\begin{align} \label{eq:x.3}
	\text{MSE}(\hat{M}_1^-) &\le \dfrac{1}{T_0} \mathbb{E} \norm{ (\bM^- - \bhM^-)^T \beta^*}^2 + \dfrac{2 \sigma^2 \abs{S}}{T_0} + \dfrac{\eta \norm{\beta^*}^2}{T_0} \nonumber
	\\ &\le \dfrac{C_7}{p^2 T_0} \mathbb{E} \Big( \sqrt{T q} + \norm{ (\hat{p} - p) \bM^- } \Big)^2 + \dfrac{2 \sigma^2 (N-1)}{T_0} + \dfrac{\eta \norm{\beta^*}^2}{T_0} + C_8 e^{-c_9 \phi T}.
\end{align}
From Jensen's Inequality, $\mathbb{E}\abs{\hat{p} - p} \le \sqrt{\text{Var}(\hat{p})}$ where $\text{Var}(\hat{p}) = \frac{p(1-p)}{(N-1)T}$. Therefore, 
\begin{align*}
	\mathbb{E} \Big(  \sqrt{Tq}  \norm{(\hat{p} - p) \bM^-} \Big) &\le \dfrac{ q^{1/2} \sqrt{p(1-p)}}{\sqrt{N-1}} \norm{\bM^-}
	\\ &\le \sqrt{q p(1-p) T_0}.
\end{align*}
At the same time, 
\begin{align*}
	\mathbb{E} \norm{ (\hat{p} - p) \bM^-}^2 &= \mathbb{E}(\hat{p} - p)^2 \cdot \norm{ \bM^-}^2 
	\\ &\le \dfrac{p (1-p) T_0}{T}
	\\ &\le p(1-p). 
\end{align*}
Putting everything together, we arrive at the inequality
\begin{align*}
	\text{MSE}(\hat{M}_1^-) &\le \dfrac{C_7}{p^2 T_0} \Big( qT + p(1-p) + 2\sqrt{q p(1-p) T_0} \Big) + \dfrac{2 \sigma^2 (N-1)}{T_0} + \dfrac{\eta \norm{\beta^*}^2}{T_0} + C_8 e^{-c_9 \phi T}
	\\ &= \dfrac{C_{10} q}{p^2 } + \dfrac{C_7 (1-p)}{p T_0} + \dfrac{C_{11} (q(1-p))^{1/2}}{p^{3/2}\sqrt{T_0}} +  \dfrac{2 \sigma^2 (N-1)}{T_0} + \dfrac{\eta \norm{\beta^*}^2}{T_0} + C_8 e^{-c_9 \phi T}
	\\ &= \dfrac{C_{10}}{p }(\sigma^2 + (1-p)) + \mathcal{O}(1 / \sqrt{T_0}). 
\end{align*}
The proof is complete assuming we re-label $C_{10}$ as $C_1$. 
\end{proof}

\subsubsection{Proof of Theorem \ref{thm:consistency}}
\begin{thm*} [\ref{thm:consistency}]
	Fix any $\gamma \in (0,1/2)$ and $\omega \in (0.1, 1)$. Let $\Delta = T_0^{\frac{1}{2} + \gamma}$ and $\mu = (2 + \omega) \sqrt{ T_0^{2 \gamma} (\hat{\sigma}^2 \hat{p} + \hat{p}(1 - \hat{p}))}$. Suppose $p \ge \frac{T_0^{-2 \gamma}}{\sigma^2 + 1}$ is known. Then for any $\eta \ge 0$, 
	\begin{align*}
		\emph{MSE}(\hat{\bar{M}}_1^{-}) &= \mathcal{O}(T_0^{-1/2 + \gamma}).
	\end{align*}
\end{thm*}

\begin{proof}
To establish Theorem \ref{thm:consistency}, we shall follow the proof of Corollary \ref{corollary:universal_threshold}, 
using the block partitioned matrices instead. Recall that $\tau = T_0 / \Delta$ where $\Delta = T_0^{1/2 + \gamma}$. For analytical simplicity, we define the random variable 
\begin{align*}
	D_{it} = \begin{cases}
	1 & \text{w.p. } p, \\
	0 & \text{otherwise,}
	\end{cases}
\end{align*}
whose definition will soon prove to be useful. As previously described in Section \ref{sec:results}, for all $i > 1$ and $j \in [\Delta]$, we define
\begin{align*}
	\bar{X}_{ij} &= \dfrac{1}{\tau} \sum_{t \in B_j} X_{it} \cdot D_{it}
\end{align*}
and
\begin{align*}
	\bar{M}_{ij} &= \dfrac{p}{\tau} \sum_{t \in B_j} M_{it}.
\end{align*}
Let us also define $\bar{\bE}^{-} = [\bar{\epsilon}_{ij}]_{2 \le i \le N, j \le \Delta}$ with entries
\begin{align} \label{eq:avgnoise}
	\bar{\epsilon}_{ij} &= \dfrac{1}{\tau} \sum_{t \in B_j} \epsilon_{it} \cdot D_{it}.
\end{align}	
For the first row (treatment unit), since we know $p$ by assumption, we define for all $j \in [\Delta]$
\begin{align}
	\bar{X}_{1j} &= \dfrac{p}{\tau} \sum_{t \in B_j} X_{1t}
	\\ &= \dfrac{p}{\tau} \sum_{t \in B_j} (M_{1t} + \epsilon_{1t}) \nonumber
	\\ &= \dfrac{p}{\tau} \sum_{t \in B_j} M_{1t} + \dfrac{p}{\tau} \sum_{t \in B_j} \epsilon_{1t} \nonumber
	\\ &= \bar{M}_{1j} + \bar{\epsilon}_{1j},
\end{align}
whereby $\bar{M}_{1j} = \frac{p}{\tau} \sum_{t \in B_j} M_{1t}$ and $\bar{\epsilon}_{1j} = \frac{p}{\tau} \sum_{t \in B_j} \epsilon_{1t}$. Under these constructions, the noise entries remain zero-mean random variables for all $i,j$, i.e. $\mathbb{E} [\bar{\epsilon}_{ij}] = 0$. However, the variance of each noise term is now rescaled, i.e. for $i = 1$
\begin{align*}
	\text{Var}(\bar{\epsilon}_{1j}) &= \dfrac{p^2}{\tau^2} \sum_{t \in B_j} \text{Var}( \epsilon_{1t})
	\\ &\le \dfrac{\sigma^2}{\tau},
\end{align*}
and for $i > 1$, 
\begin{align*}
	\text{Var}(\bar{\epsilon}_{ij}) &= \dfrac{1}{\tau^2} \sum_{t \in B_j} \text{Var}( \epsilon_{it} \cdot D_{it})
	\\ &\stackrel{(a)} {\le} \dfrac{1}{\tau^2} \sum_{t \in B_j} (\sigma^2 p (1-p) + \sigma^2 p^2)
	\\ &\le \dfrac{\sigma^2}{\tau}.
\end{align*}
(a) used the fact that for any two independent random variables, $X$ and $Y$, $\text{Var}(XY) = \text{Var}(X) \text{Var}(Y) + \text{Var}(X) (\mathbb{E} [Y])^2 + \text{Var}(Y) (\mathbb{E}[X])^2$. Thus, for all $i, j$, $\text{Var}(\bar{\epsilon}_{ij}) \le \sigma^2 / \tau := \bar{\sigma}^2$. 

We now show that the key assumption of (\ref{eq:3}) still holds under this setting with respect to the newly defined variables. In particular, for every partition $j \in [\Delta]$ of row one,
	\begin{align*} 
		\bar{M}_{1j} &= \dfrac{p}{\tau} \sum_{t \in B_j} M_{1t}
		\\ &= \dfrac{p}{\tau} \sum_{t \in B_j} \Big( \sum_{k=2}^N \beta_k^* M_{kt} \Big)
		\\ &= \sum_{k=2}^{N}\beta_k^* \Big(\dfrac{p}{\tau}\sum_{t \in B_j} M_{kt}  \Big)
		\\ &= \sum_{k=2}^N \beta_k^* \bar{M}_{kj}.
	\end{align*}
As a result, we can express $\bar{M}_1^- = (\bbM^-)^T \beta^*$ for the same $\beta^*$ as in (\ref{eq:3}). 
	
Following a similar setup as before, we define the matrix $\bbY^- = [\bar{Y}_{ij}]_{2 \le i \le N, j \le \Delta}$. Since we have assumed that each block contains at least one observed entry, we subsequently have that $\bar{Y}_{ij} = \bar{X}_{ij}$ for all $i$ and $j$. We now proceed with our analysis in the exact same manner with the only difference being our newly defined set of variables and parameters. For completeness, we will highlight certain details below. 

To begin, observe that $\mathbb{E}[\bar{Y}_{ij}] = \bar{M}_{ij}$ while
\begin{align*}
	\text{Var}(\bar{Y}_{ij}) &\le \dfrac{\sigma^2p + p(1-p)}{\tau}.
\end{align*}
Consequently, we redefine the event $E_3 := \{ \norm{ \bbY^- -  \bbM^-} \le (2 + \omega) \sqrt{ \Delta \bar{q}} \}$ for some choice $\omega \in (0.1,1)$ and for $\bar{q} = \frac{\sigma^2p + p(1-p)}{\tau}$. By Theorem \ref{thm:talagrand}, it follows that $\Pb(E_3) \geq 1 - C' e^{-c \bar{q} \Delta}$. 

Similar to before, let $\delta$ be defined by the relation
\begin{align*}
	(1 + \delta) \norm{\bbY^- - \bbM^-} &= (2 + \omega) \sqrt{ \Delta \hat{\bar{q}}},
\end{align*}
where $\hat{\bar{q}} = \frac{\hat{\sigma}^2 \hat{p} + \hat{p} (1 - \hat{p})}{\tau}$. Letting $E = E_1 \cap E_2 \cap E_3$ and using arguments (\eqref{eq:thresh_1}, \eqref{eq:thresh_2}, \eqref{eq:reference_again}) that led us to \eqref{eq:thresh_3}, we obtain 
\begin{align*} 
	\mathbb{E} \norm{ (\hat{\bbM}^{-} - \bbM^-)^T \beta^*}^2 
	&\le \mathbb{E} \Big[ \norm{(\hat{\bbM}^{-} - \bbM^-)^T \beta^*}^2 \given E \Big] + \mathbb{E} \Big[\norm{(\hat{\bbM}^{-} - \bbM^-)^T \beta^*}^2 \given E^c \Big] \mathbb{P}(E^c) \nonumber
	\\ &\le C_1 \Delta \bar{q}  + C_2 e^{-c_3 \phi \Delta},
\end{align*}
where $\phi := p(N-1) + \sigma^2 + \bar{q}$. Utilizing Lemmas \ref{lemma:mse} and \ref{lemma:mse_ridge} gives us (for appropriately defined constants and defining $q = \sigma^2p + p(1-p)$ as before such that $\bar{q} = q / \tau$)
\begin{align*}
	\text{MSE}(\hat{\bar{M}}_1^-) &\le C_1 \bar{q}   + \dfrac{2 \bar{\sigma}^2 k}{\Delta} + \dfrac{\eta \norm{\beta^*}^2}{\Delta} + C_4 e^{-c_5 \phi \Delta}.
	\\ &= \dfrac{C_1 q}{\tau} + \dfrac{2 \sigma^2 k}{\tau \Delta} + \dfrac{\eta \norm{\beta^*}^2}{\Delta} + C_4 e^{-c_5 \frac{q}{\tau} \Delta}
	\\ &= \dfrac{C_1 q}{T_0^{1/2 - \gamma}} + \dfrac{2 \sigma^2 k}{T_0} + \dfrac{\eta \norm{\beta^*}^2}{T_0^{1/2 + \gamma}} + C_4 e^{- c_5 q T_0^{2 \gamma}}
	\\ &= \mathcal{O} (T_0^{-1/2 + \gamma}). 
\end{align*}
This concludes the proof.
\end{proof}

\section{Forecasting Analysis: Post-Intervention Regime}

We now bound the post-intervention $\ell_2$ error of our estimator. 

\subsection{Proof of Theorem \ref{thm:post}}

\begin{thm*}[\ref{thm:post}]
	Let \eqref{eq:3} hold for some $\beta^* \in {\mathbb R}^{N-1}$. Let $\rank(\bM^{-}) = \rank(\bM)$. Then $M_{1}^{+} = (\bM^{+})^T \beta^*$.
\end{thm*}

\begin{proof}
	Suppose we begin with only the matrix $\bM^-$, i.e. $\bM = \bM^-$. From the assumption that $M_1^- =( \bM^- )^T\beta^*$, we have for $t \le T_0$
	\begin{align*}
	M_{1t} &= \sum_{j=2}^N \beta_j^* M_{jt}.
	\end{align*}
	Suppose that we now add an extra column to $\bM^-$ so that $\bM$ is of dimension $N \times (T_0+1)$. Since $\rank(\bM^-) = \rank(\bM)$, we have for $j \in [N]$
	\begin{align*}
	M_{j,T_0+1} &= \sum_{t=1}^{T_0} \pi_t M_{jt},
	\end{align*}
	for some set of weights $\pi \in \mathbb{R}^{T_0}$. In particular, for the first row we have
	\begin{align*}
	M_{1, T_0+1} &= \sum_{t=1}^{T_0} \pi_t M_{1t}
	\\ &= \sum_{t=1}^{T_0} \pi_t \Big( \sum_{j=2}^N \beta_j^* M_{jt}  \Big)
	\\ &= \sum_{j=2}^N \beta_j^* \Big( \sum_{t=1}^{T_0} \pi_t M_{jt}  \Big)
	\\ &= \sum_{j=2}^N \beta_j^* M_{j, T_0+1}.
	\end{align*}
	By induction, we observe that for any number of columns added to $\bM^-$ such that $\rank(\bM^-) = \rank(\bM)$, we must have $M_{1}^{+} = (\bM^{+})^T \beta^*$ where $\bM^+ = [M_{it}]_{2 \le i \le N, T_0 < t \le T}$. 	
\end{proof}


\subsection{Proof of Theorem \ref{thm:post_rmse}}

\begin{thm*} [\ref{thm:post_rmse}]
Suppose $p \ge \frac{T^{-1 + \zeta}}{\sigma^2 + 1}$ for some $\zeta > 0$. Suppose $\norm{\hat{\beta}(\eta)}_{\infty} \le \psi$ for some $\psi > 0$. Let $\alpha' T_0 \le T \le \alpha T_0$ for some constants $\alpha', \alpha > 1$. Then for any $\eta \ge 0$ and using $\mu$ as defined in \eqref{eq:goldilocks}, the post-intervention error is bounded above by
	\begin{align*}
		\emph{RMSE}(\hat{M}_1^+)  &\le \dfrac{C_1}{\sqrt{p}} (\sigma^2 + (1-p))^{1/2} + \dfrac{C_2 \norm{\bM}}{\sqrt{T_0}} \cdot \E\norm{\hat{\beta}(\eta) - \beta^*} + \mathcal{O}(1/ \sqrt{T_0}),
	\end{align*}
	where $C_1$ and $C_2$ are universal positive constants.
\end{thm*}

\begin{proof}
	We will prove Theorem \ref{thm:post_rmse} by drawing upon techniques and results from prior proofs. We begin by applying triangle inequality to obtain
	\begin{align*} 
		\norm{\hat{M}_1^+ - M_1^+} &= \norm{ (\bhM^{+})^T \hat{\beta}(\eta) - (\bM^{+})^T \beta^*} 
		\\ &= \norm{ (\bhM^{+})^T \hat{\beta}(\eta) -(\bM^{+})^T \beta^* + (\bM^{+})^T \hat{\beta}(\eta) - (\bM^{+})^T \hat{\beta}(\eta)} 
		\\ &\le \norm{(\bhM^{+} - \bM^{+})^T \hat{\beta}(\eta)} + \norm{ (\bM^{+})^T( \hat{\beta}(\eta) - \beta^*)}.
	\end{align*}
	Taking expectations and using the property of induced norms gives
	\begin{align}  \label{eq:result_3}
		\E \norm{\hat{M}_1^+ - M_1^+} &\le \E \Big[ \norm{\bhM^{+} - \bM^{+}} \cdot \norm{\hat{\beta}(\eta)} \Big] + \norm{\bM^+} \cdot \E \norm{ \hat{\beta}(\eta) - \beta^*} \nonumber
		\\ &\le \sqrt{N} \psi \cdot \E \norm{\bhM^{+} - \bM^{+}} + \norm{\bM^+} \cdot \E \norm{ \hat{\beta}(\eta) - \beta^*},
	\end{align}
	where the last inequality uses the boundedness assumption of $\hat{\beta}(\eta)$. Observe that the first term on the right-hand side of \eqref{eq:result_3} is similar to that of \eqref{eq:x.20} and \eqref{eq:x.21} with the main difference being \eqref{eq:result_3} uses the post-intervention submatrices, $\bhM^+$ and $\bM^+,$ as opposed to the pre-intervention submatrices, $\bhM^-$ and $\bM^-,$ in \eqref{eq:x.20} and \eqref{eq:x.21}. Therefore, using \eqref{eq:trick} and the arguments that led to \eqref{eq:thresh_3}, it follows that (with appropriate constants $C_1, C_2, c_3$)
	\begin{align*}
		\mathbb{E} \norm{\bhM^{+} - \bM^{+}} 
		&\le  \dfrac{C_1}{p} \mathbb{E} \Big( \sqrt{T q} + \norm{ (\hat{p} - p) \bM^+ } \Big)  + C_2 ((N-1)T)^{3/2} e^{-c_3 \phi T},
	\end{align*}
	where the slight modification arises due to the fact that we are now operating in the post-intervention regime. In particular, $\norm{\bM^+} \le \sqrt{(N-1) (T - T_0)}$ and $\norm{\bhM^+} \le ((N-1)T)^{3/2}$. Further, note that $q$ and $\phi$ are defined exactly as before, i.e. $q = \sigma^2p + p(1-p)$ and $\phi = p(N-1) + \sigma^2 + q$. Following the proof of Corollary \ref{corollary:universal_threshold}, we apply Jensen's Inequality to obtain
	\begin{align*}
		\mathbb{E} \norm{(\hat{p} - p) \bM^+} &= \mathbb{E} \abs{ \hat{p} - p} \cdot \norm{\bM^+} 
		\\ &\le  \sqrt{ \dfrac{p (1-p)}{(N-1) T}} \cdot \sqrt{(N-1)(T - T_0)}
		\\ &\le \sqrt{p (1-p)}.
	\end{align*}
	Putting the above results together, we have (for appropriately defined constants)
	\begin{align*}
		\text{RMSE}(\hat{M}_1^+) &\le  \dfrac{C_1 \sqrt{N} \psi}{p \sqrt{T - T_0}} \Big( \sqrt{Tq} + \sqrt{p(1-p)} \Big) + \dfrac{\norm{\bM^+}}{\sqrt{T-T_0}} \cdot \E \norm{\hat{\beta}(\eta) - \beta^*} + C_4 e^{-c_5 \phi T} 
		\\ &\stackrel{(a)}{\le} \dfrac{C_6 \sqrt{q}}{p} + \dfrac{C_7 \sqrt{1-p}}{\sqrt{p T_0}} + \dfrac{C_8 \norm{\bM}}{\sqrt{T_0}} \cdot \E\norm{\hat{\beta}(\eta) - \beta^*} + C_4 e^{-c_5 \phi T} 
		\\ &=  \dfrac{C_6}{\sqrt{p}} (\sigma^2 + (1-p))^{1/2} + \dfrac{C_8 \norm{\bM}}{\sqrt{T_0}} \cdot \E\norm{\hat{\beta}(\eta) - \beta^*} + \mathcal{O}(1/ \sqrt{T_0}),
	\end{align*}
where (a) follows from Lemma \ref{lemma:spectral_norm}. Renaming constants would provide the desired result.  
\end{proof}

\section{A Bayesian Perspective}

\noindent \textbf{Derivation of posterior parameters.}

The following is based on the derivation presented in Section 2.2.3 of \cite{bishop}, and is presented here for completeness. Suppose we are given a multivariate Gaussian marginal distribution $p(x)$ paired with a multivariate Gaussian conditional distribution $p(y \given x)$ -- where $x$ and $y$ may have differing dimensions -- and we are interested in computing the posterior distribution over $x$, i.e. $p(x \given y)$. We will derive the posterior parameters of $p(x \given y)$ here. Without loss of generality, suppose
\begin{align*}
	p(x) &= \mathcal{N} (x \given \mu, \bLambda^{-1})
	\\ p(y \given x) &= \mathcal{N} (y \given \bA x + b, \bSigma^{-1}),
\end{align*}
where $\mu, \bA,$ and $b$ are parameters that govern the means, while $\bLambda$ and $\bSigma$ are precision (inverse covariance) matrices. 

We begin by finding the joint distribution over $x$ and $y$. Ignoring the terms that are independent of $x$ and $y$ and encapsulating them into the ``const.'' expression, we obtain
\begin{align*}
	\ln{p(x,y)} &= \ln{p(x)} + \ln{p(y \given x)}
	\\ &= -\dfrac{1}{2}(x - \mu)^T \bLambda (x - \mu) - \dfrac{1}{2}(y- \bA x - b)^T \bSigma (y - \bA x - b) + \text{const.}
	\\ &= -\dfrac{1}{2}x^T (\bLambda + \bA^T \bSigma \bA)x - \dfrac{1}{2} y^T \bSigma y + \dfrac{1}{2} x^T \bA^T \bSigma y + \text{const.}
	\\ &= -\dfrac{1}{2} \begin{bmatrix} x \\    y  \end{bmatrix}^T \begin{bmatrix} \bLambda + \bA^T \bSigma \bA & -\bA^T \bSigma \\ -\bSigma \bA & \bSigma \end{bmatrix}  \begin{bmatrix} x \\    y  \end{bmatrix} + \text{const.}
	\\ &= -\dfrac{1}{2} z^T \bQ z + \text{const.},
\end{align*}
where $z = [x, y ]^T$, and 
\begin{align*}
	\bQ &= \begin{bmatrix} \bLambda + \bA^T \bSigma \bA & -\bA^T \bSigma \\ -\bSigma \bA & \bSigma \end{bmatrix}
\end{align*}
is the precision matrix. Applying the matrix inversion formula, we have that the covariance matrix of $z$ is
\begin{align*}
	\text{Var}(z) &= \bQ^{-1} = \begin{bmatrix} \bLambda^{-1} & \bLambda^{-1} \bA^T \\ \bA \bLambda^{-1} & \bSigma^{-1} + \bA \bLambda^{-1} \bA^T \end{bmatrix}.
\end{align*}
After collecting the linear terms over $z$, we find that the mean of the Gaussian distribution over $z$ is defined as
\begin{align*}
	\mathbb{E}[z] &= \bQ^{-1} \begin{bmatrix} \bLambda \mu - \bA^T \bSigma b \\ \bSigma b \end{bmatrix}. 
\end{align*}
Now that we have the parameters over the joint distribution of $x$ and $y$, we find that the posterior distribution parameters over $x$ are
\begin{align*}
	\mathbb{E}[x \given y] &= (\bLambda + \bA^T \bSigma \bA)^{-1} \{ \bA^T \bSigma (y - b) + \bLambda \mu \}
	\\ \text{Var}(x \given y) &= (\bLambda + \bA^T \bSigma \bA)^{-1}.
\end{align*}

\end{document}